\documentclass[runningheads]{llncs}

\usepackage[utf8]{inputenc}
\usepackage[english]{babel}
\usepackage{graphicx,amssymb,amsmath}
\usepackage{algorithm,algorithmic}

\newcommand{\C}{\mathbb{C}}
\newcommand{\DMIN}{{\delta_\text{\rm min}}}
\newcommand{\DMAX}{{\delta_\text{\rm max}}}
\newcommand{\E}{{e_\text{\rm max}}}
\newcommand{\F}{\mathbb{F}}
\newcommand{\FL}{{\F_{L}}}
\newcommand{\FLK}{{\F_{L,K}}}
\newcommand{\G}{\mathbb{G}}
\newcommand{\GG}{{\cal G}}

\newcommand{\GLKE}{{\G_{L,K,\E}}}
\newcommand{\GO}{{\cal O}}
\newcommand{\N}{\mathbb{N}}
\newcommand{\Q}{\mathbb{Q}}
\newcommand{\R}{\mathbb{R}}
\newcommand{\Z}{\mathbb{Z}}
\newcommand{\LG}{{L_\text{\rm grid}}}
\newcommand{\LS}{{L_\text{\rm safe}}}
\newcommand{\LOG}{\log_2}
\newcommand{\NE}{{N_\text{\rm E}}}
\newcommand{\NP}{{N_\text{\rm P}}}
\newcommand{\NT}{{N_\text{\rm T}}}
\newcommand{\pgrid}{{p_\text{\rm grid}}}
\newcommand{\pinf}{{p_\text{\rm inf}}}
\newcommand{\psup}{{p_\text{\rm sup}}}
\newcommand{\varphiinf}{\varphi_{\inf f}}
\newcommand{\varphisup}{\varphi_{\sup f}}
\newcommand{\AUG}{{\text{\rm aug}}}
\newcommand{\RAUG}{R_\AUG}%
\newcommand{\RF}{\!|_\F} %
\newcommand{\RFL}{\!|_{\F_L}}
\newcommand{\RFLK}{\!|_\FLK}
\newcommand{\RG}{\!|_\G} %
\newcommand{\RGL}{\!|_{\G_L}}
\newcommand{\RGLK}{\!|_{\G_{L,K}}}
\newcommand{\RGLKE}{\!|_{\G_{L,K,\E}}}
\newcommand{\RQ}{\!|_\Q} %
\newcommand{\Rplus}{\R_{>0}}

\newcommand{\Sinf}{S_{\inf f}}
\newcommand{\Ssup}{S_{\sup f}}
\newcommand{\IND}{\text{\rm ind}}
\newcommand{\SUP}{\text{\rm sup}}
\newcommand{\RE}{\text{\rm Re}}
\newcommand{\Lex}{{\prec}}%
\newcommand{\Perm}{{\cal P}}%
\newcommand{\Ind}{{\cal I}}%
\newcommand{\IM}{{\Ind_{\rm max}}}%
\newcommand{\CM}{{c_{\text{\rm m}}}}
\newcommand{\CU}{{c_{\text{\rm u}}}}
\newcommand{\DOM}{{\rm dom}}
\newcommand{\FSI}{{f^{*i}}}%
\newcommand{\FSIX}{{f^{*i}_\xi}}%

\newcommand{\GAB}{{\Gamma\text{\rm-box}}}
\newcommand{\GABG}{{\Gamma\text{\rm-box\,}_{\hat\gamma}}}
\newcommand{\GAL}{{\Gamma\text{\rm-line}}}
\newcommand{\GALG}{{\Gamma\text{\rm-line\,}_{\hat\gamma}}}

\newcommand{\GARG}{{\Gamma\text{\rm-region\,}_{\gamma}}}

\newcommand{\GASG}{{\Gamma\text{\rm-safe\,}_{\gamma}}}

\newcommand{\PIG}{{\pi_{>i}}}
\newcommand{\PII}{{\pi_{i}}}
\newcommand{\PIL}{{\pi_{<i}}}
\newcommand{\PIN}{{\pi_{\neq i}}}
\newcommand{\REP}{{\text{\rm rep}}}
\newcommand{\X}{{X}}
\newcommand{\Alg}{{\cal A}}
\newcommand{\AG}{{\cal A}_\text{\rm G}}
\newcommand{\ACP}{{\cal A}_\text{\rm CP}}
\newcommand{\BACP}{\text{\rm basic-}{\cal A}_\text{\rm CP}}

\newcommand{\inbox}{\text{\rm in\_box}}
\newcommand{\incirc}{\text{\rm in\_circle}}

\newcommand{\fmin}{f_\text{\rm min}}
\newcommand{\fmax}{f_\text{\rm max}}
\newcommand{\sign}{\text{\rm sign}}
\newcommand{\PR}{\text{\rm pr}}
\newcommand{\U}{\bar{U}}
\newcommand{\UU}{{\cal U}}
\newcommand{\UUU}{{\overline{\cal U}}}
\newcommand{\PRED}{{(f, k, A, \delta, \E)}}
\newcommand{\PREDG}{{(f, k, A, \delta, \E, \Gamma)}}
\newcommand{\PREDGt}{{(f, k, A, \delta, \E, \Gamma, t)}}
\newcommand{\PREDB}{{(f, k, A, \delta, \E, \GAB, t)}}
\newcommand{\PREDL}{{(f, k, A, \delta, \E, \GAL, t)}}
\newcommand{\x}{\bar{x}}
\newcommand{\y}{\bar{y}}
\newcommand{\ST}{\,:\,}%
\newcommand{\CARDI}[1]{{\left|\,#1\,\right|}}
\newcommand{\FORMSEP}{\quad}
\newcommand{\tabrule}{\rule[-2.0ex]{0ex}{5ex}}%

\newcommand{\keywords}[1]{\par\addvspace\baselineskip%
  \noindent\keywordname\enspace\ignorespaces#1}

\begin{document}

\mainmatter
\title{General Analysis Tool Box for\\Controlled Perturbation}
\titlerunning{General Analysis Tool Box for Controlled Perturbation}
\author{Ralf Osbild}
\institute{Saarbr{\"u}cken, March 29, 2012}
\maketitle

\begin{abstract}
  The implementation of reliable and efficient geometric algorithms
  is a challenging task.
  The reason is the following conflict:
  On the one hand,
  computing with rounded arithmetic
  may question the reliability of programs
  while, on the other hand,
  computing with exact arithmetic may be too expensive and hence inefficient.
  One solution is the implementation of
  {controlled perturbation algorithms} which
  combine the speed of floating-point arithmetic with
  a protection mechanism that guarantees reliability,
  nonetheless.

  This paper is concerned with
  the performance analysis of controlled perturbation algorithms
  in theory.
  We answer this question with the presentation of a
    \emph{general analysis tool box} for controlled perturbation algorithms.
  This tool box is separated into independent components
  which are presented individually with their interfaces.
  This way, the tool box supports alternative approaches
  for the derivation of the most crucial bounds.
  We present three approaches for this task.
  Furthermore,
  we have thoroughly reworked the concept of controlled perturbation
  in order to include rational function based predicates into the theory;
  polynomial based predicates are included anyway.
  Even more we introduce object-preserving perturbations.
  Moreover,
  the tool box is designed such that it reflects the actual behavior
  of the controlled perturbation algorithm at hand
  without any simplifying assumptions.
\end{abstract}

\section{Introduction}

\subsection{Robust Geometric Computing}

  It is a notoriously difficult task to cope with
  rounding errors in computing~\cite{F70,DH02}.
  In computational geometry,
  predicates are decided on the sign of mathematical expressions.
  If rounding errors cause a wrong decision of the predicate,
  geometric algorithms may fail in various ways:
  inconsistency of the data (e.g., contradictory topology),
  loops that do not terminate or
  loops that terminate unexpectedly~\cite{KMPSY08}.
  In addition,
  the thoughtful processing of degenerate cases
  makes the implementation of geometric algorithms laborious~\cite{BKOS00}.
  The meaning of degeneracy always depends on the context
  (e.g., three points on a line, four points on a circle).
  There are several ways to overcome the numerical robustness issues
  and to deal with degenerate inputs.

  The \emph{exact computation paradigm}~\cite{JRZ91,KLN91,MN94,FW96,Y97,LEDA99}
  suggests an implementation of an exact arithmetic.
  This is established by a number representation of variable precision
  (i.e., variable bit length)
  or the use of symbolic values which are not evaluated
  (e.g., roots of integers).
  There are several implementations of such
  number types~\cite{CGAL,CORE02,GMPFR11,MPFI,LEDA99}.
  Each program must be developed carefully such that
  it can deal with all possible degenerate cases.
  The software libraries {\sc Leda} and {\sc Cgal}
  follow the exact computation paradigm~\cite{LEDA99,KN04,FGK00}.
  The paradigm was also taken as a basis in~\cite{BEH05,Sh97,HK05}.

  As opposed to that,
  the \emph{topology oriented approach}~\cite{SI92,Im96,SII00}
  is based on an arithmetic of finite precision.
  To avoid numerical robustness issues,
  the main guideline is the maintenance of the topology.
  This objective requires individual alterations of the algorithm at hand
  and it seems that it cannot be turned into an easy-to-use general framework.
  Furthermore,
  this approach must also cope with degenerate inputs.
  However,
  the speed of floating-point arithmetic may be worth the trouble;
  in addition with other accelerations,
  Held~\cite{He01} has implemented a very fast computation
  of the Voronoi diagram of line segments.

  There are also \emph{problem-oriented} solutions.
  In computational geometry,
  the sign of determinants decides an interesting class of predicates.
  For example,
  the side-of-line or the in-circle predicate in the plane
  belong to this class
  and are used in the computation of Delaunay diagrams.
  Some publications attack the numerical issues in
  the evaluation of determinants directly~\cite{ABD97,BM00}.

  The previous approaches have in common
  that they primarily focus on the numerical issues.
  Other approaches are originated from the degeneracy issue.
  A slight perturbation of the input seems to solve this problem.
  There are different approaches which are based on perturbation.
  The \emph{symbolic perturbation},
  see for example \cite{EM90,Y90a,Y90b,EC95,ECS97,S98,M95},
  provides a general way to distort inputs such that
  degeneracies do not occur.
  This definitely provides a shorter route
  for the presentation of geometric algorithms.
  Practically this approach requires exact arithmetic
  to avoid robustness issues.
  Therefore the pitfall in this approach is that,
  if the concept requires very small perturbations,
  it implicates a high precision and possibly a slow implementation.
  
  In this paper we focus on \emph{controlled perturbation}.
  This variant was introduced
  by Halperin et al.~\cite{HS98}
  for the computation of spherical arrangements.
  There a perturbed input is a random point
  in the neighborhood of the initial input.
  It is unlikely, but not impossible, that the input is degenerate.
  Therefore the algorithm has a repeating perturbation process
  with two objectives:
  Finding an input
  that does not contain degeneracies and
  that leads to numerically robust floating-point evaluations.
  Halperin et al.~have presented
  mechanisms to respond to inappropriate perturbations.
  Moreover,
  they have argued formally
  under which conditions there is a chance for a successful termination
  of their algorithm.
  \emph{Controlled perturbation leads to numerically robust implementations
  of algorithms
  which use non-exact arithmetic and
  which do not need to process degenerate cases.}

  This idea of controlled perturbation
  was applied to further geometric problems afterwards:
  The arrangement of polyhedral surfaces~\cite{HR99},
  the arrangement of circles~~\cite{HL04},
  Voronoi diagrams and Delaunay triangulations~\cite{K04,FKMS05}.
  However, the presentation of each specific algorithm
  has required a specific analysis of its performance.
  This broaches the subject of
  a \emph{general method} to analyze controlled perturbation algorithms.

  We remark that
  controlled perturbation has also a shady side:
  Although it solves the problem for the perturbed input exactly,
  it does not solve it for the initial input.
  Furthermore,
  it is non-obvious how to receive a solution for the initial input in general.
  In case the input is highly degenerated,
  the running time of the algorithm may increase significantly
  after the permutation~\cite{BMS94,ABS97}.
  In this case,
  the specialized treatment of degeneracies may be much faster.

\subsection{Our contribution}

  The study of a {general method}
  to analyze controlled perturbation algorithms
  is a joint work with Kurt Mehlhorn and Michael Sagraloff.
  We have firstly presented the idea
  in \cite{MOS06}.
  Then Caroli~\cite{C07} studied the applicability of the method
  for predicates which are used for
  the computation of arrangements of circles (according to~\cite{HL04})
  and the computation of Voronoi diagrams of line segments
  (according to~\cite{Bu96,Se96}).
  Our significantly improved journal article
  contains, furthermore, a detailed discussion of
  the analysis of multivariate polynomials~\cite{MOS11}.

  Independent of former publications,
  the author has redeveloped the topic from scratch
  to design a sophisticated tool box for the analysis
  of controlled perturbation algorithms.
  The tool box is valid for floating-point arithmetic,
  guides step by step through the analysis and
  allows alternative components.
  Furthermore,
  the solutions of two open problems
  are integrated into the theory.
  We briefly present our achievements below.

  \emph{We present a general tool box to
  analyze algorithms and their predicates.}
  The tool box is subdivided into independent components and their interfaces.
  Step-by-step instructions for the analysis
  are associated with each component.
  Interfaces represent bounds that are used in the analysis.
  The result is a precision function or a probability function.
  Furthermore,
  necessary conditions for the analysis
  are derived from the interfaces
  (e.g., the notion of \emph{criticality} differs from former publications).

  \emph{We present alternative approaches to derive necessary bounds.}
  Because we have subdivided the tool box into
  independent components and their interfaces,
  it is possible to make alternative components available
  in the most crucial step of the analysis.
  The \emph{direct approach} is based on the geometric meaning of predicates,
  the \emph{bottom-up approach} is based on the composition of functions,
  and the \emph{top-down approach} is a coordinate-wise analysis of functions.
  Similar direct and top-down approaches are presented in~\cite{MOS06,MOS11}.
  This is the first time that a bottom-up approach is presented
  for this task.

  \emph{The result of the analysis is valid
  for floating-point arithmetic.}
  A random floating-point number generator
  that guarantees a uniform distribution
  was introduced in~\cite{MOS11}.
  But, so far, the result of the analysis was never proven to be valid
  for the finite set of floating-point numbers
  since the Lebesgue measure cannot take sets of measure zero
  into account.
  To overcome this issue,
  we define a specialized perturbation generator and
  pay attention to the finiteness in the analysis, namely,
  in the success probability,
  in the (non-)exclusion of points and
  in the usage of the Lebesgue measure.

  \emph{We present an alternative analysis of multivariate polynomials.}
  An analysis of multivariate polynomials,
  which resembles the top-down approach,
  is presented in~\cite{MOS11}.
  Here we present an alternative analysis
  which makes use of the bottom-up approach.

  \emph{We solve the open problem of analyzing rational functions.}
  We include poles of rational functions into the theory and
  describe the treatment of floating-point range errors in the analysis.
  We suggest a general way to guard rational functions in practice and
  we show how to analyze the behavior of these guards in theory.

  \emph{We solve the open problem of object-preserving perturbations.}
  We introduce a perturbation generator that makes it possible
  to perturb the location of input objects without deforming the objects itself.
  To achieve this goal,
  we have designed the perturbation
  such that
  the relative floating-point input specifications of the objects
  are preserved despite of the usage of rounded arithmetic.

  \emph{We suggest an implementation
  that is in accordance with the analysis tool box.}
  We define a fixed-precision perturbation generator and
  extend it to be object-preserving.
  We explain the particularities
  in the practical treatment of range errors
  that occur especially in the case of rational functions.
  Finally, we show how to realize guards for rational functions.

\subsection{Content}

  In this paper we present a tool box
  for a general analysis of controlled perturbation algorithms.
  In Section~\ref{sec-cp-algo},
  we present the basic design principles of controlled perturbation
  from a practical point of view.
  Fundamental quantities and definitions of the analysis
  are introduced in Section~\ref{sec-fund-quant-def}.
  The \emph{general analysis tool box}
  and all of its components
  are briefly introduced
  in Section~\ref{sec-ana-tool-box}.
  Its detailed presentation is structured in two parts:
  The \emph{function analysis} and the \emph{algorithm analysis.}

  Geometric algorithms base their decisions on geometric predicates
  which are decided by signs of real-valued functions.
  Therefore the analysis of algorithms
  requires a general analysis of such functions.
  The \emph{function analysis}
  is visualized
  in Figure~\ref{fig-illu-ana-func} on Page~\pageref{fig-illu-ana-func}.
  Since the analysis is performed with real arithmetic,
  we must also prove its validation for actual floating-point inputs.
  This validation is anchored in Section~\ref{sec-validation}.
  The function analysis itself works in two stages.
  The required bounds form the interface between the stages
  and are presented
  in Section~\ref{sec-nec-con-func}.
  The \emph{method of quantified relations}
  represents the actual analysis in the second stage and is introduced
  in Section~\ref{sec-meth-quan-rela}.
  The derivation of the bounds in the first stage uses
  the \emph{direct approach} of Section~\ref{sec-direct-approach},
  the \emph{bottom-up approach}
  of Section~\ref{sec-bottom-up}, or
  the \emph{top-down approach}
  of Section~\ref{sec-top-down},
  together with an \emph{error analysis}
  which is introduced in Section~\ref{sec-guards-safetybounds}.
  In Section~\ref{sec-rational-function}
  we extend the analysis and the implementation
  such that both properly deal with floating-point range errors.
  As examples,
  we present the analysis of
  \emph{multivariate polynomials} in Section~\ref{sec-bottom-up}
  and the analysis of
  \emph{rational functions} in Section~\ref{sec-ana-rational-func}.

  The \emph{algorithm analysis}
  is visualized
  in Figure~\ref{fig-illu-ana-algo} on Page~\pageref{fig-illu-ana-algo}.
  The algorithm analysis works also in two stages.
  In the first stage, we perform the function analyses and
  derive some algorithm specific bounds.
  The analysis itself in the second stage is represented by the
  \emph{method of distributed probability}.
  The algorithm analysis is entirely presented in Section~\ref{sec-ana-algo}.

  Furthermore, we present a general way
  to \emph{implement} controlled perturbation algorithms
  in Section~\ref{sec-gen-cp-imple}
  such that
  our analysis tool box can be applied to them.
  Even more,
  we suggest a way to implement \emph{object-preserving perturbations}
  in Section~\ref{sec-pertub-policy}.

  A \emph{quick reference}
  to the most important definitions of this paper
  can be found in the appendix in Section~\ref{sec-append-identifiers}.

\section{Controlled Perturbation Algorithms}
\label{sec-cp-algo}

  This section contains an introduction
  to the basic principles for controlled perturbation algorithms.
  We have already mentioned
  that implementations of geometric algorithms
  must address degeneracy issues and numerical robustness issues.
  We review floating-point arithmetic
  in Section~\ref{sec-intro-floating-point}
  and present the basic design principles
  of controlled perturbation algorithms
  in Section~\ref{sec-intro-cp-algo}.

\subsection{Floating-point Arithmetic}
\label{sec-intro-floating-point}

  Variable precision arithmetic is necessary
  for a general implementation of controlled perturbation algorithms.
  We explain this statement with the following thought experiment\footnote{%
    This consideration is absolutely conform to
    Halperin et~al.~\cite{HL04}:
    If the augmented perturbation parameter $\delta$
    exceeds a given threshold $\Delta$,
    the precision is augmented and $\delta$ is reset.}
  that can be skipped during first reading:
  Assume we compute an arrangement of $n$ circles incrementally
  with a fixed precision arithmetic.
  Let us further assume that there is an upper bound
  on the radius of the circles.
  Then, because of the fixed precision,
  the number of distinguishable intersections per circle must be limited.
  Hence the computation of a dense arrangement
  gets stuck after a certain amount of insertions
  unless we allow circles to be moved (perturbed) further away
  from their initial location.
  Asymptotically,
  this policy transforms a very dense arrangement
  into an arrangement of almost uniformly distributed circles.
  Therefore we demand that the precision of the arithmetic can be chosen
  arbitrarily large.
  
  A \emph{floating-point number} is given by
  a sign, a mantissa, a radix and a signed exponent.
  In the regular case,
  its value is defined as
  \begin{eqnarray*}
    \text{value} 
      &:=&
    \text{sign} \cdot \text{mantissa} \cdot \text{radix}^\text{exponent}.
  \end{eqnarray*}
  Without loss of generality,
  we assume the radix to be 2.
  The bit length
  of the mantissa is called \emph{precision}
  $L\label{def-L-inline}$.
  We denote the \emph{bit length of the exponent} by
  $K\label{def-K-inline}$.
  The discrete set of regular floating-point numbers is a subset of
  the rational numbers.
  Furthermore,
  this set is finite
  for fixed $L$ and $K$.

  A \emph{floating-point arithmetic}
  defines the number representation (the radix, $L$ and $K$),
  the operations,
  the rounding policy
  and the exception handling for floating-point numbers
  (see Goldberg~\cite{G91}).
  A technical standard for
  fixed precision floating-point arithmetic
  is IEEE 754-2008 (see \cite{IEEE08}).
  Nowadays,
  the built-in types single, double and quadruple precision
  are usual for radix 2.

  There are several software libraries that offer
  \emph{variable\footnote{%
    With variable we subsume all types of arithmetic
    that support arbitrarily large precisions.
    Some are called variable precision, multiple precision
    or arbitrary precision.}
  precision floating-point arithmetic}.
  {\sc Cgal} provides the multi-precision floating-point number type
  {\tt MP\_Float} (see the {\sc Cgal} manual~\cite{CGAL}).
  {\sc Core} provides
  the variable precision floating-point number type
  {\tt CORE::BigFloat} (see~\cite{CORE02}).
  And
  {\sc Leda} provides
  the variable precision floating-point number type
  {\tt leda\_bigfloat} (see the {\sc Leda} book~\cite{LEDA99}).
  Be aware that
  the rounding policy and exception handling of certain libraries
  may differ from the IEEE standard.
  Since our analysis partially presumes\footnote{%
    A standardized behavior of floating-point operations
    is presumed in Section~\ref{sec-guards-safetybounds}.}
  this standard,
  we must ensure that the arithmetic in use is appropriate.
  The {\sc Gnu} Multiple Precision Floating-Point Reliable Library,
  for example,
  ``provides the four rounding modes from the IEEE 754-1985 standard,
  plus away-from-zero, as well as for basic operations as for other
  mathematical functions'' (see the {\sc Gnu Mpfr} manual~\cite{GMPFR11}).
  Moreover,
  {\sc Gnu Mpfr} is used for
  the multiple precision interval arithmetic
  which is provided by the
  Multiple Precision Floating-point Interval library
  (see the {\sc Gnu Mpfi} manual~\cite{MPFI}).
  
  Variable precision arithmetic
  is more expensive than built-in fixed precision arithmetic.
  We remark that, in practice, we try to solve the problem at hand
  with built-in arithmetic first
  and, in addition, try to make use of floating-point filters.
  Throughout the paper we use the following notations.

\begin{definition}[floating-point]\label{def-fp-numbers}
  Let $L,K\in\N$.
  By $\FLK$ we denote: \smallskip \\
\indent
    1. The set of floating-point numbers
    with radix 2, precision $L$
    and $K$-bit exponent. \\
\indent
    2. The floating-point arithmetic
    that is induced by the set characterized in 1. \smallskip \\
Furthermore, we define the suffix $\RF$ for sets and expressions: \smallskip \\
\indent
    1. Let $k\in\N$ and let $X\subset\R^k$.
    Then $X\RF := X \cap \F^k$. \\
\indent
    2. $f(x)\RF$
    denotes the floating-point value of $f(x)$
    evaluated with arithmetic $\F$.
\end{definition}
\noindent
  That means,
  by $X\RF$ we denote the restriction of $X$ to
  its subset that can be represented with floating-point numbers in $\F$.
  To simplify the notation
  we omit the indices $L$ or $K$ of $\FLK$
  whenever they are given by the context.
  For the same reason
  we have already skipped the dimension $k$ in the suffix $\RF$.

\subsection{Basic Controlled Perturbation Implementations}
\label{sec-intro-cp-algo}

  Rounding errors of floating-point arithmetic
  may influence the result of predicate evaluations.
  Wrong predicate evaluations may cause
  erroneous results of the algorithm
  and even lead to non-robust implementations
  (see Kettner et al.~\cite{KMPSY08}).
  In order to get correct and robust implementations,
  we introduce guards which testify the reliability of predicate evaluations
  (see~\cite{F97,BFS01,MOS11}).
\begin{definition}[guard]\label{def-guard}
  Let $\F$ be a floating-point arithmetic
  and let $f:X\to\R$ be a function
  with $X\subset\R^k$.
  We call a predicate $\GG_{f}:X\to\text{\{true, false\}}$
  a \emph{guard for $f\!$ on $X$} if
  \begin{eqnarray*}
    \text{$\GG_{f}(x)$ is true}
      \FORMSEP &\Rightarrow& \FORMSEP
    \sign(f(x)\RF) = \sign(f(x))
  \end{eqnarray*}
  for all $x\in X\RF$.
  Presumed that there is such a predicate $\GG_f$,
  we say that an input $x\in X\RF$ is \emph{guarded}
  if $\GG_{f}(x)$ is true
  and \emph{unguarded}
  if $\GG_{f}(x)$ is false.
\end{definition}
  That means, guards testify the sign of function evaluations.
  A design of guards is presented in Section~\ref{sec-guards-safetybounds}.
  By means of guards
  we can implement geometric algorithms
  such that they can either verify or disprove their result.
\begin{definition}[guarded algorithm]\label{def-guarded-algo}
  We call an algorithm $\AG$ a \emph{guarded algorithm}
  if there is a guard for each predicate evaluation
  and if the algorithm halts either with the correct combinatorial result
  or with the information that a guard has failed.
  If $\AG$ halts with the correct result,
  we also say that $\AG$ is \emph{successful}, and
  we say that $\AG$ has \emph{failed}
  if a guard has failed.
\end{definition}
  Let $\y$ be an input of $\AG$.
  In case $\AG(\y)$ is successful,
  we obtain the desired result
  for input $\y$.
  Of course,
  the situation is unsatisfying
  if $\AG$ fails.
  Therefore we introduce controlled perturbation
  (see Halperin et al.~\cite{HL04}):
  We execute $\AG$ for randomly perturbed inputs $y$
  (i.e., random points in the neighborhood of $\y$)
  \emph{until} $\AG$ terminates successfully.
  Furthermore,
  we increase the precision $L$
  of the floating-point arithmetic $\F$
  after each failure
  in the hope to improve the chance to succeed.
  (It is the task of the analysis to give evidence.)
  We summarize this idea
  in the provisional controlled perturbation algorithm \mbox{$\BACP$}
  which is shown in Algorithm~\ref{algo-plaincp}.
  The general controlled perturbation algorithm
  is presented on page~\pageref{algo-cp} in Section~\ref{sec-ana-algo}.
  \begin{algorithm}
    \caption{: $\BACP(\AG, \y, \UU_\delta)$}
    \label{algo-plaincp}
    \begin{algorithmic}
      \STATE \emph{/* initialization */}
      \STATE $L \leftarrow$ precision of built-in floating-point arithmetic
     \medskip
     \REPEAT
      \STATE \emph{/* run guarded algorithm */}
      \STATE $y \leftarrow$ random point in $\UUU_\delta(\y)\RFL$
      \STATE $\omega \leftarrow \AG(y,\FL)$
      \medskip
      \STATE \emph{/* adjust parameters */}
      \IF{$\AG$ failed}
        \STATE $L \leftarrow 2L$
      \ENDIF
     \UNTIL{$\AG$ succeeded}
     \medskip
      \STATE \emph{/* return perturbed input $y$ and result $\omega$ */}
      \RETURN $(y,\omega)$
    \end{algorithmic}
  \end{algorithm}

  We see that there is an implementation of $\BACP(\AG)$
  for every guarded algorithm $\AG$,
  or to say it in other words,
  for every algorithm that is only based on geometric predicates
  that can be guarded.
  It is important to note that
  this does not necessarily imply that $\BACP$ performs well.
  It is the main objective of this paper
  to develop a general method to analyze the performance
  of controlled perturbation algorithms $\ACP$.

\section{Fundamental Quantities and Definitions}
\label{sec-fund-quant-def}

  Our main aim is the derivation of a general method
  to analyze controlled perturbation algorithms.
  In order to achieve this,
  we introduce fundamental quantities first.
  In Section~\ref{sec-basic-quantities}
  we define the quantities that describe the situation
  which we want to analyze.
  We encounter and discuss many issues
  during the definition of the success probability
  in Section~\ref{sec-succ-prob}.
  \emph{This is the first presentation of a
  detailed modelling of the floating-point success probability.}
  Controlled perturbation specific quantities are introduced
  in Section~\ref{sec-further-quantities}.
  (Further analysis specific bounds
  are defined in the presentation of the analysis later on.)
  The overview in Section~\ref{sec-over-func-argu}
  summarizes the classification of inputs
  in practice and in the analysis.
  In Section~\ref{sec-veri-succ}
  we present conditions
  under which we may \emph{apply} controlled perturbation to a predicate
  in practice
  and under which we can actually \emph{justify} its application
  in theory.

\subsection{Perturbation, Predicate, Function}
\label{sec-basic-quantities}

  Here we define the quantities
  that are needed to describe the initial situation:
   the original input,
   the perturbation area,
   the perturbation parameter,
   the perturbed input,
   the input value bound,
   functions that realize geometric predicates,
   and
   predicate descriptions.

  In the analysis
  we assume that the \emph{original input $\y$\label{def-ybar-inline}}
  of a controlled-perturbation algorithm $\ACP$
  consists of $n\label{def-n-inline}$
  floating-point numbers,
  that means, $\y\in\F^n$ or, as we prefer to say, $\y\in\R^n\RF$.
  At this point
  we do not care for a geometrical interpretation
  of the input of $\ACP$.
  We remark that this is no restriction:
  a complex number can be represented by two numbers;
  a vector can be represented by the sequence of its components;
  geometric objects can be represented by their coordinates and measures;
  and so on.
  A circle in the plain, for example,
  can be represented by a 6-tuple
  (the coordinates of three distinct points in the circle)
  or a 3-tuple (the coordinates of the center and the radius).
  And, to carry the example on,
  an input of $m$ circles can be interpreted as a tuple
  $\y\in\R^{n}\RF$
  with $n:=6m$ if we choose the first variant.

  We define the \emph{perturbation of $\y$}
  as a random additive distortion of its components.\footnote{%
    There is no unique definition of perturbation in geometry
    (see the introduction in~\cite{S98}).}
  We call $\UU_\delta(\y)\subset\R^n$
  a \emph{perturbation area} with
  \emph{perturbation parameter $\delta$\label{def-pert-para-inline}}
  if
  \begin{quote}
    1. $\delta\in\R_{>0}^n$, \\
    2. $y\in\UU_\delta(\y)$ implies $|y_i-\y_i|\le\delta_i$ for $1\le i\le n$ and \\
    3. $\UU_\delta(\y)$ contains an (open) neighborhood of $\y$.
  \end{quote}
  Note that $\UU_\delta(\y)$ is not a discrete set
  whereas $\UU_\delta(\y)\RF$ is finite.
  In our example,
  if we allow a circular perturbation of the $3m$ points which
  define the $m$ input circles,
  the perturbation area is the Cartesian product of $3m$ planar discs.
  We make the observation that
  even if we consider the input as a plain sequence of numbers,
  the perturbation area may look very special---we cannot
  neglect the geometrical interpretation here!
  In this context,
  we define an
  \emph{axis-parallel perturbation area}
  $U_\delta(\y)\label{def-U-inline}$
  as a box which is centered in $\y$
  and has edge length $2\delta_i$ parallel to the $i$-th main axis
  (and always denote it by the latin letter $U$ instead of $\UU$).
  This definition significantly simplifies the shape of the perturbation area.

  Naturally,
  the perturbed input must also be a vector of floating-point numbers.
  For now,
  we denote the \emph{perturbed input} by
  $y\in\UU(\y)\RF\label{def-y-inline}$.
  (We remark that we refine this definition
  on page~\pageref{inline-random-grid}).

  The analysis of $\ACP$
  depends on the analysis of $\AG$ and its predicates
  (see Section~\ref{sec-ana-algo}).
  We remember that
  a \emph{geometric predicate},
  which is true or false,
  is decided by the sign of a
  \emph{real-valued function} $f$.\label{def-f-inline}
  Therefore we introduce further quantities to describe such functions.
  We assume that $f$ is a $k$-ary real-valued function and
  that $k\ll n$.
  We further assume that we evaluate $f$ at $k$ distinct perturbed input values,
  that means,
  we evaluate $f(y_{\sigma(1)},\ldots,y_{\sigma(k)})$
  where $\sigma:\{1,\ldots,k\}\to\{1,\ldots,n\}$ is injective.
  The mapping $\sigma$ is injective to guarantee
  that the variables in the formula of $f$ are independent of each other.
  To not get the indices mixed up in the analysis,
  we rename the argument list of $f$ into
  $x_i := y_{\sigma(i)}\label{def-x-inline}$
  for $1\le i\le k$.
  In the same way we also rename the affected input values
  $\x_i := \y_{\sigma(i)}\label{def-xbar-inline}$.
  We denote the set of \emph{valid arguments for $f$} by $A\label{def-A-inline}$.

  In the analysis,
  $\E$ implicitly describes an upper-bound
  on the absolute value of perturbed input values
  in the way
  \begin{eqnarray}\label{def-e}\label{for-e-min}
    \E &:=& \min
      \left\{
        e'\in\N  \ST
	  \text{$|\y_i| + \delta_i \le 2^{e'}$ for all $1\le i \le n$}
      \right\}.
  \end{eqnarray}
  We call $\E$ the \emph{input value parameter}.
  Be aware that this is just a bound on the arguments of $f$
  and not a bound on the absolute value of $f$.
  At the moment we assume that
  the absolute value of $f$ is bounded on $A$
  and that the size $K$ of the exponent of the floating-point arithmetic $\FLK$
  is sufficiently large to avoid overflow errors
  during the evaluation of $f$.
  In Section~\ref{sec-rational-function},
  we drop this assumption and discuss the treatment of range issues.

  Below we summarize the basic quantities
  which are needed for the analysis
  of a function $f$.
\begin{definition}\label{def-predi-con}
    We call $\PRED$
    a \emph{predicate description} if:
    \begin{quote}
      1. $k\in\N$, \\
      2. $A\subset\R^k$, \\
      3. $\delta\in\R_{>0}^k$, \\
      4. $\E$ is as it is defined in Formula~(\ref{def-e}), \\
      5. $\U_\delta(A)\subset[-2^{\E},2^{\E}]^k$ and \\ %
      6. $f:\U_\delta(A)\to\R$.
    \end{quote}
\end{definition}
\noindent
  Predicate descriptions are used on and on.
  We extend the notion
  in Definition~\ref{def-predi-con-2} on page~\pageref{def-predi-con-2}
  and in Definition~\ref{def-predi-con-3} on page~\pageref{def-predi-con-3}.

\subsection{Success Probability, Grid Points}
\label{sec-succ-prob}

  The controlled-perturbation algorithm $\ACP$ terminates eventually
  if there is a positive probability that $\AG$ terminates successfully.
  The latter condition is fulfilled if $f$ has the property:
  The probability of a successful evaluation of $f$
  gets arbitrarily close to the certain event
  just by increasing the precision $L$.
  We call this property \emph{applicability}
  and specify it
  in Section~\ref{sec-veri-succ}.

  In this section we derive a definition for the success probability
  that is appropriate for the analysis
  and that is valid for floating-point evaluations.
  We begin with the question:
  What is the least probability that
  a guarded evaluation of $f$ is successful in a run of $\AG$
  under the arithmetic $\F$?
  We assume
  that each random point is chosen with the same probability.
  Then the answer is
  \begin{eqnarray*}
    \PR (f\RF)
      &:=&
	\min_{\x\in A} \;
        \frac
	  {\CARDI{
	    \left\{
	      x\in\U_\delta(\x)\RF
	      \ST
	      \text{$\GG(x)$ is true}
	    \right\}
	  }}
          {\CARDI{\U_\delta(\x)\RF}}.
  \end{eqnarray*}
  The definition really reflects the actual behavior of $f$.
  The probability is the number of guarded (floating-point) inputs
  divided by the total number of inputs
  and considers the worst-case for all perturbation areas.

\subsection*{Issue~1: Floating-point arithmetic is hard to be analyzed directly}

  Because floating-point arithmetic and its rounding policy
  can hardly be analyzed directly,
  we aim at deriving a corresponding formula for real arithmetic.
  In real space,
  we use the Lebesgue measure\footnote{%
    Measure Theory: The Lebesgue measure is defined in Forster~\cite{F11}.}
  $\mu\label{def-mu-inline}$
  to determine the volume of areas.
  Therefore we are looking for a formula like
  \begin{eqnarray}\label{for-prob-guess}
    \PR (\text{$f$})
      &:=&
	\min_{\x\in A} \;
        \frac
	  {\mu\left(
	    \left\{
	      x\in\U_\delta(\x)
	      \ST
	      \text{$\GG'(x)$ is true}
	    \right\}
	  \right)}
          {\mu({\U_\delta(\x)})}
  \end{eqnarray}
  where the predicate $\GG':\U_\delta(A)\to\text{\{true, false\}}$
  equals $\GG$
  at arguments with floating-point representation.
\subsection*{Issue~2: The set of floating-point numbers has measure zero}

  It is well-known that
  the set $\U_\delta(\x)\RF$ is finite and
  that its superset $\U_\delta(\x)\RQ$ is a set of measure zero.
  Be aware that
  the fraction in Formula~(\ref{for-prob-guess})
  does not change
  if we redefine $f$ on a set of measure zero.
  This implies some bizarre situations.
  For example,\footnote{%
    Note that
    there are finite sets of exceptional points that
    lead to similar counter-examples
    since every exception influences the practical behavior of the function
    (and $L$ is finite).}
\newcommand{\FT}{{f_\text{true}}}%
\newcommand{\FF}{{f_\text{false}}}%
  let $\FF:\U_\delta(A)\to\R$ be
  \begin{eqnarray*}
    \FF(x)&:=&\left\{
      \begin{array}{r@{\quad:\quad}l}
        f(x) & x\not\in\U_\delta(A)\RQ \\
        0 & otherwise
      \end{array} \right.
  \end{eqnarray*}
  and let $\FT:\U_\delta(A)\to\R$ be
  \begin{eqnarray*}
    \FT(x)&:=&\left\{
      \begin{array}{r@{\quad:\quad}l}
        f(x) & x\not\in\U_\delta(A)\RQ \\
        B & otherwise
      \end{array} \right.
  \end{eqnarray*}
  where $B\in\R_{>0}$ is large enough to guarantee
  that the guard $\GG$ evaluates to true in the latter case.
  Be aware that
  $\PR(\text{$\FF$}) = \PR(\text{$\FT$})$
  due to Formula~(\ref{for-prob-guess})
  whereas both implementations ``$\AG$ with $\FT$'' and ``$\AG$ with $\FF$''
  behave most conflictive:
  The former is always successful whereas
  the latter never succeeds.
  We remark that
  the assumption ``$f$ is (upper) continuous almost everywhere''
  does not solve the issue because
  ``almost everywhere'' means
  ``with the exception of a set of measure zero.''
  We have to introduce several restrictions
  to get able to deal with situations like that.

\subsection*{Issue~3: There is no general relation between $\PR(f\RF)$ and $\PR(f)$}

  This problem gets already visible in the 1-dimensional case.
\begin{example}\label{ex-density-1}
  Let $\F=\F_{2,3}$ be the floating-point arithmetic with
  $L=2$ and $K=3$.
  In addition
  let $U=[0,2]$, $R_1=[0,1]$ and $R_2=[1,2]$ be intervals.
  The situation is depicted in Figure~\ref{fig-fp-ratio-f}.
  \begin{figure}[h]\centering
    \includegraphics[width=.95\columnwidth]{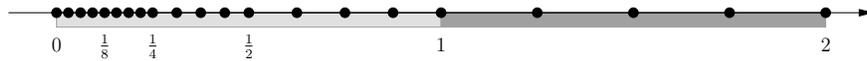}
    \caption{Distribution of the discrete set $\F_{2,3}$
      within the interval $[0,2]$.}
    \label{fig-fp-ratio-f}
  \end{figure}

\noindent
  What is the probability that a randomly chosen point $x\in U$
  lies inside of $R_1$,
  respectively $R_2$,
  for points in $U$ or $U\RF$?
  Note that $R_1$ and $R_2$ have the same length.
  For $R_1=[0,1]$ we have
  \begin{eqnarray*}
    \PR(R_1) = \frac{1}{2} \FORMSEP {<} \FORMSEP \PR(R_1\RF) = \frac{17}{21},
  \end{eqnarray*}
  that means,
  the probability is higher for floating-point arithmetic.
  On the other hand,
  for $R_2=[1,2]$ we have
  \begin{eqnarray*}
    \PR(R_2) = \frac{1}{2} \FORMSEP {>} \FORMSEP \PR(R_2\RF) = \frac{5}{21},
  \end{eqnarray*}
  that means,
  the probability is higher for real arithmetic.
  \hfill$\bigcirc$
\end{example}
  We derive from Example~\ref{ex-density-1}
  that there is no general relation between $\PR(f\RF)$ and $\PR(f)$
  because of the distribution of $\F$.

\subsection*{Issue~4: Distribution of $\F$ is non-uniform}

  Because the discrete set of floating-point numbers
  is non-uniformly distributed in general,
  we smartly alter the perturbation policy:
  We restrict the random choice of floating-point numbers to selected numbers
  that lie on a regular grid.
\begin{definition}[grid]\label{def-grid-points}
  Let $\E$ be as it is defined in Formula~(\ref{def-e}) and
  let $\FLK$ be a floating-point arithmetic (with $\E \ll 2^{K-1}$).
  We define
  \begin{eqnarray}\label{for-def-tau}
    \tau
      &:=& 2^{\E-L-1}\label{def-tau-inline}.
  \end{eqnarray}
  We call
  \begin{eqnarray}\label{for-def-grid}
    \GLKE
      &:=& \left\{
             \lambda\tau \ST \text{$\lambda\in\Z$ and $\lambda\tau\in[-2^\E,2^\E]$}
	   \right\}
  \end{eqnarray}
  the \emph{grid points induced by $\E$ with respect to $\FLK$}
  and we call $\tau$ the \emph{grid unit of $\GLKE$}.
  Furthermore,
  we denote the grid points $\G$ inside of a set $X\subset\R^k$ by
  \begin{eqnarray*}
    X\RG &:=& X\cap \G^k.
  \end{eqnarray*}
\end{definition}
  Again we omit the indices
  whenever they do not deserve special attention.
  We observe that
  the grid unit $\tau$ is the maximum distance between two adjacent
  points in $\F \cap [-2^\E,2^\E]$.
  We observe further that
  the grid points $\G$ form a subset of $\F$.
  Be aware that the symbol $\F$ represents a set or an arithmetic
  whereas the symbol $\G$ always represents a set.
  It is important to see that the underlying arithmetic is still $\F$.
  We have introduced $\G$
  only to change the definition of the \emph{original perturbation area}
  into $\UUU_\delta(\y)\RG\label{inline-random-grid}$.
  This leads to the \emph{final version of the success probability of $f$}:
  The least probability that
  a guarded evaluation of $f$ is successful for inputs in $\G$
  under the arithmetic $\F$ is
  \begin{eqnarray}\label{for-prob-rest-G}
    \PR (f\RG)
      &:=&
	\min_{\x\in A} \;
        \frac
	  {\CARDI{
	    \left\{
	      x\in\U_\delta(\x)\RG
	      \ST
	      \text{$\GG(x)$ is true}
	    \right\}
	  }}
          {\CARDI{\U_\delta(\x)\RG}}.
  \end{eqnarray}
  Before we continue this consideration,
  we add a remark on the implementation of the perturbation area
  $\UUU_\delta(\y)\RG$.
\begin{remark}\label{perturb-implem-inline}
  Because the points in $\G$ are uniformly distributed,
  the implementation of the perturbation
  is significantly simplified
  to the random choice of integer $\lambda$ in Formula~(\ref{for-def-grid}).
  This functionality is made available by
  basically all higher programming languages.
  Apart from that
  we generate floating-point numbers with the largest possible number
  of trailing zeros.
  {This possibly reduces the rounding error in practice.}
\end{remark}

\subsection*{Issue~5: Projection of $\UUU_\delta(\y)\RG$ is non-uniform}

  The \emph{original perturbation area}
  $\UUU_\delta(\y)\RG$
  is a discrete set of uniformly distributed points
  of which every point is chosen with the same probability.
  As a consequence,
  the \emph{predicate perturbation area} $\U_\delta(\x)\RG$
  is also uniformly distributed.
  But it is important to see that
  this does not imply
  that all points in the projected grid
  appear with the same probability!
  We illustrate, explain and solve this issue in Section~\ref{sec-ana-algo}.
  For now
  we continue our consideration
  under the assumption
  that all points in $\U_\delta(\x)\RG$ are uniformly distributed and
  randomly chosen with the same probability.

\subsection*{Issue~6: Analyses for various perturbation areas may differ}

  In the determination of $\PR(f\RG)$ in Formula~(\ref{for-prob-rest-G}),
  we encounter the difficulty
  to find the minimum ratio
  between the {guarded} and {all possible} inputs
  \emph{for all possible perturbation areas},
  that means, for all $\bar{x}\in A$.
  We can address this problem with a simple worst-case consideration
  if we cannot gain (or do not want to gain)
  further insight into the behavior of $f$:
  We just expect that,
  whatever could negatively affect the analysis of $f$ within
  the total predicate perturbation area $\U_\delta(A)$,
  affects the perturbation area $\U_\delta(\bar{x})$
  under consideration.
  This way, we safely obtain a lower bound on the minimum.

\subsection*{Issue~7: There is no general relation between $\PR(f\RG)$ and $\PR(f)$}

\begin{example}\label{ex-density-2}
  We continue Example~\ref{ex-density-1}.
  In addition
  let $R_3=[\frac{1}{10},\frac{9}{10}]$ be an interval.
  Because $U\subseteq[-2^1,2^1]$,
  we have $\E=1$ and $\tau=2^{\E-L-1}=\frac{1}{4}$.
  The situation is depicted in Figure~\ref{fig-fp-ratio-g}.
  \begin{figure}[h]\centering
    \includegraphics[width=.95\columnwidth]{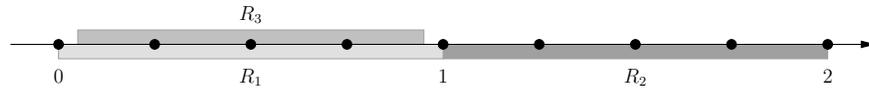}
    \caption{The distribution of the grid points $\G_{2,3,1}$
      within the interval $[0,2]$.}
    \label{fig-fp-ratio-g}
  \end{figure}

\noindent
  Again we compare the continuous and the discrete case:
  What is the probability that a randomly chosen point $x\in U$
  lies inside of $R_1$ ($R_2$ or $R_3$, respectively)?
  The probability is now higher
  for $R_1$ and $R_2$ in the discrete case
  \begin{eqnarray*}
    \PR(R_1) =
    \PR(R_2) = \frac{1}{2}
    \quad < \quad
    \PR(R_1\RG) =
    \PR(R_2\RG) = \frac{5}{9},
  \end{eqnarray*}
  and higher for $R_3$
  \begin{eqnarray*}
    \PR(R_3) = \frac{2}{5} \quad > \quad \PR(R_3\RG) = \frac{1}{3}
  \end{eqnarray*}
  in the real case.
  \hfill$\bigcirc$
\end{example}
  We make the observation that
  the restriction to points in $\G$
  does not entirely solve the initial problem:
  We still cannot relate the probability $\PR(f)$ with $\PR(f\RG)$ in general.
  To improve the estimate,
  we need another trick that we indicate in Example~\ref{ex-density-3}:
  \emph{If we make the interval
  slightly larger,
  we can safely determine the inequality.}
\begin{example}\label{ex-density-3}
  Let $\tau$ be the grid unit of $\G$.
  We define three intervals $R\subset\RAUG\subset U$.
  Let $U\subset\R$
  be a closed interval of length $\lambda_0\tau$
  with $\lambda_0\in\N$.
  Let $\RAUG\subset U$
  be an interval of length at least $\tau$
  that has the limits $\RAUG:=[a-\frac{\tau}{2},b+\frac{\tau}{2}]$
  for $a,b\in\R$.
  Finally, we define $R:=[a,b]$.
  In addition let $\lambda\in\N$ be such that
  \begin{eqnarray*}
    \lambda\tau \;\; \le \;\; \mu(\RAUG) \;\; < \;\; (\lambda+1)\tau.
  \end{eqnarray*}
  We observe that the number of grid points in
  $R\RG$ and $\RAUG\RG$
  is bounded by
  \begin{eqnarray*}
    \lambda-1 \;\;\le\;\; \CARDI{R\RG} \;\;\le\;\;
    \lambda \;\;\le\;\; \CARDI{\RAUG\RG} \;\;\le\;\;
    \lambda+1.
  \end{eqnarray*}
  Moreover, we make the important observation that
  \begin{eqnarray*}
    \frac{\CARDI{R\RG}}{\CARDI{U\RG}}
      \;\; \le \;\; \frac{\lambda}{\lambda_0+1}
      \;\; \le \;\; \frac{\lambda}{\lambda_0}
      \;\; =   \;\; \frac{\lambda\tau}{\lambda_0\tau}
      \;\; \le   \;\; \frac{\mu(\RAUG)}{\mu(U)}.
  \end{eqnarray*}
  That means,
  it is more likely that a random point in $U$ lies inside of $\RAUG$
  than a random point in $U\RG$ lies inside of $R\RG$.
  The inequality
  \begin{eqnarray*}
    \PR(R\RG) &\le& \PR(\RAUG)
  \end{eqnarray*}
  is valid independently of the actual choice or location of $R$.
\hfill$\bigcirc$
\end{example}

\subsection*{Issue~8: There is still no general relation between $\PR(f\RG)$ and $\PR(f)$}

  The probability $\PR(f)$
  is defined as the ratio of volumes.
  The definition is, in particular,
  independent of the location and shape of the involved sets.
  As an example,
  we consider the three different (shaded) regions
  in Figure~\ref{fig-volume-location-shape}
  which all have the same volume.
  \begin{figure}[h]\centering
    \includegraphics[width=.95\columnwidth]{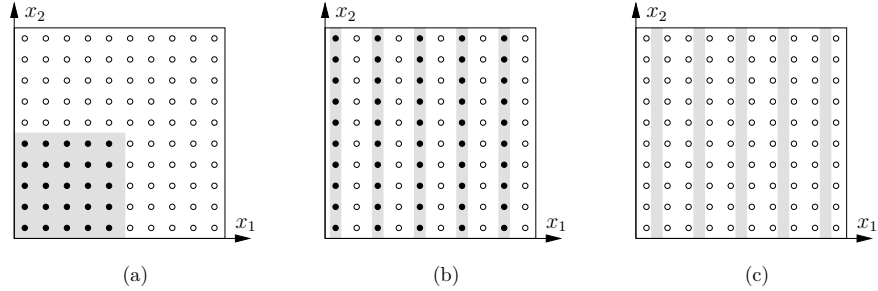}
    \caption[The shape of the region of uncertainty matters.]
      {The volume of the shaded region $R$ is the same in the three pictures.
      Depending on the shape and location of $R$,
      it covers various fractions of the discrete set $\G$.
      For example:
      (a) a quarter, (b) a half, (c) nothing.}
    \label{fig-volume-location-shape}
  \end{figure}

  We make the important observation that
  the shape and location matter
  if we derive the induced ratio for points in $\G$.
  The discrepancy between the ratios
  is caused by the implicit assumption
  that the grid unit $\tau$ is sufficiently small.
  (Asymptotically, the ratios approach the same limit
  in the three illustrated examples for $\tau\to 0$.)
  Be aware that making this assumption explicit
  leads to a second constraint on the precision $L$
  which we call the \emph{grid unit condition}.
  To solve this issue,
  we need a way to adjust the grid unit $\tau$ to the shape of $R$.
  We address this issue in general in Section~\ref{sec-relate-tau-gamma}.
  For now
  we continue our consideration under the assumption that
  this problem is solved.

\subsection*{Summary and validation of $\PR(f\RG)$}

  We summarize our considerations so far.
  The analysis of a guarded algorithm
  {must} reflect its actual behavior.
  (What would be the meaning of the analysis, otherwise?)
  Therefore we have defined the success probability
  of a floating-point evaluation of $f$
  in Formula~(\ref{for-prob-rest-G})
  such that it is based on the behavior of guards.
  Furthermore,
  we have studied the interrelationship
  between the success probability
  for floating-point and real arithmetic
  to prepare the analysis in real space.
  Be aware that
  we have introduced a specialized perturbation on a regular grid $\G$
  (in practice and in analysis)
  which is necessary for the derivation of the interrelationship.
  Moreover,
  we now make this relationship explicit for a single interval.
  (The general relationship is formulated
  in Section~\ref{sec-relate-tau-gamma}.)
\begin{example}\label{ex-density-4}
  (Continuation of Example~\ref{ex-density-3}.)
  Let $f:U\to\R$.
  We assume the following property of $R$:
  {If $x\in U\RG$ lies outside of $R$
  then the guard $\GG(x)$ is true.}
  Then we have
  \begin{eqnarray*}
        \PR(f\RG)
      &=&
        \frac
	  {\CARDI{
	    \left\{
	      x\in U\RG
	      \ST
	      \text{$\GG(x)$ is true}
	    \right\}
	  }}
          {\CARDI{U\RG}} \\
      &\ge&
        1 \,-\, \frac{\CARDI{R\RG}}{\CARDI{U\RG}} \\
      &\ge&
        1 \,-\, \frac{\mu(\RAUG)}{\mu(U)}.
  \end{eqnarray*}
  We conclude:
  \emph{If we prove
  by means of abstract mathematics
  that
  \begin{eqnarray*}
        1 \,-\, \frac{\mu(\RAUG)}{\mu(U)} &\ge& p
  \end{eqnarray*}
  for a probability $p\in(0,1)$,
  we have implicitly proven that}
  \begin{eqnarray*}
        \PR(f\RG) &\ge& p
  \end{eqnarray*}
  for a randomly chosen grid point in $\G$.
  Be aware that $\PR(f\RG)$
  is defined only by discrete quantities.
\hfill$\bigcirc$
\end{example}

\subsection*{Warning: processing exceptional points}

  We explain in this paragraph
  why it is absolutely non-obvious
  how to process exceptional points in general.
  Assume that we want to exclude the set $D\subset A$
  from the analysis.
  This changes our success probability from Formula~(\ref{for-prob-rest-G})
  into
  \begin{eqnarray*}
    \PR (f\RG)
      &=&
	\min_{\x\in A} \;
        \frac
	  {\CARDI{
	    \left\{
	      x\in\U_\delta(\x)\RG
	      \ST
	      \text{$\GG(x)$ is true}
	    \right\}
	    \setminus D
	  }}
          {\CARDI{\U_\delta(\x)\RG}} \\
	&\ge&
	\min_{\x\in A} \;
	 \frac
	  { \max \left\{
	      0, \;
	      {
	      \left|\left\{
	        x\in\U_\delta(\x)\RG
	        \ST
	        \text{$\GG(x)$ is true}
	      \right\}\right|
	      }
	      -
	      {|D|}
	    \right\}
	  }
          {\CARDI{\U_\delta(\x)\RG}}.
  \end{eqnarray*}
  To obtain a practicable solution,
  it is reasonable to assume that $D$ is finite and, moreover, that
  $|D|\ll\CARDI{\U_\delta(\x)\RG}$.
  This changes the relation in Example~\ref{ex-density-4}
  into:
  \begin{eqnarray*}
        \PR(f\RG)
      &\ge&
        \max \left\{
	  0, \;
        1 - \frac{\mu(\RAUG)}{\mu(U)}
	  -
	  \frac
	  {|D|}
          {\CARDI{\U_\delta(\x)\RG}}
        \right\}.
  \end{eqnarray*}
  It is important to see
  that this estimate still contains two quantities
  that depend on the floating-point arithmetic.
  But our plan was to get rid of this dependency.
  In spite of the simplifying assumptions
  it is non-obvious how to perform the analysis in real space in general.
  \emph{Our suggested solution to this issue is to avoid
  exceptional points.
  Alternatively
  we declare them critical (see next section)
  which triggers an exclusion of their environment.}

\subsection{Fp-safety Bound, Critical Set, Region of Uncertainty}
\label{sec-further-quantities}

\subsection*{The fp-safety bound}

  We introduce a predicate that can certify
  the correct sign of floating-point evaluations.
  The essential part of this predicate is the fp-safety bound.
  We show in Section~\ref{sec-guards-safetybounds}
  that there are fp-safety bounds for a wide class of functions.
\begin{definition}[lower fp-safety bound]\label{def-fpsafetybound}
  Let $\PRED$ be a predicate description.
  Let
  $\Sinf:\N\to\R_{\ge 0}$ be a monotonically decreasing function
  that maps a precision $L$ to a non-negative value.
  We call $\Sinf$ a \emph{(lower) fp-safety bound for $f\!$ on $A$}
  if the statement
  \begin{eqnarray}\label{for-fpsafety-final}
      |f(x)| > \Sinf(L)
      \FORMSEP &\Rightarrow& \FORMSEP
    \sign(f(x)\RFL) = \sign(f(x))
  \end{eqnarray}
  is true for every precision
  $L\in\N$ and for all $x\in \U_\delta(A)\RFL$.
\end{definition}
  For the time being,
  we consider $K$ to be a constant.
  We drop this assumption
  in Section~\ref{sec-rational-function}
  where we introduce \emph{upper} fp-safety bounds.
  Until then
  we only consider \emph{lower} fp-safety bounds.

\subsection*{The critical set}

  Next we introduce a classification of the points in $\U_\delta(A)$
  in dependence on their neighborhood.
  (We refine the definition on Page~\pageref{def-critical-set-second}.)
\begin{definition}[critical]\label{def-critical-set}
  Let $\PRED$ be a predicate description.
  We call a point $c\in\U_\delta(\x)$
  \emph{critical} if
  \begin{eqnarray}\label{for-def-crit}
    \inf_{x\in U_\varepsilon(c)\setminus\{c\}} \; \left|f(x)\right| &=& 0
  \end{eqnarray}
  on a neighborhood $U_\varepsilon(c)$
  for infinitesimal
  small $\varepsilon>0$.
  Furthermore,
  we call zeros of $f$ that are not critical \emph{less-critical}.
  Points that are neither critical nor less-critical
  are called \emph{non-critical}.
  We define
  the \emph{critical set $C_{f,\delta}$
  of $f$ at $\x\in A$ with respect to $\delta$}
  as the union of critical and less-critical points within $\U_\delta(\x)$.
\end{definition}
  In other words,
  we call $c$ critical
  if there is a Cauchy sequence\footnote{%
    Analysis: Cauchy sequence is defined in Forster~\cite{F06}.}
  $(a_i)_{i\in \N}$ in $\bar{U}_\delta(\x)\setminus\{c\}$
  where $\lim_{i\to\infty} a_i=c$
  and $\lim_{i\to\infty} f(a_i)=0$.
  We remember that
  the metric space\footnote{%
    Topology: Metric space and completeness are defined in Jänich~\cite{J01}.}
  $\R^k$ is complete,
  that means,
  the limit of the sequence $(a_i)$
  lies inside of the closure
  $\U_\delta(\x)$.
  Sometimes we omit the indices
  of the critical set $C$
  if they are given by the context.
\begin{example}
  We consider the three functions
  that are depicted in Figure~(\ref{fig-crit-set-compare}).
  Let $f_1(x)=x^2$.
  Let $f_2(x)=x^2$ for $x\ne 0$ and $f_2(0)=2$.
  Let $f_3(x)=x^2+1$ for $x\not\in \{-2,1\}$ and $f_3(-2)=0$ and $f_3(1)=0.2$.
  \begin{figure}[h]\centering
    \includegraphics[width=.95\columnwidth]{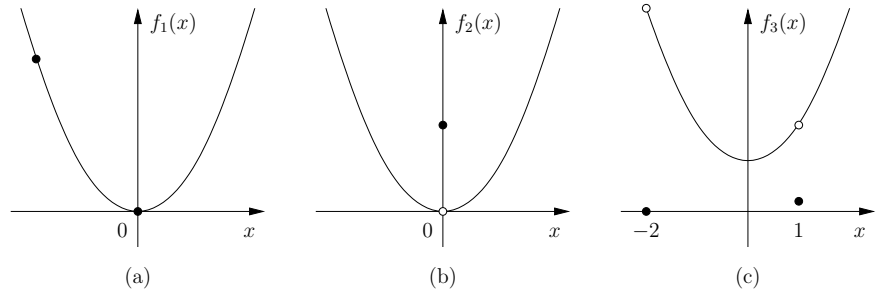}
    \caption{Examples of critical, less-critical and non-critical points.}
    \label{fig-crit-set-compare}
  \end{figure}
  The point $x=0$ in Picture~(a)
  is a zero and a critical point for $f_1$.
  In (a), every argument $x\ne 0$ is non-critical for $f_1$.
  In (b), $f_2$ is non-zero at $x=0$, but $x=0$ is a critical point for $f_2$.
  In (c), the argument $x=-2$ is less-critical for $f_3$ and
  the argument $x=1$ is non-critical for $f_3$.
\hfill$\bigcirc$
\end{example}
  What is the difference of critical and less-critical points?
  We observe
  that the point $c$ is excluded from its neighborhood
  in Formula~(\ref{for-def-crit}).
  Zeros of $f$ would trivially be critical otherwise.
  Furthermore,
  we observe that zeros of continuous functions are always critical.
  For our purpose it is important to see that
  the infimum of $|f|$ is positive
  if we exclude the less-critical points \emph{itself}
  and \emph{neighborhoods} of critical points.
  Be aware
  that we technically could treat both kinds differently in the analysis
  and still ensure that the result of the analysis is
  valid for floating-point arithmetic.
  Only for simplicity we deal with them in the same way
  by adding these points to the critical set.
  Only for simplicity we also add exceptional points to the critical set.

\subsection*{The region of uncertainty}

  The next construction is a certain environment of the critical set.
\begin{definition}[region of uncertainty]\label{def-region-uncertainty}
  Let $\PRED$ be a predicate description.
  In addition
  let $\gamma\in\R_{>0}^k$.
  We call
  \begin{eqnarray}\label{for-region-uncert}
    R_{f,\gamma}(\x)
    &:=&
      \bar{U}_\delta(\x)
      \;\cap\;
      \left( \bigcup_{c\in C_{f,\delta}(\x)} U_\gamma(c) \right)
  \end{eqnarray}
  the \emph{region of uncertainty for $f$
  induced by $\gamma$ with respect to $\x$.}
\end{definition}
  In our presentation
  we use the axis-parallel boxes $U_\gamma(c)$
  to define
  the specific $\gamma$-neighborhood of $C$;
  other shapes require adjustments, see Section~\ref{sec-relate-tau-gamma}.
  The sets $U_\gamma(c)$ are open and
  the complement of $R_{f,\gamma}(\x)$ in $\U_\delta(\x)$
  is closed.
  We omit the indices
  of the region of uncertainty $R$
  if they are given by the context.

  The vector $\gamma\label{def-gamma-inline}$ 
  defines the tuple of componentwise distances to $c$.
  The presentation requires a formal definition of the
  set of all admissible $\gamma$.
  This set is either a box or a line.
  Let $\hat\gamma\in\R_{>0}^k$.
  Then we define the unique open axis-parallel box
  with vertices $0$ and $\hat\gamma$ as
  \begin{eqnarray*}\label{def-GAB-inline}
    \GABG &:=& \left\{ \gamma'=(\gamma'_1,\ldots,\gamma'_k)
      : \text{$\gamma'_i\in(0,\hat\gamma_i)$ for all $i\in I$} \right\}
  \end{eqnarray*}
  and the open diagonal from 0 to $\hat\gamma$ inside of $\GABG$ as
  \begin{eqnarray*}\label{def-GAL-inline}
    \GALG &:=& \left\{ \gamma : \text{$\gamma = \lambda \hat\gamma$ with $\lambda\in(0,1)$} \right\}.
  \end{eqnarray*}
  It is important
  that the $\gamma_i$ can be chosen arbitrarily small whereas
  the upper bounds $\hat\gamma_i$ are only introduced for technical reasons;
  we assume that $\hat\gamma$ is ``sufficiently'' small.\footnote{%
    It is fine to ignore this information during first reading.
    More information and the formal bound is given
    in Remark~\ref{rem-def-region-suit}.2
    on Page~\pageref{rem-def-region-suit}.}
  Occasionally we omit $\hat\gamma$.

  We have already seen
  that there is need to augment the region of uncertainty
  (see Issue~7 and~8
  in Section~\ref{sec-succ-prob}).
  This task is accomplished
  by
  the mapping $\gamma\mapsto\AUG(\gamma):=\frac{\gamma}{t}$
  for $t\in(0,1)$.
  For technical reasons we remark that
  $\gamma\in\GABG$ if $\AUG(\gamma)\in\GABG$, and
  $\gamma\in\GALG$ if $\AUG(\gamma)\in\GALG$.
  We call $R_{f,\AUG(\gamma)}\label{inline-def-aug-rou}$
  the \emph{augmented region of uncertainty for $f$ under} $\AUG(\gamma)$.
  By $\Gamma\label{def-Gamma-inline}$
  we denote the \emph{set of valid augmented $\gamma$}
  and include it in the predicate description.
\begin{definition}\label{def-predi-con-2}
  We extend Definition~\ref{def-predi-con} and
  call $\PREDG$ a \emph{predicate description} if:
  7. $\Gamma=\GALG$ or $\Gamma=\GABG$
       for a sufficiently small $\hat\gamma\in\R_{>0}^k$.
\end{definition}

\subsection{Overview: Classification of the Input}
\label{sec-over-func-argu}

  In practice and in the analysis
  we deal with real-valued functions whose signs decide predicates.
  The arguments of these functions belong to the perturbation area.
  In this section
  we give an overview of the various characteristics for function arguments
  that we have introduced so far.
  We strictly distinguish between terms of practice and terms of the analysis.

  The diagram of the practice-oriented terms
  is shown in Figure~\ref{fig-practice-oriented}.
  We consider the discrete perturbation area $U_{\!\delta}\,\RG$.
  Controlled perturbation algorithms $\ACP$
  are designed with intent
  to avoid the implementation of degenerate cases and
  to compute the combinatorial correct solution.
  Therefore
  the guards in the embedded algorithm $\AG$
  must fail for the zero set and
  for arguments whose evaluations lead to wrong signs.
  The guard is designed such that
  the evaluation is definitely fp-safe
  if the guard does not fail (light shaded region).
  Unfortunately
  there is no convenient way
  to count (or bound) the number of arguments in $U_{\!\delta}\,\RG$
  for which the guard fails.
  That is the reason why we perform the analysis with real arithmetic
  and introduce further terms.
  \begin{figure}[t]\centering
    \includegraphics[width=.85\columnwidth]{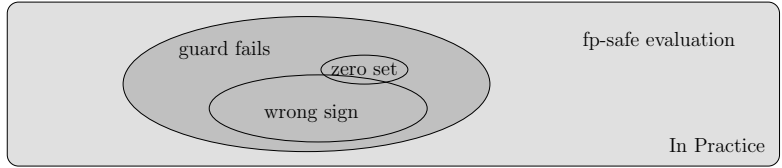}
    \caption{The diagram of the practice-oriented terms.}
    \label{fig-practice-oriented}
  \end{figure}

  The diagram of the analysis-oriented terms
  is shown in Figure~\ref{fig-analysis-oriented}.
  We consider the real perturbation area $U_{\!\delta}$.
     Instead of the zero set,
     we consider the critical set
     (see Definition~\ref{def-critical-set}).
     The critical set is a superset of the zero set.
     Then we choose the region of uncertainty
     as a neighborhood of the critical set
     (see Definition~\ref{def-region-uncertainty}).
     We augment the region of uncertainty to obtain a result
     that is also valid for floating-point evaluations.
     We intent to prove fp-safety outside of the augmented region of uncertainty
     (i.e. on the light shaded region).
     Therefore we design a fp-safety bound that is true
     outside of the region.
     This way we can guarantee that
     the evaluation of a guard (in practice)
     only fails on a subset of the augmented region (in the analysis).
  \begin{figure}[t]\centering
    \includegraphics[width=.85\columnwidth]{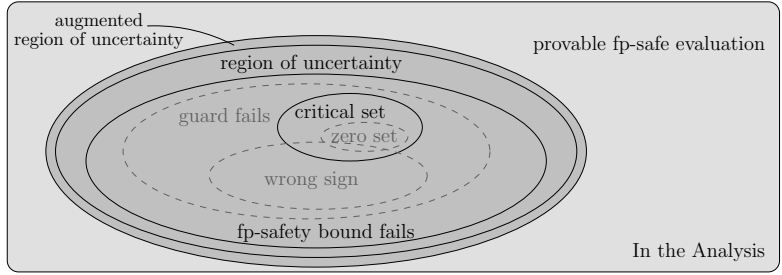}
    \caption{The diagram of the analysis-oriented terms (shown in black).}
    \label{fig-analysis-oriented}
  \end{figure}

\subsection{Applicability and Verifiability of Functions}
\label{sec-veri-succ}

  We study the circumstances
  under which we may \emph{apply} controlled perturbation to a predicate
  in practice
  and under which we can actually \emph{verify} its application
  in theory.
  We stress that we talk about a \emph{qualitative} analysis here;
  the desired \emph{quantitative} analysis is derived in the following sections.

  Furthermore, we want to remark that
  \emph{verifiability} is not necessary
  for the presentation of the analysis tool box.
  However,
  the distinction between applicability, verifiability and analyzability
  was important for the author during the development of the topic.
  We keep it in the presentation
  because it may also be helpful to the reader.
  Anyway, skipping this section is possible and even
  assuming equality between verifiability and analyzability
  will do no harm.

\subsection*{In practice}

  We specify the function property that
  the probability of a successful evaluation of $f$
  gets arbitrarily close to the certain event
  by increasing the precision.
\begin{definition}[applicable]
  \label{def-func-app}
  Let $\PRED$ be a predicate description.
  We call $f$ \emph{applicable}
  if for every $p\in(0,1)$ there is $L_p\in\N$ such that
  the guarded evaluation of $f$
  is successful
  at a randomly perturbed input $x\in \U_\delta(\x)\RGL$
  with probability at least $p$
  for every precision $L\in\N$ with $L\ge L_p$
  and every $\x\in A$.
\end{definition}
  Applicable functions can safely be used
  in guarded algorithms:
  Since the precision $L$ is increased (without limit)
  after a predicate has failed,
  the success probability gets arbitrarily close to 1
  for each predicate evaluation.
  As a consequence,
  the success probability of $\AG$ gets arbitrarily close to 1, too.

\subsection*{In the qualitative analysis}

  Unfortunately
  we cannot check directly
  if $f$ is applicable.
  Therefore we introduce two properties
  that imply applicability.
\begin{definition}\label{def-verifiable}
  Let $\PREDL$ be a predicate description.
  \begin{itemize}
    \item
{(region-condition).}
    For every $p\in(0,1)$ there is $\gamma\in\Rplus^k$ such that
    the geometric failure probability is bounded in the way
      \begin{eqnarray}
      \label{for-def-region-const}
        \frac{\mu(R_\gamma(\x))}{\mu(U_\delta(\x))} &\le& (1-p)
      \end{eqnarray}
    for all $\x\in A$.
    We call this condition the \emph{region-condition}.
    \item
{(safety-condition).}
    \label{def-safety-cond}
    There is a fp-safety bound $\Sinf:\N\to\R_{>0}$
    on $\U_\delta(A)$ with\footnote{%
      Technically,
      the assumption $\Sinf(L)\stackrel{!}{>}0$ is no restriction.}
    \begin{eqnarray}\label{for-safety-cond}
      \lim_{L\to\infty}\Sinf(L) &=& 0.
    \end{eqnarray}
    We call this condition
    the \emph{safety-condition}.
    \item
{(verifiable).}
  We call $f$
  \emph{verifiable on $\U_\delta(A)$ for controlled perturbation}
  if $f$ fulfills the region- and safety-condition.
  \end{itemize}
\end{definition}
\noindent
  The region-condition guarantees the adjustability of the
  volume of the region of uncertainty.
  Note that
  the region-condition is actually a condition on the critical set.
  It states that the critical set is sufficiently ``sparse''.

  The safety-condition guarantees the adjustability of the fp-safety bound.
  It states that
  for every $\varphi>0$
  there is a precision $\LS\in\N$ with the property that
  \begin{eqnarray}\label{for-def-safety-const}
    \Sinf(L)\le\varphi
  \end{eqnarray}
  for all $L\in\N$ with $L\ge \LS$.
  We give an example of a verifiable function.
\begin{example}%
  Let $A\subset\R$ be an interval,
  let $\delta\in\R_{>0}$ and
  let $f:\U_\delta(A)\to\R$
  be a univariate polynomial\footnote{%
    We avoid the usual notation $f\in\R[x]$
    to emphasize that the domain of $f$ \emph{must be bounded.}}
  of degree $d$ with real coefficients, i.e.,
  \begin{eqnarray*}
    f(x) &=& a_d \cdot x^d + a_{d-1} \cdot x^{d-1} +
               \ldots + a_1 \cdot x + a_0.
  \end{eqnarray*}
  We show that $f$ is verifiable.
  Part 1 (region-condition).
  Because of the fundamental theorem of algebra
  (e.g., see Lamprecht~\cite{L93}),
  $f$ has at most $d$ real roots.
  Therefore the size of the critical set $C_f$ is bounded by $d$
  and the volume of the region of uncertainty
  $R_\gamma(\x)$ is upper-bounded by $2d\gamma$.
  For a given $p\in(0,1)$
  we then choose
  \begin{eqnarray*}
    \gamma &:=& \frac{(1-p)\delta}{d}
  \end{eqnarray*}
  which fulfills the region-condition because of
  \begin{eqnarray*}
    \frac{\mu(R_\gamma(\x))}{\mu(U_\delta(\x))}
    &\le& \frac{2\gamma d}{2\delta} \;\; = \;\; 1-p.
  \end{eqnarray*}
Part 2 (safety-condition).
  Corollary~\ref{col-unipolysafety} on page~\pageref{col-unipolysafety}
  provides the fp-safety bound
  \begin{eqnarray*}
    \Sinf(L) &:=& (d+2) \, \max_{1\le i\le d} |a_i| \;\, 2^{\E(d+1)+1-L}
  \end{eqnarray*}
  for univariate polynomials.
  Since $\Sinf(L)$ converges to zero as $L$ approaches infinity,
  the safety-condition is fulfilled.
  Therefore $f$ is verifiable.
  \hfill$\bigcirc$
\end{example}
  We show that,
  if a function is verifiable,
  it has a positive lower bound on its absolute value
  outside of its region of uncertainty.
\begin{lemma}\label{lem-veri-minval}
  Let $\PREDL$ be a predicate description
  and let $f$ be verifiable.
  Then for every $\gamma\in\Rplus^k$, there is $\varphi\in\R_{>0}$ with
    \begin{eqnarray}
      \label{for-def-veri-minval}
      \varphi \le |f(x)|
    \end{eqnarray}
    for all $x\in \bar{U}_\delta(\x)\setminus R_\gamma(\x)$ and
    for all $\x\in A$.
\end{lemma}
\begin{proof}
  We assume the opposite.
  That means,
  in particular,
  for every $i\in\N$
  there is $a_i\in\bar{U}_\delta(\x)\setminus R_\gamma(\x)$
  such that $|f(a_i)| < \frac{1}{i}$.
  Then $(a_i)_{i\in\N}$ is a bounded sequence with accumulation points in
  $\bar{U}_\delta(\x)\setminus R_\gamma(\x)$.
  Those points must be critical
  and hence belong to $R_\gamma(\x)$.
  This is a contradiction.
  \qed
\end{proof}
  Finally we prove that verifiability of functions implies applicability.
\begin{lemma}\label{lem-veri-is-app}
  Let $\PREDL$ be a predicate description
  and let $f$ be verifiable.
  Then $f$ is applicable.
\end{lemma}
\begin{proof}
  Let $p\in(0,1)$.
  Then the geometric success probability is bounded by $p$.
  Therefore there must be an upper bound on the volume
  of the region of uncertainty
  (see Definition~\ref{def-verifiable}).
  In addition
  there is a precision $\LG$ such that
  we may interpret this region as an augmented region $R_{\AUG(\gamma)}$
  (see Theorem~\ref{theo-validation}).
  Furthermore,
  there must be a positive lower bound on $|f|$ outside of $R_\gamma$
  (see Lemma~\ref{lem-veri-minval}).
  Moreover,
  there must be a precision $\LS$
  for which the fp-safety bound is smaller than the bound on $|f|$.
  Be aware
  that this implies
  that the guarded evaluation of $f$ is successful
  at a randomly perturbed input
  with probability at least $p$
  for every precision $L\ge \max\{\LS,\LG\}$.
  That means,
  $f$ is applicable
  (see Definition~\ref{def-func-app}).
  \qed
\end{proof}

\section[General Analysis Tool Box (Introduction)]{General Analysis Tool Box}
\label{sec-ana-tool-box}

  The general analysis tool box
  to analyze controlled perturbation algorithms
  is presented in the remainder of the paper.
  We call the presentation a {tool box}
  because its components are strictly separated from each other
  and sometimes allow alternative derivations.
  In particular,
  we present three ways
  to analyze functions.
  Here we briefly introduce the tool box
  and refer to the detailed presentation of its components
  in the subsequent sections.
  \emph{The decomposition of the analysis into well-separated components
  and their precise description is an innovation of this presentation.}
  \begin{figure}[h]\centering
    \includegraphics[width=.95\columnwidth]{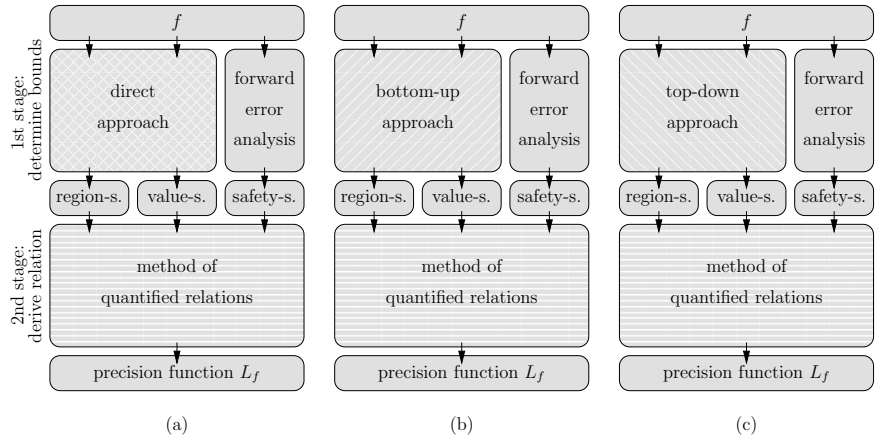}
    \caption{Illustration of the various ways to analyze functions.}
    \label{fig-illu-ana-func}
  \end{figure}

  The tool box is subdivided into components.
  At first we explain the \emph{analysis of functions}.
  The diagram in
  Figure~\ref{fig-illu-ana-func}
  illustrates three ways to analyze functions.
  We subdivide the function analysis in two stages.
  The analysis itself in the second stage requires three necessary bounds,
  also known as the \emph{interface},
  which are defined in Section~\ref{sec-nec-con-func}:
  \emph{region-suitability},
  \emph{value-suitability} and
  \emph{safety-suitability}.
  In Section~\ref{sec-meth-quan-rela}
  we introduce the \emph{method of quantified relations}
  which represents the actual analysis in the second stage.
  In the first stage,
  we pay special attention to the derivation of two bounds of the interface
  and suggest three different ways to solve the task.
  We show in Section~\ref{sec-direct-approach}
  how the bounds can be derived
  in a \emph{direct approach} from geometric measures.
  Furthermore,
  we show how to build-up the bounds for the desired function
  from simpler functions in a
  \emph{bottom-up approach}
  in Section~\ref{sec-bottom-up}.
  Moreover,
  we present a derivation of the bounds
  by means of a ``sequence of bounds'' in a \emph{top-down approach}
  in Section~\ref{sec-top-down}.
  Finally,
  we show how we can derive the third necessary bound of the interface
  with an \emph{error analysis}
  in Section~\ref{sec-guards-safetybounds}
  
  We deal with the \emph{analysis of algorithms}
  in Section~\ref{sec-ana-algo}.
  The idea is illustrated in
  Figure~\ref{fig-illu-ana-algo} on page~\pageref{fig-illu-ana-algo}.
  Again we subdivide the analysis in two stages.
  The actual analysis of algorithms
  is the \emph{method of distributed probability}
  which represents the second stage
  and is explained in Section~\ref{sec-meth-distri-prob}.
  The \emph{interface} between the stages is subdivided in two groups.
  Firstly,
  there are algorithm prerequisites
  (to the left of the dashed line in the figure).
  These bounds are defined and derived
  in Section~\ref{sec-nec-con-algo}:
  \emph{evaluation-suitability},
  \emph{predicate-suitability} and
  \emph{perturbation-suitability}.
  Secondly,
  there are predicate prerequisites
  (to the right of the dashed line in the figure).
  These are determined by means of function analyses.

\section[Justification of the Floating-point Analysis (Validation)]{Justification of Analyses in Real Space}
\label{sec-validation}

  This section addresses the problem
  to derive the success probability for floating-point evaluations
  from the success probability
  which we determine
  in real space.
  Analyses in real space are without meaning for
  controlled perturbation implementations (which use floating-point arithmetic),
  unless we determine a reliable relation
  between floating-point and real arithmetic.
  To achieve this goal,
  we introduce an additional constraint on the precision
  in Section~\ref{sec-relate-tau-gamma}
  and summarize our efforts in the determination of the success probability
  in Section~\ref{sec-over-validation}.
  \emph{This is the first presentation that
  adjusts the precision of the floating-point arithmetic to
  the shape of the region of uncertainty.}

\subsection{The Grid Unit Condition}
\label{sec-relate-tau-gamma}

  Here we adjust the distance of grid points
  (i.e., the grid unit $\tau$)
  to the ``width'' of the region of uncertainty $\gamma$.
  As we have seen in Issue~8
  in Section~\ref{sec-succ-prob},
  the grid unit $\tau$ must be sufficiently small
  (i.e., $L$ must be sufficiently large)
  to derive a reliable probability $\PR(f\RG)$
  from $\PR(f)$.
  The problem is illustrated in Figure~\ref{fig-volume-location-shape}
  on page~\pageref{fig-volume-location-shape}.
  We call this additional constraint on $L$
  the \emph{grid unit condition}
  \begin{eqnarray}\label{for-grid-unit-con}
    L &\ge& \LG
  \end{eqnarray}
  for a certain $\LG\in\N$.
  Informally,
  we demand that $\tau\ll\gamma$.
  Here we show
  how to derive the threshold $\LG$ formally.
  We refine the concept of the augmented region of uncertainty
  which we have mentioned briefly in Section~\ref{sec-succ-prob}.
  The discussion of Issue~7
  suggests an additive augmentation $\gamma=\AUG(\gamma')$
  that fulfills
  \begin{eqnarray*}
    \tau_0
      &\stackrel{(I)}{\le}&
        \gamma'_i
        \;\; \stackrel{(II)}{\le} \;\;
        \gamma_i-\tau_0
  \end{eqnarray*}
  for all $1\le i\le k$
  where $\tau_0$ is an upper bound on the grid unit.
  However,
  in the analysis
  it is easier to handle a multiplicative augmentation
  \begin{eqnarray*}
    \gamma
      &\stackrel{(III)}{:=}&
        \frac{\gamma'}{t}
  \end{eqnarray*}
  for a factor $t\in(0,1)$,
  that means,
  we define $\AUG(\gamma') := \frac{\gamma'}{t}$.
  We call $\frac{1}{t}\label{def-t-inline}$
  the \emph{augmentation factor} for the region of uncertainty.
  Together this leads to the implications
  \begin{eqnarray*}
    \text{$(I)$ and $(III)$}
    \FORMSEP&\Rightarrow&\FORMSEP
    \tau_0 \;\;\le\;\; t \cdot \min_{1\le i\le k} \gamma_i \\
    \text{$(II)$ and $(III)$}
    \FORMSEP&\Rightarrow&\FORMSEP
    \tau_0 \;\;\le\;\; (1-t) \cdot \min_{1\le i\le k} \gamma_i \\
    \text{and consequently}
    \FORMSEP&\Rightarrow&\FORMSEP
    \tau_0 \;\;\stackrel{(IV)}{\le}\;\; \min\left\{t,1-t\right\} \cdot \min_{1\le i\le k} \gamma_i
  \end{eqnarray*}
  Furthermore,
  we demand that $\tau_0$ is a power of $2$
  which turns $(IV)$ into the equality
  \begin{eqnarray*}
    \tau_0
      &\stackrel{(V)}{=}&
      2^{\left\lfloor
           \log_2 
	   \left(
	     \min\left\{t,1-t\right\} \cdot \min_{1\le i\le k} \gamma_i
	   \right)
         \right\rfloor}.
  \end{eqnarray*}
  Due to
  Formula~(\ref{def-tau-inline})
  in Definition~\ref{def-grid-points}
  we also know that
  \begin{eqnarray*}
    \tau_0
      &\stackrel{(VI)}{=}& 2^{\E-\LG-1}.
  \end{eqnarray*}
  Therefore we can deduce $\LG$ from $(V)$ and $(VI)$ as
  \begin{eqnarray}\label{for-def-lgrid}
    \LG(\gamma) &:=&
      \E
      -1
      - \left\lfloor
          \log_2 \left(
		    \min\left\{t,1-t\right\} \cdot \min_{1\le i\le k} \gamma_i
	         \right)
        \right\rfloor.
  \end{eqnarray}
  As an example,
  for $t=\frac{1}{2}$
  we obtain
  $\LG(\gamma) = \E - \lfloor \log_2 \min_{1\le i\le k} \gamma_i \rfloor$.
  We refine the notion of a predicate description.
\begin{definition}\label{def-predi-con-3}
  We extend Definition~\ref{def-predi-con-2} and
  call $\PREDGt$ a \emph{predicate description} if:
  8. $t\in(0,1)$.
\end{definition}
  Now we are able to summarize the construction above.
\begin{theorem}\label{theo-validation}
  Let $\PREDGt$ be a predicate description.
  Then
  \begin{eqnarray}\label{for-grid-impact}
    \frac{\mu\left(R_{\gamma}(\x)\right)}{\mu\left(U_\delta(\x)\right)}
    \ge
    \frac{\CARDI{R_{t\gamma}(\x)\RGL}}{\CARDI{U_\delta(\x)\RGL}}
  \end{eqnarray}
  for all precisions $L\ge\LG(\gamma)$
  where
  $\LG$ is defined in Formula~(\ref{for-def-lgrid}).
\end{theorem}
\begin{remark}
  We add some remarks on the grid unit condition.
  
  1.
  Unequation~(\ref{for-grid-impact}) guarantees that
  the success probability for grid points is at least
  the success probability that is derived from the volumes of areas.
  This justifies the analysis in real space at last.

  2.
  Be aware that
  the grid unit condition is a fundamental constraint:
  It does not depend on the function that realize the predicate,
  the dimension of the (projected or full) perturbation area,
  the perturbation parameter or
  the critical set.
  The threshold $\LG$
  mainly depends on
  the augmentation factor $\frac{1}{t}$ and
  $\gamma$.
  In particular
  we observe that an additional bit of the precision is sufficient
  to fulfill the grid unit condition for $\frac{\gamma}{2}$, i.e.
  \begin{eqnarray*}
    \LG\left(\frac{\gamma}{2}\right)
      &=&
        \LG(\gamma) + 1.
  \end{eqnarray*}

  3.
  We have defined the region of uncertainty $R_f$
  by means of axis-parallel boxes $U_\gamma(c)$
  for $c\in C_f$
  in Definition~\ref{def-region-uncertainty}.
  If $R_f$ is defined in a different way,
  we must appropriately adjust
  the derivation of $\LG$ in this section.

  4.
  We observe that the grid unit condition solves
  Issue~8 from Section~\ref{sec-succ-prob}.
  Now we reconsider the example in Figure~\ref{fig-volume-location-shape}
  on page~\pageref{fig-volume-location-shape}.
  We observe
  that the grid unit in Picture~(a) fulfills the grid unit condition
  whereas the condition fails in Pictures~(b) and~(c).
  Obviously, $\tau\gg\gamma$ in the latter cases.
  \hfill$\bigcirc$
\end{remark}

\subsection{Overview: Prerequisites of the Validation}
\label{sec-over-validation}

  It is important to see that the analysis
  \emph{must} reflect the behavior of
  the underlying floating-point implementation
  of a controlled perturbation algorithms
  to gain a meaningful result.
  Only for the purpose to achieve this goal,
  we have introduced some principles
  that we summarize below.
  The items are meant to be reminders, not explanations.
  \begin{itemize}
    \item
      We guarantee that the perturbed input lies on the grid $\G$.
    \item
      We analyze an augmented region of uncertainty.
    \item
      The region of uncertainty
      is a union of axis-parallel boxes and, especially,
      intervals in the 1-dimensional case.
      There are lower bounds on the measures of the box.
    \item
      The grid unit condition is fulfilled.
    \item
      We do not exclude isolated points,
      unless we can prove that their exclusion does not change the
      floating-point probability.
      It is always safe to exclude environments of points.
    \item
      We analyze $\eta$ runs of $\AG$ at a time
      (see Section~\ref{sec-meth-distri-prob}).
  \end{itemize}
  With this principles at hand
  we are able to derive a valid analysis in real space.

\section[Necessary Conditions for the Analysis of Functions (Interface)]
{Necessary Conditions for the Analysis of Functions}
\label{sec-nec-con-func}

  The method of quantified relations,
  which is introduced in the next section,
  actually performs the analysis of real-valued functions.
  Here we prepare its applicability.
  In Section~\ref{sec-analyzability-func}
  we present three necessary conditions:
  the \emph{region-, value- and safety-suitability.}
  Together these conditions are also sufficient to apply the method.
  Because these conditions are deduced in the first stage
  of the function analysis
  (see Section~\ref{sec-direct-approach}--\ref{sec-guards-safetybounds})
  and are used in the second stage
  (see Section~\ref{sec-meth-quan-rela}),
  we also refer to them as the interface between the two stages
  (see Figure~\ref{fig-illu-interface}).
  \emph{This is the first time
  that we precisely define the prerequisites of the function analysis.}
  The definitions are followed by an example.
  In Section~\ref{sec-over-func-prop}
  we summarize all function properties.
\begin{figure}[h]\centering
  \includegraphics[width=.95\columnwidth]{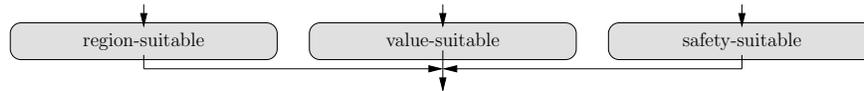}
  \caption{The interface between the two stages
    of the analysis of functions.}
  \label{fig-illu-interface}
\end{figure}

\subsection{Analyzability of Functions}
\label{sec-analyzability-func}

  Here we define and explain the three function properties
  that are necessary for the analysis.
  Their associated bounding functions constitute
  the interface between the two stages.
  Informally, the properties have the following meanings:
\begin{itemize}
  \item
    We can reduce the volume of the region of uncertainty
    to any arbitrarily small value (region-suitability).
  \item
    There are positive and finite limits on the absolute value of $f$
    outside of the region of uncertainty (value-suitability).
  \item
    We can reduce the rounding error in the floating-point evaluation of $f$
    to any arbitrarily small value
    (safety-suitability).
\end{itemize}

\subsection*{The region-suitability}

  The region-suitability is a geometric condition
  on the neighborhood of the critical set.
  We demand that
  we can adjust the volume of the region of uncertainty
  to any arbitrarily small value.
  For technical reasons we need an invertible bound.
\begin{definition}[region-suitable]\label{def-region-suit}
  Let $\PREDL$ be a predicate description.
  We call $f$ \emph{region-suitable}
  if the critical set of $f$ is either empty
  or if there is an invertible upper-bounding function\footnote{%
    Instead of $\nu_f$
    we can also use its complement $\chi_f$.
    See the following Remark~\ref{rem-def-region-suit}.4 for details.}
  \begin{eqnarray*}\label{def-nu-inline}
    \nu_f : \GAL \to\R_{>0}
  \end{eqnarray*}
  on the volume of the region of uncertainty
  that has the property:
  For every $p\in(0,1)$ there is $\gamma\in\GAL$ such that
  \begin{eqnarray}\label{for-defregsuit}
     \frac{\mu(R_\gamma(\x))}{\mu(U_\delta(\x))}
     &\le& \frac{\nu_f(\gamma)}{\mu(U_\delta(\x))}
     \;\;\le\;\; (1-p)
  \end{eqnarray}
  for all $\x\in A$.
\end{definition}
\begin{remark}\label{rem-def-region-suit}
  We add several remarks on the definition above.

  1. Region-suitability is related to the region-condition
    in the following way:
    The criterion for region-suitability results from the
    replacement of $\mu(R_\gamma(\x))$
    in Formula~(\ref{for-def-region-const})
    with a function $\nu_f$.
    This changes the region-condition in Definition~\ref{def-verifiable} %
    into a quantitative bound.

  2. Of course,
    controlled perturbation cannot work
    if the region of uncertainty covers the entire perturbation area of $\x$.
    We have said that
    we consider $\gamma\in\GALG$
    for a ``sufficiently'' small $\hat\gamma\in\R_{>0}^k$.
    That means formally,
    we postulate
    $\nu(\hat\gamma) \ll \mu(U_\delta(\x))$.
    To keep the notation as plain as possible,
    we are aware of this fact
    and do not make this condition explicit in our statements.

  3. The invertibility of the bonding function $\nu_f$ is essential for
    the method of quantified relations
    as we see in the proof of Theorem~\ref{theo-quan-rela}.
    There it is used to deduce
    the parameter $\gamma$
    from the volume of the region of uncertainty---with
    the exception of an empty critical set
    which does not imply any restriction on $\gamma$.

  4a. The function
    $\nu_f$ provides an upper bound on the volume of the
    region of uncertainty within the perturbation area of $\x$.
    Sometimes it is more convenient to consider its complement
    \begin{eqnarray}
      \label{def-chi-region-suit}
      \chi_f(\gamma)
      &:=&
       \mu\left(U_\delta(\x)\right) \, - \, \nu_f(\gamma).
    \end{eqnarray}
    The function $\chi_f(\gamma)$
    provides a lower bound on the volume of
    the \emph{region of provable fp-safe inputs.}

  4b. The special case $\nu_f\equiv 0$
    corresponds to the special case
    $\chi_f\equiv \mu\left(U_\delta(\x)\right)$.
    Then the critical set is empty and
    there is no region of uncertainty.
    This implies that $\varphi_f(\gamma)$
    can also be chosen as a constant function
    (see the value-suitability below).

  4c. Based on Formula~(\ref{def-chi-region-suit}),
    we can demand the existence of an invertible function
    $\chi_f:\GAL\to\R_{>0}$
    instead of $\nu_f$
    in the definition of region-suitability.
    That means,
    either $\chi_f\equiv\mu\left(U_\delta(\x)\right)$
    or $\chi_f:\GAL\to\R_{>0}$ in an invertible function.

  5. We make the following observations about region-suitability:
    (a) If the critical set is finite, $f$ is region-suitable.
    (b) If the critical set contains an open set, $f$ cannot be region-suitable.
    (c) If the critical set is a set of measure zero,
        it does not imply that $f$ is region-suitable.
	Be aware that these properties are not equivalent:
	If $f$ is region-suitable,
	the critical set is a set of measure zero.
	But a critical set of measure zero
	does not necessarily imply that $f$ is region-suitable:
	In topology we learn that $\Q$ is dense\footnote{%
	  Topology: ``$\Q$ is dense in $\R$''
	  means that $\overline{\Q}=\R$.
	  For example, see Jänich~\cite[p.~63]{J01}.}
	in $\R$;
	hence any open $\varepsilon$-neighborhood of $\Q$ equals $\R$.
	In set theory we learn that\footnote{%
	  Set Theory:
	  Cardinalities of (infinite) sets are denoted by $\aleph_i$.
	  For example, see Deiser~\cite[162ff]{D04}.}
	$|\Q|=\aleph_0 < 2^{\aleph_0}=|\R|$;
	hence $f$ cannot be region-suitable if the critical set
	is (locally) ``too dense.''
  \hfill$\bigcirc$
\end{remark}

\subsection*{The inf-value-suitability}

  The inf-value-suitability is a condition on the behavior of
  the function $f$.
  We demand that
  there is a positive lower bound on the absolute value of $f$
  outside of the region of uncertainty.
\begin{definition}[inf-value-suitable]\label{def-value-suit}
  Let $\PREDL$ be a predicate description.
  We call $f$ \emph{(inf-)value-suitable}
  if there is a lower-bounding function
  \begin{eqnarray*}\label{def-varphi-inline}
    \varphiinf : \GAL \to\Rplus
  \end{eqnarray*}
  on the absolute value of $f$
  that has the property:
  For every $\gamma\in\GAL$, we have
  \begin{eqnarray}
      \varphiinf(\gamma)
        &\le& |f(x)|
  \end{eqnarray}
  for all $x\in \bar{U}_\delta(\x)\setminus R_\gamma(\x)$ and
  for all $\x\in A$.
\end{definition}
  We extend this definition
  by an upper bound on the absolute value of $f$
  in Section~\ref{sec-rational-function}
  and call this property sup-value-suitability;
  until then
  we call the inf-value-suitability simply the value-suitability and
  also write $\varphi_f$ instead of $\varphiinf$.
  The criterion for value-suitability results from
  the replacement of the constant $\varphi$
  in Formula~(\ref{for-def-veri-minval})
  with the bounding function $\varphi_f$.
  This changes the existence statement of
  Lemma~\ref{lem-veri-minval} into a quantitative bound.

\subsection*{The inf-safety-suitability}

  The inf-safety-suitability is a condition on the error analysis
  of the floating-point evaluation of $f$.
  We demand that
  we can adjust the rounding error in the evaluation of $f$
  to any arbitrarily small value.
  For technical reasons we demand an invertible bound.\footnote{%
    We leave the extension to non-invertible or discontinuous
    bounds to the reader;
    we do not expect that there is any need in practice.}
\begin{definition}[inf-safety-suitable]\label{def-safety-suit}
  Let $\PRED$ be a predicate description.
  We call $f$ \emph{(inf-)safety-suitable}
  if there is an injective fp-safety bound $\Sinf(L):\N\to\R_{>0}$
  that fulfills the safety-condition
  in Formula~(\ref{for-safety-cond})
  and if
  \begin{eqnarray*}
    \Sinf^{-1} : (0,\Sinf(1)] \to \R_{>0}.
  \end{eqnarray*}
  is a strictly monotonically decreasing real continuation of its inverse.
\end{definition}
  We extend the definition by sup-safety-suitability
  in Section~\ref{sec-rational-function};
  until then
  we call the inf-safety-suitability simply the safety-suitability.

\subsection*{The analyzability}

  Based on the definitions above,
  we next define analyzability,
  relate it to verifiability and
  give an example for the definitions.
\begin{definition}[analyzable]\label{def-analyzable}
  We call $f$ \emph{analyzable}
  if it is region-, value- and safety-suitable.
\end{definition}
\begin{lemma}
  \label{lem-ana-is-veri}
  Let $f$ be analyzable. Then $f$ is verifiable.
\end{lemma}
\begin{proof}
  If $f$ is analyzable, $f$ is especially region-suitable.
  Then the region-condition in Definition~\ref{def-verifiable}
  is fulfilled because of the bounding function $\nu_f$.
  In addition $f$ must also be safety-suitable.
  Then the safety-condition in Definition~\ref{def-verifiable}
  is fulfilled because of the bounding function $\Sinf$.
  Together both conditions imply that $f$ is verifiable.
  \qed
\end{proof}
  We support the definitions above with the example
  of univariate polynomials.
  Because we refer to this example later on,
  we formulate it as a lemma.
\begin{lemma}\label{lem-uni-poly-ana}
  Let $f$ be the univariate polynomial
  \begin{eqnarray}\label{for-unipolyrep}
    f(x) &=& a_d \cdot x^d + a_{d-1} \cdot x^{d-1} +
               \ldots + a_1 \cdot x + a_0
  \end{eqnarray}
  of degree $d$
  and let $\PREDL$ be a predicate description for $f$.
  Then $f$ is analyzable on $\U_\delta(A)$
  with the following bounding functions
  \begin{eqnarray}
    \nu_f(\gamma) &:=& 2d\gamma \label{for-unipolynu} \\
    \varphi_f(\gamma) &:=& |a_d| \cdot \gamma^d \label{for-unipolyphi} \\
    \Sinf(L) &:=& (d+2) \, \max_{1\le i\le d} |a_i| \;\, 2^{\E(d+1)+1-L}. \nonumber
  \end{eqnarray}
\end{lemma}
\begin{proof}
  For a moment
  we consider the complex continuation of the polynomial,
  i.e.\ $f\in\C[z]$.
  Because of the fundamental theorem of algebra
  (e.g., see Lamprecht~\cite{L93}),
  we can factorize $f$ in the way
  \begin{eqnarray*}
    f(z) &=& a_d \cdot \prod_{i=1}^d (z-\zeta_i)
  \end{eqnarray*}
  since $f$ has $d$ (not necessarily distinct) roots $\zeta_i\in\C$.
  Now let $\gamma\in\R_{>0}$.
  Then we can lower bound the absolute value of $f$ by
  \begin{eqnarray*}
    |f(z)| &\ge& |a_d| \cdot \gamma^d
  \end{eqnarray*}
  for all $z\in\C$
  whose distance to every (complex) root of $f(z)$ is at least $\gamma$.
  Naturally, the last estimate
  is especially true for real arguments $x$
  whose distance to the orthogonal projection of the complex roots $\zeta_i$
  onto the real axis is at least $\gamma$.
  So we set the critical set to\footnote{%
    Complex Analysis:
    The function $\RE(z)$ maps a complex number $z$ to its real part.
    For example, see Fischer et al.~\cite{FL05}.}
  $C_f(\x):=\{ \RE(\zeta_i) : 1\le i\le d\} \cap \U_\delta(\x)$.
  This validates the bound $\varphi_f$
  and implies that $f$ is value-suitable.

  Furthermore,
  the size of $C_f$ is upper-bounded by $d$ for all $\x\in A$.
  This validates the bound $\nu_f$.
  Because $\nu_f$ is invertible,
  $f$ is region-suitable.

  The bounding function $\Sinf(L)$ is proven
  in Corollary~\ref{col-unipolysafety}
  in Section~\ref{sec-guards-safetybounds}.
  Because $\Sinf(L)$ is invertible, $f$ is also safety-suitable.
  As a consequence, $f$ is \emph{analyzable} with the given bounds.
  \qed
\end{proof}
  We admit that we have chosen a quite simple example.
  But a more complex example would have been a waste of energy since
  we present
  \emph{three general approaches to derive the bounding functions
  for the region- and value-suitability}
  in Sections~\ref{sec-direct-approach},~\ref{sec-bottom-up}
  and~\ref{sec-top-down}.
  That means, for more complex examples we use more convenient tools.
  A well-known approach to derive the bounding function for the safety-bound
  is given in Section~\ref{sec-guards-safetybounds}.

\subsection{Overview: Function Properties}
\label{sec-over-func-prop}

  At this point,
  we have introduced all properties that are necessary
  to precisely characterize functions
  in the context of the analysis.
  So let us take a short break to see what we have defined and related so far.
  We have summarized the most important implications
  in Figure~\ref{fig-func-prop-implic}.
\begin{figure}[t]\centering
  \includegraphics[width=.95\columnwidth]{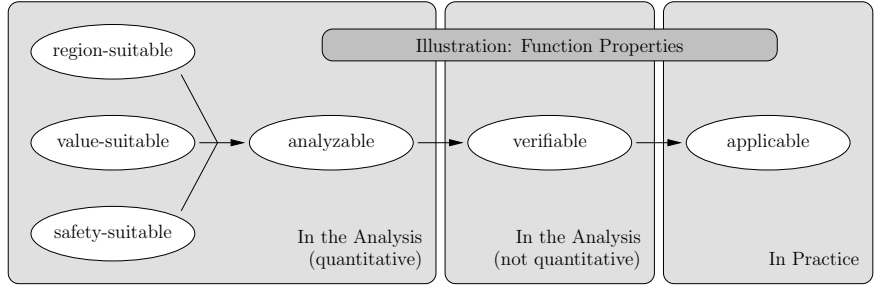}
  \caption{The illustration summarizes the implications of
    the various function properties that we have defined in this paper.
    A function that is region-, value- and safety-suitable
    at the same time is also analyzable (see Definition~\ref{def-analyzable}).
    An analyzable function is also verifiable (see Lemma~\ref{lem-ana-is-veri}).
    And a verifiable function is also applicable
    (see Lemma~\ref{lem-veri-is-app}).}
  \label{fig-func-prop-implic}
\end{figure}
  Controlled perturbation is \emph{applicable}
  to a certain class of functions.
  But only for a subset of those functions,
  we can actually \emph{verify} that
  controlled perturbation {works in practice}---without
  the necessity, or even ability, to analyze their performance.
  We remember
  that no condition on the absolute value is needed
  for verifiability
  because it is not a quantitative property.
  A subset of the verifiable functions
  represents the set of \emph{analyzable} functions in a quantitative sense.
  For those functions there are \emph{suitable bounds} on
  the maximum volume of the region of uncertainty,
  on the minimum absolute value outside of this region
  and on the maximum rounding error.
  In the remaining part of the paper,
  we are only interested in the class of analyzable functions.

\section[The Method of Quantified Relations (2nd Stage)]
{The Method of Quantified Relations}
\label{sec-meth-quan-rela}

  The method of quantified relations actually performs
  the function analysis
  in the second stage.
  The component and its interface are illustrated in Figure~\ref{fig-ana-qr}.
  We introduce the method in Section~\ref{sec-meth-qr-pres}.
  Its input consists of three bounding functions
  that are associated with the three suitability properties
  from the last section.
  The applicability does not depend on any other condition.
  The method provides general instructions
  to relate the three given bounds.
  The prime objective is to derive a relation between
  the probability of a successful floating-point evaluation
  and the precision of the floating-point arithmetic.
  More precisely,
  the method provides a precision function $L(p)$
  or a probability function $p(L)$.
  \emph{This is the first presentation
  of step-by-step instructions for the second stage
  of the function analysis.}
  An example of its application follows
  in Section~\ref{sec-meth-qr-ex}.
  \begin{figure}[t]\centering
    \includegraphics[width=.45\columnwidth]{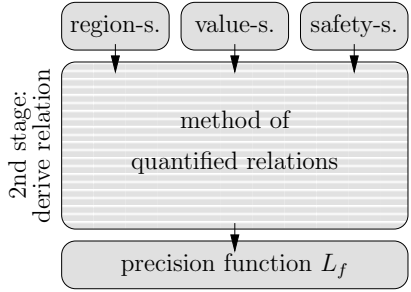}
    \caption{The method of quantified relations
      and its interface.}
    \label{fig-ana-qr}
  \end{figure}

\subsection{Presentation}
\label{sec-meth-qr-pres}

  There are no further prerequisites than
  the three necessary suitability properties
  from the last section.
  Therefore we can immediately state the main theorem
  of this section whose proof contains the
  method of quantified relations.
\begin{theorem}[quantified relations]\label{theo-quan-rela}
  Let $\PREDL$ be a predicate description and
  let $f$ be analyzable.
  Then there is a method
  to determine a precision function $L_f:(0,1)\to\N$
  such that the guarded evaluation of $f$
  at a randomly perturbed input
  is successful with probability at least $p\in(0,1)$
  for every precision $L\in\N$ with $L\ge L_f(p)$.
\end{theorem}
\begin{proof}
  We show in six steps how we can determine
  a precision function $L_f(p)$ which has the property:
  If we use a floating-point arithmetic with precision $L_f(p)$
  for a given $p\in(0,1)$,
  the evaluation of $f(x)\RG$ is guarded
  with success probability of at least $p$
  for a randomly chosen $x\in\U_\delta(\x)\RG$
  and for any $\x\in A$.
  An overview of the steps is given in Table~\ref{tab-steps-qr}.
  Usually we begin with Step~1.
  However, there is an exception:
  In the special case that $\nu_f\equiv 0$,
  we know that the bounding function $\varphi$ is constant,
  see Remark~\ref{rem-def-region-suit}.4 for details.
  Then we just skip the first four steps and begin with Step~5.

  \begin{table}
  \centerline{\fbox{
  \begin{minipage}{0.9\columnwidth}\medskip
    Step 1: relate probability with volume of region of uncertainty (define ${\varepsilon_{\nu}}$)\\
    Step 2: relate volume of region of uncertainty with distances (define $\gamma$)\\
    Step 3: relate distances with floating-point grid (choose $t$)\\
    Step 4: relate new distances with minimum absolute value (define $\varphi$)\\
    Step 5: relate minimum absolute value with precision (define $\LS$)\\
    Step 6: relate $\LS$ with $\LG$ (define $\LG$ and $L_f$)\medskip
  \end{minipage}}}
    \caption{Instructions for performing the method of quantified relations.}
    \label{tab-steps-qr}
  \end{table}
{Step 1 (define ${\varepsilon_{\nu}}$).}
  We derive
  an upper bounding function ${\varepsilon_{\nu}(p)}$ on the volume of the
  augmented region of uncertainty
  from the success probability $p$
  in the way
  \begin{eqnarray}\label{for-epsilon-nu}
    {\varepsilon_{\nu}(p)}
     &:=& (1-p)\cdot\mu(U_\delta) \\
     &=& (1-p)\cdot\prod_{i=1}^k (2\delta_i). \nonumber
  \end{eqnarray}
  That means,
  a randomly chosen point $x\in U_\delta(\x)$
  lies inside of a given region of volume $\varepsilon_\nu(p)$
  with probability at least $p$.
  Be aware
  that we argue about the \emph{real space} in this step.

{Step 2 (define $\gamma$).}
  We know that
  there is $\gamma\in\R_{>0}^k$ that fulfills
  the region-condition in Definition~\ref{def-verifiable} %
  because $f$ is verifiable.
  Since $f$ is even region-suitable,
  we can also determine such $\gamma\in\GAL$
  by means of the inverse of the bounding function $\nu_f$.
  The existence and invertibility of $\nu_f$
  is guaranteed by Definition~\ref{def-region-suit}. %
  Hence we define the function
  \begin{eqnarray}\label{for-gamma-epsilon-nu}
    \gamma(p) &:=& \nu_f^{-1}({\varepsilon_{\nu}(p)}) \in\GAL.
  \end{eqnarray}
  We remember that there is an alternative definition of the
  region-suitability which we have mentioned
  in Remark~\ref{rem-def-region-suit}.4.
  Surely it is also possible to use the bounding function $\chi_f$
  instead of $\nu_f$ in the method of quantified relations directly;
  the alternative Steps~$1'$ and~$2'$
  are introduced in Remark~\ref{rem-meth-quan-rela}.2.

{Step 3 (choose $t$).}
  We aim at a result
  that is valid for floating-point arithmetic
  although we base the analysis on real arithmetic
  (see Section~\ref{sec-validation}).
  We choose\footnote{%
    The analysis works for any choice.
    However, finding the best choice is an optimization problem.}
  $t\in(0,1)$
  and define $R_{t\gamma}$ as
  the normal sized region of uncertainty.
  Due to Theorem~\ref{theo-validation},
  the probability that
  a random point $x\in U_\delta(\x)\RG$
  lies inside of $R_{t\gamma}(\x)\RG$
  is smaller than
  the probability that a random point $x\in U_\delta(\x)$
  lies inside of $R_\gamma(\x)$.
  Consequently,
  if a randomly chosen point
  lies outside of the augmented region of uncertainty with probability $p$,
  it lies outside of the normal sized region of uncertainty with probability
  at least $p$.
  Our next objective is to guarantee
  a floating-point safe evaluation
  outside of the
  \emph{normal sized} region of uncertainty.

{Step 4 (define $\varphi$).}
  Now we want to determine the minimum absolute value
  outside of the region of uncertainty
  $R_{t\gamma}(\x)$.
  We have proven in Lemma~\ref{lem-veri-minval}
  that a positive minimum exists.
  Because $f$ is value-suitable,
  we can use the bounding function $\varphi_f$
  for its determination
  (see Definition~\ref{def-value-suit}). %
  That means, we consider
  \begin{eqnarray*}
    \varphi(p) &:=& \varphi_f(t\cdot\gamma(p)).
  \end{eqnarray*}

{Step 5 (define $\LS$).}
  So far we have fixed the region of uncertainty and
  have determined the minimum absolute value
  outside of this region.
  Now we can use the safety-condition
  from Definition~\ref{def-verifiable} %
  to determine a precision $\LS$
  which implies fp-safe evaluations outside of $R_{t\gamma}$.
  That means,
  we want that Formula~(\ref{for-def-safety-const})
  is valid for every $L\in\N$ with $L\ge \LS$.
  Again we use the property that $f$ is analyzable
  and use the inverse of the fp-safety bound $\Sinf^{-1}$
  in Definition~\ref{def-safety-suit} %
  to deduce the precision from the minimum absolute value $\varphi(p)$ as
  \begin{eqnarray}
    \LS(p)
    &=&
      \left\lceil
        \Sinf^{-1} \left(
	  \varphi_f \left(
	    t\cdot\nu_f^{-1}\left({\varepsilon_{\nu}\left(p\right)}\right)
	\right)\right)
      \right\rceil.
        \label{for-def-lsafe}
  \end{eqnarray}

{Step 6 (define $\LG$ and $L_f$).}
  We numerate the component functions of $\nu_f^{-1}$
  in the way
  $\nu_f^{-1}(\varepsilon)=(\nu_{1}^{-1}(\varepsilon),\ldots,\nu_{k}^{-1}(\varepsilon))$.
  Then we deduce the bound $\LG$ from
  Formula~(\ref{for-def-lgrid}) and
  Formula~(\ref{for-gamma-epsilon-nu}) in the way
  \begin{eqnarray}\label{for-lgrid-in-proof}
    \LG(p) &:=&
      \E
      -1
      - \left\lfloor
          \log_2 \left(
		    \min\left\{t,1-t\right\}
		    \cdot \min_{1\le i\le k} \nu_i^{-1}(\varepsilon_\nu(p))
	         \right)
        \right\rfloor.
  \end{eqnarray}
  Finally we define the \emph{precision function $L_f(p)$}
  pointwise as
  \begin{eqnarray}\label{def-lf-final}
    L_f(p) &:=& \max\left\{\LS(p), \LG(p)\right\}.
  \end{eqnarray}
  Due to the used estimates,
  any precision $L\in\N$ with $L\ge L_f(p)$ is a solution.
  \qed
\end{proof}
\begin{remark}\label{rem-meth-quan-rela}
  We add some remarks on the theorem above.

  1.
  It is important to see that
  $\LS$ is derived from the \emph{volume} of $R_f$
  and is based on the region- and safety condition in
  Definition~\ref{def-verifiable}
  whereas $\LG$ is derived from the \emph{narrowest width} of $R_f$
  and is based on the grid unit condition in Section~\ref{sec-relate-tau-gamma}.
  Of course,
  $L_f(p)$ must be large enough to fulfill both constraints.

  2.
  As we have seen,
  we can also use the function $\chi_f$ to define the region-suitability
  in Definition~\ref{def-region-suit}.
  Therefore we can modify the first two steps of
  the method of quantified relations as follows:
\newline
  Step $1'$ (define $\varepsilon_{\chi}$).
  Instead of Step~1,
  we define a bounding function ${\varepsilon_{\chi}(p)}$
  on the volume of the complement of $R_f$
  from the given success probability $p$.
  That means, we replace Formula~(\ref{for-epsilon-nu}) with
  \begin{eqnarray*}
    {\varepsilon_{\chi}(p)}
      &:=& p \cdot \mu(U_\delta) \\
      &=& p \cdot \prod_{i=1}^k \left(2\delta_i\right).
  \end{eqnarray*}
  Step $2'$ (define $\gamma$).
  Then we can determine $\gamma(p)$
  with the inverse of the bounding function $\chi_f$.
  That means, we replace Formula~(\ref{for-gamma-epsilon-nu}) with
  \begin{eqnarray*}
    \gamma(p) &:=& \chi_f^{-1}({\varepsilon_{\chi}(p)}) \in\GAL
  \end{eqnarray*}
  which finally changes Formula~(\ref{for-def-lsafe}) into
  \begin{eqnarray*}
    \LS(p)
    &=&
      \left\lceil
        \Sinf^{-1} \left(
	  \varphi_f \left(
	    t\cdot
	    \chi_f^{-1}\left(
	      {\varepsilon_{\chi}\left(p
	      \right)}\right)
	\right)\right)
      \right\rceil.
  \end{eqnarray*}
  We make the observation that
  these changes do not affect the correctness of
  the method of quantified relations.

  3.
  It is important to see
  that the method of quantified relations
  is absolutely independent of the derivation of the bounding functions
  which are associated with the necessary suitability properties.
  Especially in Step~2,
  $\gamma$ is determined solely by means of the function $\nu^{-1}$.
  We illustrate this generality with the examples in
  Figure~\ref{fig-crit-set-nu}.
  \begin{figure}[h]\centering
    \includegraphics[width=.95\columnwidth]{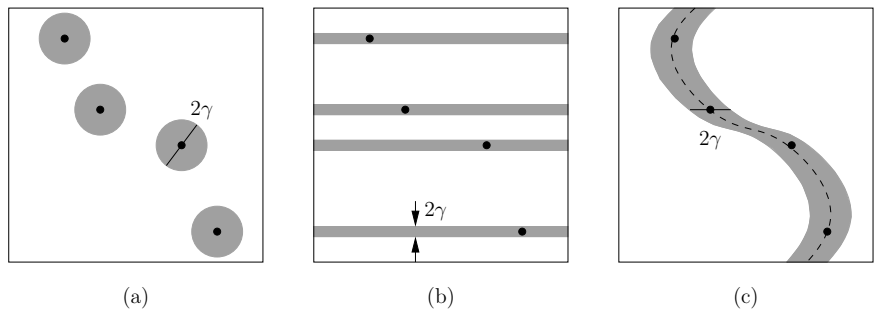}
    \caption{Visualization of $\nu^{-1}(\varepsilon_\nu)$ in Step~2
      of the method of quantified relations.}
    \label{fig-crit-set-nu}
  \end{figure}
  The three pictures show different regions of uncertainty for
  the \emph{same} critical set
  and the \emph{same} volume $\varepsilon_\nu$.
  This is because the region of uncertainties
  result from different functions $\nu^{-1}$.
  We could say that the function $\nu^{-1}$ ``knows'' how to distribute
  the region of uncertainty around the critical set
  because of its definition in the first stage
  of the analysis. For example:
  (a) as local neighborhoods,
  (b) as axis-parallel stripes, or
  (c) as neighborhoods of local minima of
  $f$ (the dashed line).
  (We remark that case (c) presumes that $f$ is continuous.)
  Naturally, different functions $\nu^{-1}$ lead to different values of $\gamma$
  as is illustrated in the pictures.
  Be aware that
  the method of quantified relations
  itself is absolutely independent of the \emph{derivation} of $\nu$
  and especially independent of the \emph{approach}
  by which $\nu$ is derived.
  (We present three different approaches soon.)

  4.
  If $f$ is analyzable and $\varphi_f$ invertible,
  we can also derive the success probability $p$
  from a precision $L$.
  We observe that the function $\varepsilon_\nu$
  in Formula~(\ref{for-epsilon-nu}) is always invertible.
  Therefore
  we can transform Formula~(\ref{for-def-lsafe})
  and~(\ref{for-lgrid-in-proof}) into
  \begin{eqnarray*}
    \pinf(L)\label{for-pinf-inline}
    &:=&
      \varepsilon_\nu^{-1}
        \left( \nu_f
          \left( \frac{1}{t} \cdot \varphi_f^{-1}
            \left( \Sinf(L)
      \right) \right) \right) \\
    \pgrid(L)\label{for-pgrid-inline}
    &:=&
      \varepsilon_\nu^{-1} \left(
        \nu_*\left(
	  \frac{2^{-L+\E-1}}
	    {\min\{t,1-t\}}
	\right)
      \right),
  \end{eqnarray*}
  respectively,
  where $\nu_*^{-1}$
  is the least growing component function of $\nu_f^{-1}$
  and $\nu_*$ is the inversion of $\nu_*^{-1}$.
  This leads to
  the (preliminary) \emph{probability function\label{def-prob-func-inline}}
  $p_f:\N\to(0,1)$,
  \begin{eqnarray*}
    p_f(L) &:=& \min \left\{ \pinf(L), \pgrid(L) \right\}
  \end{eqnarray*}
  for parameter $t\in(0,1)$.
  We develop the final version of the probability function
  in Section~\ref{sec-range-error-analysis}.
  Self-evidently
  we can also derive appropriate bounding functions for
  $\chi$ instead of $\nu$ (see Remark~2).
  \hfill$\bigcirc$
\end{remark}

\subsection{Example}
\label{sec-meth-qr-ex}

  To get familiar with the usage of
  the method of quantified relations,
  we give a detailed application in the proof
  of the following lemma.
\begin{lemma}%
  Let $f$ be a univariate polynomial of degree $d$
  as shown in Formula~(\ref{for-unipolyrep})
  and let $\PREDL$ be a predicate description.
  Then we obtain for $f$:
  \begin{eqnarray}\label{for-theo-uni-lsafe}
    \LS(p) &:=& \left\lceil -d\LOG(1-p) \; + \; \CU \right\rceil
  \end{eqnarray}
  where
  \begin{eqnarray*}
      \CU &:=&
      \LOG\frac
      {(d+2) \cdot \max_{1\le i\le d}|a_i| \cdot 2^{\E (d+1) +1}}
      {|a_d| \cdot (t\delta/d)^d}.
  \end{eqnarray*}
\end{lemma}
\begin{proof}
  The polynomial $f$ is analyzable
  because of Lemma~\ref{lem-uni-poly-ana}.
  Therefore we can determine $\LS$
  with the first 5 steps of the method of quantified relations
  (see Theorem~\ref{theo-quan-rela}).\\
  Step 1:
  Since the perturbation area $U_\delta(\x)$ is an interval of length $2\delta$,
  the region of uncertainty has a volume of at most
  \begin{eqnarray*}
    {\varepsilon_{\nu}(p)} &:=& 2\delta (1-p).
  \end{eqnarray*}
  Step 2:
  Next we deduce $\gamma$ from the inverse of
  the function in Formula~(\ref{for-unipolynu}),
  that means,
  from $\nu_f^{-1}(\varepsilon)=\frac{\varepsilon}{2d}$.
  We obtain
  \begin{eqnarray*}
    \gamma(p) &:=& \nu^{-1}_f({\varepsilon_{\nu}(p)})
      \;\; = \;\; \frac{{\varepsilon_{\nu}(p)}}{2d}
      \;\; = \;\; \frac{\delta(1-p)}{d}.
  \end{eqnarray*}
  Step 3:
  We choose $t\in(0,1)$.
\newline
  Step 4:
  Due to Formula~(\ref{for-unipolyphi}),
  the absolute value of $f$
  outside of the region of uncertainty
  is lower-bounded by the function
  \begin{eqnarray*}
    \varphi(p)
    &:=&
      |a_d| \cdot (t\cdot\gamma(p))^d \\
    &=&
      |a_d| \cdot \left( \frac{t\delta(1-p)}{d} \right)^d.
  \end{eqnarray*}
  Step 5:
  A fp-safety bound $\Sinf$
  is provided by Corollary~\ref{col-unipolysafety}
  in Formula~(\ref{for-sfl-univariate}).
  The inverse of this function
  at $\varphi(p)$ is
  \begin{eqnarray*}
    \Sinf^{-1}(\varphi(p))
    &=& \LOG
      \frac
      {(d+2) \cdot \max_{1\le i\le d}|a_i| \cdot 2^{\E (d+1) +1}}
      {\varphi(p)}.
  \end{eqnarray*}
  Due to Formula~(\ref{for-def-lsafe}),
  this leads to
  \begin{eqnarray*}
    \LS(p)
    &:=& \left\lceil \Sinf^{-1}(\varphi(p)) \right\rceil\\
    &=& \left\lceil \LOG\frac
      {(d+2) \cdot \max_{1\le i\le d}|a_i| \cdot 2^{\E (d+1) +1}}
      {|a_d| \cdot (t\delta(1-p)/d)^d}
      \right\rceil \\
    &=& \left\lceil -d\LOG(1-p) \; + \; \LOG\frac
      {(d+2) \cdot \max_{1\le i\le d}|a_i| \cdot 2^{\E (d+1) +1}}
      {|a_d| \cdot (t\delta/d)^d}
      \right\rceil
  \end{eqnarray*}
  as was claimed in the lemma.
  \qed
\end{proof}
  Since the formula for $\LS(p)$
  in the lemma above looks rather complicated,
  we interpret it here.
  We observe that $c_u$ is a constant
  because it is defined only by constants:
  The degree $d$ and the coefficients $a_i$ are defined by $f$,
  and the parameters $\E$ and $\delta$ are given by the input.
  We make the asymptotic behavior
  $\LS(p) = O\left(-d\LOG(1-p)\right)$
  for $p\to1$
  explicit in the following corollary:
  We show that
  $d$ additional bits of the precision are sufficient
  to halve the failure probability.
\begin{corollary}
  Let $f$ be a univariate polynomial of degree $d$
  and let $\LS:(0,1)\to\N$
  be the precision function in Formula~(\ref{for-theo-uni-lsafe}).
  Then
  \begin{eqnarray*}
    \LS\left(\frac{1+p}{2}\right) &=& \LS(p)+d.
  \end{eqnarray*}
\end{corollary}
\begin{proof}
  Due to Formula~(\ref{for-theo-uni-lsafe}) we have:
  \begin{eqnarray*}
    \LS\left(\frac{1+p}{2}\right) &=&
      \left\lceil  -d \LOG \left( 1- \left( \frac{1+p}{2} \right) \right) + c_u \right\rceil \\[0.5ex]
    &=&
      \left\lceil  -d \LOG \left( \frac{1-p}{2} \right) + c_u \right\rceil \\[0.5ex]
    &=&
      \left\lceil  -d \left( \LOG (1-p) - \LOG (2) \right) + c_u \right\rceil \\
    &=&
      \left\lceil  -d \LOG (1-p) + d + c_u \right\rceil \\
    &=&
      \left\lceil  -d \LOG (1-p) + c_u \right\rceil + d \\
    &=&
      \LS(p) + d 
  \end{eqnarray*}
  Because $d$ is a natural number,
  we can pull it out of the brackets.
  \qed
\end{proof}

\section[The Direct Approach Using Estimates (1st Stage, rv-suit)]{The Direct Approach Using Estimates}
\label{sec-direct-approach}

  This approach derives the bounding functions
  which are associated with region- and value-suitability
  in the first stage of the analysis
  (see Figure~\ref{fig-ana-direct}).
  It is partially based on
  the geometric interpretation of the function $f$ at hand.
  More precisely, it presumes that the critical set of $f$
  is embedded in geometric objects
  for which we know simple mathematical descriptions
  (e.g., lines, circles, etc.).
  The derivation of bounds from geometric interpretations
  is also presented in~\cite{MOS06,MOS11}.
  In Section~\ref{sec-direct-present}
  we explain the derivation of the bounds.
  In Section~\ref{sec-direct-example}
  we show some examples.
  \begin{figure}[h]\centering
    \includegraphics[width=.45\columnwidth]{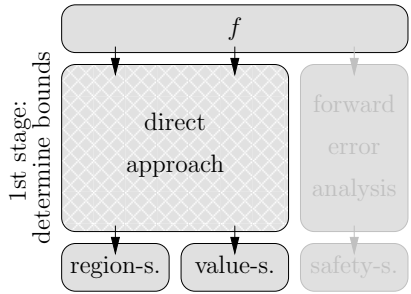}
    \caption{The direct approach and its interface.}
    \label{fig-ana-direct}
  \end{figure}

\subsection{Presentation}
\label{sec-direct-present}

  The steps of the direct approach are summarized
  in Table~\ref{tab-steps-direct}.
  To facilitate the presentation of the geometric interpretation,
  we assume that the function $f$ is continuous everywhere
  and that we do not allow any exceptional points.
  Then the critical set of $f$ equals the zero set of $f$.
  Hence the region of uncertainty is an environment of the zero set
  in this case.
  We define the region of uncertainty $R_\gamma$
  as it is defined in Formula~(\ref{for-region-uncert}).
  In the first step,
  we choose $\GAL$
  which is the domain of $\gamma$.
  Or in other words, we choose $\hat\gamma$.
  Sometimes,
  certain choices of $\GAL$
  may be more useful than others, e.g.,
  cubic environments where $\hat\gamma_i=\hat\gamma_j$
  for all $1\le i,j\le k$.

  Now assume that we have chosen $\GAL$.
  In the second step,
  we estimate
  (an upper bound on)
  the volume of the region of uncertainty $R_{\gamma}$
  by a function $\nu_f(\gamma)$ for $\gamma\in\GAL$.
  In the direct approach,
  we hope that a geometric interpretation of the zero set
  supports the estimation.
  For that purpose it would be helpful
  if the region of uncertainty is embedded in a line, a circle,
  or any other geometric structure that we can easily describe mathematically.

  Assume further that we have fixed the bound $\nu_f$.
  In the third step,
  we need to determine a function $\varphi_f(\gamma)$
  on the minimum absolute value of $f$
  outside of $R_\gamma$.
  This is the most difficult step in the direct approach:
  \emph{Although geometric interpretation may be helpful in the second step,
  mathematical considerations are necessary to derive $\varphi_f$.}
  Therefore we hope that $\varphi_f$ is ``obvious'' enough
  to get guessed.
  If there is no chance to guess $\varphi_f$,
  we need to try one of the alternative approaches
  of the next sections, that means,
  the bottom-up approach or the top-down approach.
\begin{table}[h]
  \centerline{\fbox{
  \begin{minipage}{0.9\columnwidth}\medskip
    Step 1: choose the set $\GAL$ (define $\hat\gamma$) \\
    Step 2: estimate $\nu_f(\gamma)$ in dependence on $\GAL$
      (define $\nu_f$) \\
    Step 3: estimate $\varphi_f(\gamma)$ in dependence on $\nu_f(\gamma)$
      (define $\varphi_f$) \medskip
  \end{minipage}}}
    \caption{Instructions for performing the direct approach.}
    \label{tab-steps-direct}
\end{table}

\subsection{Examples}
\label{sec-direct-example}

  We present two examples that
  use the direct approach to derive the bounds
  for the region-value-suitability.
\begin{example}
\newcommand{\UX}{{u_x}}%
\newcommand{\VX}{{v_x}}%
\newcommand{\QX}{{q_x}}%
\newcommand{\UY}{{u_y}}%
\newcommand{\VY}{{v_y}}%
\newcommand{\QY}{{q_y}}%
\newcommand{\GUX}{\gamma_{u_x}}%
\newcommand{\GVX}{\gamma_{v_x}}%
\newcommand{\GQX}{\gamma_{q_x}}%
\newcommand{\GUY}{\gamma_{u_y}}%
\newcommand{\GVY}{\gamma_{v_y}}%
\newcommand{\GQY}{\gamma_{q_y}}%
  We consider the $\inbox$ predicate in the plane.
  Let $u$ and $v$ be two opposite vertices of the box
  and let $q$ be the query point.
  Then $\inbox(u,v,q)$ is decided by the sign of the function
  \begin{eqnarray}
    f (u,v,q)
    &=&
      f (\UX,\UY, \VX,\VY, \QX,\QY) \nonumber \\
    &:=&
      \max \left\{
        \left(\QX-\UX\right)\left(\QX-\VX\right), \;
        \left(\QY-\UY\right)\left(\QY-\VY\right)
      \right\}. \label{for-inbox-uvq}
  \end{eqnarray}
  The function
  is negative if $x$ lies inside of the box,
  it is zero if $x$ lies in the boundary,
  and it is positive if $x$ lies outside of the box.

  Step~1:
  We choose an arbitrary
  $\hat\gamma=(\hat\gamma_\UX,\hat\gamma_\UY,\hat\gamma_\VX,\hat\gamma_\VY,\hat\gamma_\QX,\hat\gamma_\QY)\in\R_{>0}^6$.

  Step~2:
  The box is defined by $u$ and $v$.
  This fact is true
  independent of the choices for $\GUX$, $\GUY$, $\GVX$ and $\GVY$.
  We observe that
  the largest box inside of the perturbation area $U_\delta$
  is the boundary of $U_\delta$ itself.
  This observation leads to the upper bound
  \begin{eqnarray*}
    \nu_f(\gamma)
      &=&
        \nu_f (\GUX,\GUY, \GVX,\GVY, \GQX,\GQY) \\
      &:=&
       4 \left( \GQX \delta_y \;+\; \GQY \delta_x \right)
  \end{eqnarray*}
  on the volume of the region of uncertainty
  if we take the horizontal distance $\GQX$
  and the vertical distance $\GQY$ from the boundary of the box into account.
  That means,
  $\nu_f$ depends on the distances $\GQX$ and $\GQY$
  of the query point $q$
  from the zero set.

  Step 3:
  The evaluation of Formula~(\ref{for-inbox-uvq})
  at query points where
  $\QX$ has distance $\GQX$ from $\UX$ or $\VX$, and
  $\QY$ has distance $\GQY$ from $\UY$ or $\VY$, leads to
  \begin{eqnarray*}
    \varphi_f(\gamma)
    &:=&
      \min\left\{
        \left| \gamma^2_\QX - \GQX \cdot | \VX-\UX |\right|, \;
        \left| \gamma^2_\QY - \GQY \cdot | \VY-\UY |\right|
      \right\}.
  \end{eqnarray*}
  The derived bounds fulfill the desired properties.
  \hfill$\bigcirc$
\end{example}
\begin{example}
\newcommand{\CX}{{c_x}}%
\newcommand{\CY}{{c_y}}%
\newcommand{\QX}{{q_x}}%
\newcommand{\QY}{{q_y}}%
\newcommand{\GCX}{\gamma_{c_x}}%
\newcommand{\GCY}{\gamma_{c_y}}%
\newcommand{\GR}{\gamma_{r}}%
\newcommand{\GQX}{\gamma_{q_x}}%
\newcommand{\GQY}{\gamma_{q_y}}%
  We consider the $\incirc$ predicate in the plane.
  Let $c$ be the center of the circle,
  let $r>0$ be its radius,
  and let $q$ be the query point.
  Then $\incirc(c,r,q)$ is decided by the sign of the function
  \begin{eqnarray}
    f(c,r,q)
    &=& f(\CX,\CY,r,\QX,\QY) \nonumber \\
    &:=&
      \left(\QX-\CX\right)^2 \;+\;
      \left(\QY-\CY\right)^2 \;-\;
      r^2 \label{for-incirc-crq}
  \end{eqnarray}
  The function is negative if $x$ lies inside of the circle,
  it is zero if $x$ lies on the circle,
  and it is positive if $x$ lies outside of the circle.

  Step~1:
  We choose
  $\hat\gamma=(\hat\gamma_\CX,\hat\gamma_\CY,\hat\gamma_r,\hat\gamma_\QX,\hat\gamma_\QY)\in\R_{>0}^5$
  where $\hat\gamma_\QX=\hat\gamma_\QY$.
  In addition,
  we choose $\hat\gamma_r< r$ for simplicity.

  Step~2:
  The largest circle that fits into the perturbation area $U_\delta$ has radius
  $\min\left\{\delta_x,\delta_y\right\}$.
  If we intersect any larger circle with $U_\delta$,
  the total length of the circular arcs inside of $U_\delta$
  cannot be larger than
  $2\pi \cdot \min\left\{\delta_x,\delta_y\right\}$.
  This bounds the total length of the zero set.

  Now we define the region of uncertainty
  by spherical environments:
  The region of uncertainty is the union of open discs
  of radius $\GQX$ which are located at the zeros.
  Then the width of the region of uncertainty
  is given by the diameter of the discs, i.e., by $2\GQX$.
  As a consequence,
  \begin{eqnarray*}
    \nu_f(\gamma)
      &=& \nu_f(\GCX,\GCY,\GR,\GQX,\GQY) \\
      &:=&
        4\pi\GQX \cdot \min\left\{\delta_x,\delta_y\right\}
  \end{eqnarray*}
  is an upper bound on the volume of $R_\delta$.
  That means,
  $\nu_f$ depends on the distance $\GQX$ of the query point $q$
  from the zero set.

  Step~3:
  The absolute value of Formula~(\ref{for-incirc-crq}) is minimal
  if the query point $q$ lies inside of the circle
  and has distance $\GQX$ from it.
  This leads to
  \begin{eqnarray*}
    \varphi_f(\gamma)
      &:=&
      \left|
      \left( r \;-\; \GQX \right)^2 \;-\; r^2
      \right| \\
      &=&
      \GQX \left( \GQX - 2r \right).
  \end{eqnarray*}
  The derived bounds fulfill the desired properties.
  \hfill$\bigcirc$
\end{example}

\section[The Bottom-up Approach Using Calculation Rules (1st Stage, rv-suit)]
{The Bottom-up Approach Using Calculation Rules}
\label{sec-bottom-up}

  In the first stage of the analysis,
  this approach derives the bounding functions
  which are associated with the region- and value-suitability
  (see Figure~\ref{fig-ana-bu}).
  We can apply this approach to certain composed functions.
  That means,
  if $f$ is composed by $g$ and $h$,
  we can derive the bounds for $f$ from the bounds for $g$ and $h$
  under certain conditions.
  We present some mathematical constructs
  which preserve the region- and value-suitability
  and introduce useful calculation rules for their bounds.
  Namely we introduce
  the \emph{lower-bounding rule} in Section~\ref{sec-rule-lower-bound},
  the \emph{product rule} in Section~\ref{sec-rule-product} and
  the \emph{min rule} and \emph{max rule} in Section~\ref{sec-rule-min-max}.
  We point to a general way to formulate rules
  in Section~\ref{sec-rule-general}.
  The list of rules is by far not complete.
  Nevertheless,
  they are already sufficient to
  derive the bounding functions
  for multivariate polynomials
  as we show in Section~\ref{sec-multivariate-poly}.
  \emph{With the bottom-up approach we present an entirely new approach
  to derive the bounding functions for the region-suitability and
  value-suitability.
  Furthermore we present a new way to analyze multivariate polynomials.}
  \begin{figure}[h]\centering
    \includegraphics[width=.45\columnwidth]{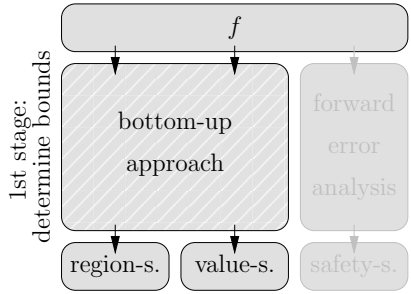}
    \caption{The bottom-up approach and its interface.}
    \label{fig-ana-bu}
  \end{figure}

\subsection{Lower-bounding Rule}
\label{sec-rule-lower-bound}

  Our first rule states
  that every function is region-value-suitable
  if there is a lower bounding function which is region-value-suitable.
  Note that there are no further restrictions on $f$.
\begin{theorem}[lower bound]\label{theo-lower-bound}
  Let $\PREDL$ be a predicate description.
  If there is a region-value-suitable function
  $g:\U_\delta(A)\to\R$ and $c\in\R_{>0}$ where
  \begin{eqnarray}\label{for-ruleboundfun}
    |f(x)| &\ge& c \, |g(x)|,
  \end{eqnarray}
  then $f$ is also region-value-suitable
  with the following bounding functions:
  \begin{eqnarray*}
    \nu_f(\gamma) &:=& \nu_g(\gamma) \\
    \varphi_f(\gamma) &:=& c\varphi_g(\gamma).
  \end{eqnarray*}
  If $f$ is in addition safety-suitable, $f$ is analyzable.
\end{theorem}
\begin{proof}
Part 1 (region-suitable).
  Let $(a_i)_{i\in\N}$ be a sequence in the set $U_\delta(\x)$ with
  $\lim_{i\to\infty} f(a_i)=0$.
  Then Formula~(\ref{for-ruleboundfun}) implies that
  $\lim_{i\to\infty} g(a_i)=0$.
  That means,
  critical points of $f$ are critical points of $g$.
  Therefore we set $C_f(\x) :=C_g(\x)$.
  As a consequence the region bound $\nu_f(\gamma) := \nu_g(\gamma)$
  is sufficient for the region-suitability of $f$.

Part 2 (value-suitable).
  Because we set $C_f(\x) = C_g(\x)$,
  we have $R_f(\x) = R_g(\x)$.
  Due to Formula~(\ref{for-ruleboundfun}),
  the minimum absolute value of $f$
  outside of the region of uncertainty $R_f(\x)$
  is bounded by the minimum absolute value of $g$
  outside of the (same) region of uncertainty $R_g(\x)$.
  Hence the bound
  $\varphi_f(\gamma) = c\varphi_g(\gamma)$
  is sufficient for the value-suitability of $f$.

Part 3 (analyzable). Trivial.
  \qed
\end{proof}

\subsection{Product Rule}
\label{sec-rule-product}

  The next rule states
  that the product of region-value-suitable functions
  is also region-value-suitable.
  Furthermore,
  we show how to derive appropriate bounds.
\begin{theorem}[product]\label{theo-product}
  Let $(f, k, A_g\times A_{gh}\times A_h, \delta, \E, \GAL, t)$
  be a predicate description where
  $A_g\subset\R^j$,
  $A_{gh}\subset\R^{\ell-j}$ and
  $A_h\subset\R^{k-\ell}$
  for
  $j\in\N_{0}$ and $\ell,k\in\N$ with $j\le \ell\le k$.
  If there are two region-value-suitable functions
  \begin{eqnarray*}
    g&:&\U_{(\delta_1,\ldots,\delta_\ell)}(A_g\times A_{gh})\to\R \\
    h&:&\U_{(\delta_{j+1},\ldots,\delta_k)}(A_{gh}\times A_h)\to\R
  \end{eqnarray*}
  such that
  \begin{eqnarray*}
    f(x_1,\ldots,x_k) &=& g(x_1,\ldots,x_\ell) \cdot h(x_{j+1},\ldots,x_k),
  \end{eqnarray*}
  then $f$ is also
  region-value-suitable with the following bounding functions:
  \begin{eqnarray}
    \varphi_{f}(\gamma) &:=&
      \varphi_g(\gamma_1,\ldots,\gamma_\ell) \cdot
      \varphi_h(\gamma_{j+1},\ldots,\gamma_k) \label{for-phi-prod-rule} \\
    \nu_{f}(\gamma)
      &:=& \min
      \left\{ \prod_{i=1}^k (2\delta_i), \right. \nonumber \\
      &&\qquad\left.\label{for-nug-plus-nuh}
      \nu_g(\gamma_1,\ldots,\gamma_\ell) \!\!\! \prod_{i=\ell+1}^k \!\!\! (2\delta_i)
      \, + \,
      \nu_h(\gamma_{j+1},\ldots,\gamma_k) \prod_{i=1}^j (2\delta_i)\right\}.
  \end{eqnarray}
  Furthermore, if $j=\ell$,
  we can replace
  the last equation
  by the tighter bound
  \begin{eqnarray}\label{for-beta-product-rule}
    \chi_{f}(\gamma) &:=&
      \chi_g(\gamma_1,\ldots,\gamma_j) \cdot
      \chi_h(\gamma_{j+1},\ldots,\gamma_k).
  \end{eqnarray}
  If $f$ is in addition safety-suitable, $f$ is analyzable
  (independent of $j=\ell$).
\end{theorem}
\begin{proof}
  Part 1 (value-suitable).
  Let $x\in U_\delta(\x)$ such that
  $(x_1,\ldots,x_\ell)$
  does not lie in the region of uncertainty\footnote{%
    To avoid confusion,
    we occasionally add the function name to the index of
    the region of uncertainty or the perturbation area
    within the proof, e.g.\ $R_{f,\gamma}$ and $U_{f,\delta}$.}
  of $g$, that means
  \begin{eqnarray}\label{for-g-region-condition}
    (x_1,\ldots,x_\ell)
    \not\in
    R_{g,(\gamma_1,\ldots,\gamma_\ell)}(\x_1,\ldots,\x_\ell),
  \end{eqnarray}
  and that $(x_{j+1},\ldots,x_k)$ does not lie in the region of uncertainty
  of $h$, that means
  \begin{eqnarray}\label{for-h-region-condition}
    (x_{j+1},\ldots,x_k)
    \not\in
    R_{h,(\gamma_{j+1},\ldots,\gamma_k)}(\x_{j+1},\ldots,\x_k).
  \end{eqnarray}
  Because $g$ and $h$ are value-suitable,
  we obtain:
  \begin{eqnarray*}
    |f(x)| &=& |g(x_1,\ldots,x_\ell)| \cdot |h(x_{j+1},\ldots,x_k)| \\
    &\ge& \varphi_g(\gamma_1,\ldots,\gamma_\ell) \cdot
    \varphi_h(\gamma_{j+1},\ldots,\gamma_k) \\
    &=& \varphi_f(\gamma)
  \end{eqnarray*}
  on the absolute value of $f$.

  Part 2 (region-suitable).
  Because of the argumentation above,
  we must construct
  the region of uncertainty $R_f$
  such that $x\in\R^k$ lies outside of $R_f$
  only if the conditions in Formula~(\ref{for-g-region-condition})
  and~(\ref{for-h-region-condition}) are fulfilled.

  Case $j=\ell$.
  Then the arguments of $g$ and $h$ are disjoint.
  This case is illustrated in Figure~\ref{fig-region-ag-ah}.
  \begin{figure}[t]\centering
    \includegraphics[width=.71\columnwidth]{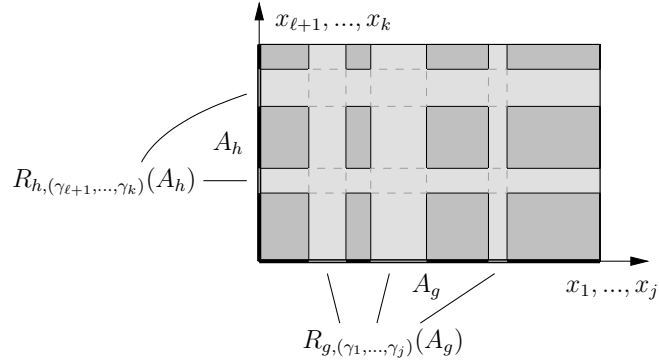}
    \caption{Case $j=\ell$:
      The (dark shaded) complement of $R_f$
      is the Cartesian product of the complement of $R_g$
      and the complement of $R_h$.}
    \label{fig-region-ag-ah}
  \end{figure}
  We observe that for each point $(x_1,\ldots,x_j)$ outside of $R_g$
  and each point $(x_{\ell+1},\ldots,x_k)$ outside of $R_h$
  their concatenation $x$ lies outside of $R_f$.
  Therefore we determine the volume of the complement of $R_f$
  inside of the perturbation area as
  \begin{eqnarray*}
    \mu\left(U_{f,\delta}(\x)\setminus R_{f,\gamma}(\x)\right)
    &=&
    \mu\left(U_{g,(\delta_1,\ldots,\delta_j)}(\x_1,\ldots,\x_j) \right.\\
    &&\qquad\qquad \left.\setminus R_{g,(\gamma_1,\ldots,\gamma_j)}(\x_1,\ldots,x_j)\right) \\
    && 
    \cdot\;\mu\left(U_{h,(\delta_{\ell+1},\ldots,\delta_k)}(\x_{\ell+1},\ldots,\x_k) \right.\\
    &&\qquad\qquad \left.\setminus R_{h,(\gamma_{\ell+1},\ldots,\gamma_k)}(\x_{\ell+1},\ldots,x_k)\right).
  \end{eqnarray*}
  As a consequence Formula~(\ref{for-beta-product-rule}) is true.

  Case $j<\ell$.
  In contrast to the discussion above,
  $g$ and $h$
  share the arguments $x_{j+1},\ldots,x_\ell$.
  This case is illustrated in Figure~\ref{fig-region-ag-agh-ah}.
  \begin{figure}[t]\centering
    \includegraphics[width=.78\columnwidth]{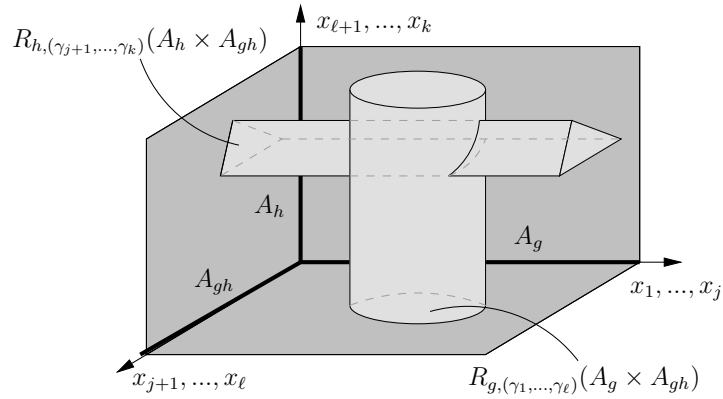}
    \caption{Case $j<\ell$:
      The (light shaded) region of uncertainty $R_f$ is the union of
      two Cartesian products.}
    \label{fig-region-ag-agh-ah}
  \end{figure}
  We denote the projection of the first $j$
  (respectively, the last $k-\ell$)
  coordinates by $\pi_{\le j}$ (respectively, $\pi_{> \ell}$).
  In this case,
  Formula~(\ref{for-beta-product-rule})
  does not have to be true.
  That is why
  we define $R_f$ as
  \begin{eqnarray*}
    R_{f,\gamma}(\x)
    &:=&
    R_{g,(\gamma_1,\ldots,\gamma_\ell)}(\x_1,\ldots,\x_\ell)
      \times
      \pi_{>\ell} (\U_\delta(\x)) \\
    && \cup \;
      \pi_{\le j} (\U_\delta(\x))
      \times
    R_{h,(\gamma_{j+1},\ldots,\gamma_k)}(\x_{j+1},\ldots,\x_k).
  \end{eqnarray*}
  Now we can upper-bound the volume of $R_f$
  by means of $\nu_g$ and $\nu_h$
  which leads immediately to the sum in the last line of
  Formula~(\ref{for-nug-plus-nuh}).
  Of course,
  the volume of the region of uncertainty is
  bounded by the volume of the perturbation area
  which justifies the first line of Formula~(\ref{for-nug-plus-nuh}).
  This finishes the proof.
  \qed
\end{proof}

\subsection{Min Rule, Max Rule}
\label{sec-rule-min-max}

  The next two rules state that the minimum and maximum
  of finitely many region-value-suitable functions
  are also region-value-suitable.
  Furthermore,
  we show how to derive appropriate bounds.
\begin{theorem}[min, max]\label{theo-ruleminmax}
  Let $g$ and $h$ be two region-value-suitable functions
  as defined in Theorem~\ref{theo-product}.
  Then the functions
  \begin{eqnarray*}
    \fmin,\fmax&:&\U_{\delta}(A_g \times A_{gh} \times A_h)\to\R,\\
    \fmin(x_1,\ldots,x_k) &:=& \min \{g(x_1,\ldots,x_\ell), h(x_{j+1},\ldots,x_k)\}\\
    \fmax(x_1,\ldots,x_k) &:=& \max \{g(x_1,\ldots,x_\ell), h(x_{j+1},\ldots,x_k)\}
  \end{eqnarray*}
  are region-value-suitable
  with bounds $\varphi_{\fmin}$ and $\nu_{\fmin}$ for $\fmin$
  and bounds $\varphi_{\fmax}$ and $\nu_{\fmax}$  for $\fmax$
  where
  \begin{eqnarray}
    \nonumber %
    \varphi_{\fmin}(\gamma) &:=&
      \min \{ \varphi_g(\gamma_1,\ldots,\gamma_\ell),
        \varphi_h(\gamma_{j+1},\ldots,\gamma_k)\} \\
    \label{for-phi-max-rule}
    \varphi_{\fmax}(\gamma) &:=&
      \max \{ \varphi_g(\gamma_1,\ldots,\gamma_\ell),
        \varphi_h(\gamma_{j+1},\ldots,\gamma_k)\}\\%
    \nu_{\fmin}(\gamma) := \nu_{\fmax}(\gamma) 
      &:=& \nu_f(\gamma)
      \text{ (see Formula~(\ref{for-nug-plus-nuh}))}. \nonumber
  \end{eqnarray}
  Furthermore, if $j=\ell$,
  we can replace
  $\nu_{\fmin}(\gamma)$ and $\nu_{\fmax}(\gamma)$
  by the tighter bounds
  \begin{eqnarray*}
    \chi_{\fmin}(\gamma) := \chi_{\fmax}(\gamma)
    &:=&
      \chi_g(\gamma_1,\ldots,\gamma_j) \cdot
      \chi_h(\gamma_{j+1},\ldots,\gamma_k).
  \end{eqnarray*}
  If $\fmin$ (respectively $\fmax$) is in addition safety-suitable,
  it is also analyzable
  (independent of $j=\ell$).
\end{theorem}
\begin{proof}
  The line of argumentation follows exactly
  the proof of Theorem~\ref{theo-product}.
  \qed
\end{proof}

\subsection{General Rule}
\label{sec-rule-general}

  We do not claim that the list of rules is complete.
  On the contrary,
  we suggest that the approach may be extended by further rules.
  We emphasize that the bottom-up approach is constructive:
  We build new region-value-suitable functions from
  already proven region-value-suitable functions.
  The argumentation always follows the proof of the product rule,
  that means,
  the compound of $g$ and $h$ inherits the desired property from $g$ and $h$:
  (a) \emph{outside of the union} of the regions of uncertainty
  for \emph{shared} arguments,
  and (b) \emph{inside of the Cartesian product of the complement}
  of the regions of uncertainty
  for \emph{disjoint} arguments
  (see Figure~\ref{fig-region-ag-ah}).

  We remark that,
  if we want to derive the bounds for a specific function $f$,
  we first need to determine the parse tree of $f$ according to the known rules;
  this may be a non-obvious task in general.
  The instructions
  of the bottom-up approach are summed up in the following table.
\begin{table}[h]
  \centerline{\fbox{
  \begin{minipage}{0.9\columnwidth}\medskip
    Step 1: determine parse tree according to the rules \\
    Step 2: determine bounds bottom-up according to the parse tree \medskip
  \end{minipage}}}
    \caption{Instructions for performing the bottom-up approach.}
    \label{tab-steps-bu}
\end{table}

\subsection{Example: Multivariate Polynomials}
\label{sec-multivariate-poly}

  It is important to see
  that the rules lead to a generic approach
  to construct entire classes of region-value-suitable functions.
  In the following we use this approach to analyze multivariate polynomials.
  (A different way to analyze multivariate polynomials
  was presented before in~\cite{MOS11}.)
  So far we know
  that univariate polynomials are region-value-suitable.
  Now we show
  how we transfer
  the region-value-suitability property of $(k-1)$-variate polynomials
  to $k$-variate polynomials
  by means of the product rule
  and the lower bound rule.
  Moreover,
  we completely analyze $k$-variate polynomials afterwards.

\subsection*{Preparation}

  We prepare the analysis of multivariate polynomials with further definitions.
  Let $k\in\N$.
  For $\beta\in\N_0^k$ and $x\in\R^k$
  we define $x^\beta$
  as the term
  $x^\beta:=x_1^{\beta_1} \cdot \ldots \cdot x_k^{\beta_k}$.

  Next
  we define the reverse lexicographic order\footnote{%
    For lexicographic order see Cormen~et~al.~\cite{CLR90}.}
  on $k$-tuples.
  Let $\alpha,\beta\in\N_0^k$.
  Then we define $\alpha\,\Lex\,\beta\label{def-lex-inline}$
  if and only if
  there is $\ell\in\{1,\ldots,k\}$
  such that
  $\alpha_{j}=\beta_{j}$
  for all $\ell < j \le k$
  and
  $\alpha_{\ell}<\beta_{\ell}$.

  In addition
  we denote by $\Perm(k)$
  the set of bijective functions
  $\sigma:\{1,\ldots,k\} \to \{1,\ldots,k\}$.
  In other words,
  $\Perm(k)$ is
  the set of
  permutations\footnote{%
    Algebra: For permutation see Lamprecht~\cite{L93}.}
  of $\{1,\ldots,k\}$.

  Now let $\alpha,\beta\in\N_0^k$
  and let $\sigma\in\Perm(k)$.
  We define the \emph{permutation $\sigma$ of a tuple
  $\alpha=(\alpha_1,\ldots,\alpha_k)$} by
  $\sigma(\alpha) :=
      \left(\alpha_{\sigma^{-1}(1)}, \ldots, \alpha_{\sigma^{-1}(k)} \right)$.
  Further we define the
  \emph{reverse lexicographic order
  after the permutation\label{def-lexafter-inline} $\sigma$} as
  \begin{eqnarray*}
    \alpha \,\Lex_\sigma\, \beta
    \FORMSEP & :\Longleftrightarrow & \FORMSEP
    \sigma(\alpha)\,\Lex\, \sigma(\beta).
  \end{eqnarray*}
  Let $\Ind\subset\N_0^k$ be finite.
  We denote the set of largest elements in $\Ind$ by
  \begin{eqnarray*}
    \IM &:=& \left\{
      \beta\in\Ind \ST
      \text{there is }
      \sigma\in\Perm(k)
      \text{ such that }
      \alpha\Lex_\sigma\beta
      \text{ for all }
      \alpha\in\Ind, \alpha\ne\beta
      \right\}.
  \end{eqnarray*}
  We observe
  that there may be $\beta\in\Ind$ which do not belong to $\IM$.
  We observe further
  that different permutations may lead to the same local maximum.
  For each $\beta\in\IM$ we collect these permutations in the set
  \begin{eqnarray*}
    \Perm_{\beta}(k) &:=& \left\{
      \sigma\in\Perm(k) \ST
      \beta = \max\nolimits_{\Lex_\sigma}\Ind
      \right\}.
  \end{eqnarray*}

\subsection*{The region- and value-suitability}

  We prove that all multivariate polynomials are region-value-suitable.
\begin{lemma}\label{lem-multi-poly-suit}
  Let $\PREDL$ be a predicate description for
  the $k$-variate polynomial ($k\ge 2$)
  \begin{eqnarray*}%
    f(x) := \sum_{\iota\in\Ind} a_\iota x^\iota
  \end{eqnarray*}
  where $\Ind\subset\N_0^k$ is finite
  and $a_\iota\in\R_{\neq 0}$ for all $\iota\in\Ind$.
  Then $f$ is region-value-suitable.
  There are bounding functions
  for every $\beta\in\IM$:
  \begin{eqnarray*}
    \varphi_f(\gamma) &:=& |a_\beta| \cdot \gamma^\beta\\
    \chi_f(\gamma) &:=& \prod_{i=1}^{k} 2\left(\delta_i-\beta_i\gamma_i\right).
  \end{eqnarray*}
\end{lemma}
\begin{proof}
  Preparing consideration.
  Let $\beta\in\IM$ and let $\sigma\in\Perm_\beta(k)$.
  Once chosen, $\beta$ and $\sigma$ are fixed in this proof.
  Because of the reverse lexicographic order,
  the maximal exponent of $x_{\sigma(k)}$ in $f(x)$ is $\beta_{\sigma(k)}$.
  Therefore we can write $f$ as
  \begin{eqnarray*}
  f(x) = b_{\beta_{\sigma(k)}} \cdot x_{\sigma(k)}^{\beta_{\sigma(k)}}
         + b_{\beta_{\sigma(k)}-1} \cdot x_{\sigma(k)}^{\beta_{\sigma(k)}-1}
	 + \ldots
	 + b_1 \cdot x_{\sigma(k)}
	 + b_0
  \end{eqnarray*}
  where the $b_i(x_{\sigma(1)},\ldots,x_{\sigma(k-1}))$
  are $(k-1)$-variate polynomials
  for $0\le i\le\beta_{\sigma(k)}$.
  For a moment
  we consider the complex continuation of the polynomial $f$,
  i.e.\ $f\in\C[z]$.
  Furthermore we \emph{assume}\footnote{%
    We discuss the assumption in Part~2 of the proof.}
  that the value of $b_{\beta_{\sigma(k)}}$ is not zero.
  Then there are $\beta_{\sigma(k)}$ (not necessarily distinct)
  functions $\zeta_i:\C^{k-1}\to\C$ such that
  we can write $f$ in the way
  \begin{eqnarray*}
    f(z) &=& b_{\beta_{\sigma(k)}}(z_{\sigma(1)},\ldots,z_{\sigma(k-1)})
           \cdot
	   \prod_{i=1}^{\beta_{\sigma(k)}}
	   (z_{\sigma(k)}-\zeta_i(z_{\sigma(1)},\ldots,z_{\sigma(k-1)})).
  \end{eqnarray*}
  We remark that if we consider $f$ as a polynomial in $z_{\sigma(k)}$
  with parameterized coefficients $b_i$,
  then the functions $\zeta_i$ define the parameterized roots.
  Even if the location of the roots is variable,
  the total number of the roots is definitely bounded by $\beta_{\sigma(k)}$.
  In case that $z_{\sigma(k)}$ has a distance of at least
  $\gamma_{\sigma(k)}$ to the values $\zeta_i$,
  we can lower bound the absolute value of $f$ by
  \begin{eqnarray}\label{for-factorizef2}
    |f(z)| &\ge&
    {\left|b_{\beta_{\sigma(k)}}\left(z_{\sigma(1)},\ldots,z_{\sigma(k-1)}\right)\right|}
    \cdot
    {\gamma_{\sigma(k)}^{\beta_{\sigma(k)}}}
  \end{eqnarray}
  Therefore this bound is especially true for real arguments.
  Before we end the consideration in the complex space,
  we add a remark.
  Sagraloff et al.~\cite{SY11,MOS11} suggested a way to improve this estimate:
  While preserving the \emph{total} region-bound $\varphi_f$,
  it is possible to redistribute the region of uncertainty
  around the zeros of $f$ in a way where
  the amount of the \emph{individual} region-contribution per zero may differ; 
  they have shown that a certain redistribution improves the estimate
  in Formula~(\ref{for-factorizef2}).
  Next we use mathematical induction to prove
  that $f$ is region-value-suitable.

Part~1 (basis). Let $j=1$.
  Due to Lemma~\ref{lem-uni-poly-ana}
  univariate polynomials are region-value-suitable.

Part~2 (inductive step). Let $1<j\le k$.
  We define the function $g_j$ as
  \begin{eqnarray*}
    g_j\left(z_{\sigma(1)},\ldots,z_{\sigma(j-1)}\right)
    &:=&
    {b_{\beta_{\sigma(j)}}\left(z_{\sigma(1)},\ldots,z_{\sigma(j-1)}\right)}.
  \end{eqnarray*}
  Since $g_j$ is a polynomial in $j-1$ variables,
  $g_j$ is region-value-suitable by induction.
  Because of Theorem~\ref{theo-lower-bound},
  the function $|g_j|$ is region-value-suitable
  with the same bounds.
  Furthermore, we define the functions
  \begin{eqnarray*}
    h_j(z_{\sigma(j)})
      &:=& {\gamma_{\sigma(j)}^{\beta_{\sigma(j)}}} \\
    \varphi_{h_j}(\gamma_{\sigma(j)})
      &:=& \gamma_{\sigma(j)}^{\beta_{\sigma(j)}} \\
    \nu_{h_j}(\gamma_{\sigma(j)})
      &:=& 2\beta_{\sigma(j)}\gamma_{\sigma(j)}.
  \end{eqnarray*}
  Obviously
  $h_j$ is region-value-suitable.
  We have $|f_j|\ge|g_j|\cdot h_j$.
  Then the product $|g_j|\cdot h_j$ is also region-value-suitable
  because of Theorem~\ref{theo-product}.
  Be aware that the construction
  of the estimate in Formula~(\ref{for-factorizef2})
  is based on the assumption that the coefficient
  $b_{\beta_{\sigma(j)}}$ of $f_j$ is not zero.
  We observe that this is only guaranteed outside of
  the region of uncertainty of $g_j$.
  We observe further
  that the construction in the proof of Theorem~\ref{theo-lower-bound}
  preserves the region of uncertainty,
  that means, $R_{g_j}\subset R_{f_j}$.
  Therefore the assumption is justified and
  we can conclude that $f_j$ is region-value-suitable.
  It remains to show that the claimed bounding functions
  $\varphi_f$ and $\nu_f$ are true.
  
Part~3 ($\varphi_f$).
  The basis $j=1$
  follows from Lemma~\ref{lem-uni-poly-ana}:
  \begin{eqnarray*}
    \varphi_{f_1}\left(\gamma_{\sigma(1)}\right)
      &:=&
        \left|a_\beta\right| \cdot \gamma_{\sigma(1)}^{\beta_{\sigma(1)}}
  \end{eqnarray*}
  (Be aware
  that the real coefficient $a_\beta$
  is contained in every $g_j$ for $1< j\le k$.)
  Now let $1<j\le k$.
  For the induction step
  we need the following observation:
  Because of the reverse lexicographic order,
  the maximal exponent of $x_{\sigma(j-1)}$
  in the parameterized coefficient
  $b_{\beta_{\sigma(j)}}(x_{\sigma(1)},\ldots,x_{\sigma(j-1)})$
  is $\beta_{\sigma(j-1)}$.
  We have
  \begin{eqnarray*}
    \varphi_{f_j}\left(\gamma_{\sigma(1)},\ldots,\gamma_{\sigma(j)}\right)
      &:=&
        \left|a_\beta\right|
	\cdot
	\prod_{\ell=1}^j \gamma_{\sigma(\ell)}^{\beta_{\sigma(\ell)}}
  \end{eqnarray*}
  The case $j=k$ proves the claim.

Part~4 ($\chi_f$).
  The basis $j=1$ follows
  from Lemma~\ref{lem-uni-poly-ana}:
  \begin{eqnarray*}
    \chi_{f_1}\left(\gamma_{\sigma(1)}\right)
      &:=&
        2\left(\delta_{\sigma(1)}-\beta_{\sigma(1)}\gamma_{\sigma(1)}\right).
  \end{eqnarray*}
  Now let $1<j\le k$.
  Because the argument list of $g_j$ and $h_j$ are disjoint,
  we apply Formula~(\ref{for-beta-product-rule}) and obtain
  \begin{eqnarray*}
    \chi_{f_j}\left(\gamma_{\sigma(1)},\ldots,\gamma_{\sigma(j)}\right)
      &:=&
        \prod_{\ell=1}^j
        2\left(\delta_{\sigma(\ell)}-\beta_{\sigma(\ell)}\gamma_{\sigma(\ell)}\right).
  \end{eqnarray*}
  The case $j=k$ proves the claim.
  \qed
\end{proof}

\subsection*{The analysis}

  Now we prove the analyzability of multivariate polynomials
  and apply the approach of quantified relations
  to derive a precision function.
\begin{theorem}[multivariate polynomial]\label{theo-multipoly-ana}
  Let $f$ be a $k$-variate polynomial ($k\ge2$) of total degree $d$
  as defined in Lemma~\ref{lem-multi-poly-suit} and
  let $\PREDL$ be a predicate description for $f$
  with cubical neighborhoods
  $\delta_i=\delta_j$ and $\gamma_i=\gamma_j$
  for all $1\le i,j \le k$.
  Then $f$ is analyzable.
  Furthermore,
  we obtain the bounding function
  \begin{eqnarray}\label{for-theo-multi-lsafe}
    \LS(p)
      &\ge&
        \left\lceil
          - \beta^* \log_2 \left(1-\sqrt[k]{p}\right)
	  \;+\; \CM(\beta)
        \right\rceil
  \end{eqnarray}
  where
  \begin{eqnarray*}
      \CM(\beta) &:=&
      \log_2\frac
      {(d+1+\lceil\log_2|\Ind|\rceil)
        \cdot |\Ind|
	\cdot \max_{\iota\in\Ind}|a_\iota|
	\cdot 2^{\E d+{\beta^*}+1}
	\cdot \hat\beta^{\beta^*}}
      {|a_\beta| \cdot
	{(t\delta_1)}^{\beta^*}}.
  \end{eqnarray*}
  for $\beta\in\IM$
  and $\hat\beta:=\max_{1\le i\le k}\beta_i$
  and $\beta^*:=\sum_{1=i}^{k} \beta_i$.
\end{theorem}
  We observe that $\hat\beta\le d$ and $\beta^*\le d$.
  Note that the choice of $\beta\in\IM$ is an optimization problem:
  We suggest to choose $\beta$
  such that the constant $\beta^*$
  in the asymptotic bound
  $\LS(p) = O\left(-\beta^*\log(1-\sqrt[k]{p})\right)$
  for $p\to1$
  is small.
\begin{proof}
  Part~1 (analyzable).
  Let $\beta\in\IM$.
  Due to Lemma~\ref{lem-multi-poly-suit},
  $f$ is region-value-suitable.
  In addition Corollary~\ref{cor-multipolysafety}
  provides a fp-safety bound $\Sinf(L)$
  for $k$-variate polynomials
  in Formula~(\ref{for-sfl-multivariate}).
  The function $\Sinf(L)$ converges to zero and is invertible.
  It follows that $f$ is safety-suitable and thus analyzable.

  Part~2 (analysis).
  We apply the approach of quantified relations.
  Let $\delta_1,\gamma_1\in\R_{>0}$ and
  $\delta_1=\delta_i$ and $\gamma_1=\gamma_i$
  for all $1\le i\le k$.
  In addition let $\hat\beta := \max_{1\le i\le k} \beta_i$.
  Step~$1'$:
  At first we derive an upper bound $\varepsilon_\chi$ on the volume
  of the complement of the region of uncertainty
  according to the precision $p$.
  Naturally we obtain
  \begin{eqnarray*}
    {\varepsilon_{\chi}}(p)
      &:=& p \prod_{i=1}^{k} 2\delta_i
      \;\; = \;\; p \left(2\delta_1\right)^k.
  \end{eqnarray*}
  Step~$2'$:
  Because of the cubical neighborhood we redefine
  \begin{eqnarray*}
    \chi_f(\gamma) &:=&
      2^k \left( \delta_1 - \hat\beta \gamma_1 \right)^k.
  \end{eqnarray*}
  Then we use $\varepsilon_\chi$ and $\chi_f$ to determine $\gamma_1$:
  \begin{eqnarray*}
       \chi_f(\gamma) &=& {\varepsilon_{\chi}(p)} \\
    \Leftrightarrow
       \makebox[3.0cm]
        {\hfill $2^k \left(\delta_1-\hat\beta\gamma_1\right)^k$}
        &=& p \, 2^k \, \delta_1^k\\
    \Leftrightarrow
       \makebox[3.0cm]
        {\hfill $\left(1-\frac{\hat\beta\gamma_1}{\delta_1}\right)^k$}
        &=& p \\
    \Rightarrow
       \makebox[3.0cm]
       {\hfill $1-\frac{\hat\beta\gamma_1}{\delta_1}$}
        &=& \sqrt[k]{p} \\
    \Leftrightarrow
       \makebox[3.0cm]
      {\hfill $\gamma_1(p)$}
        &:=& \frac{\delta_1\left(1-\sqrt[k]{p}\right)}{\hat\beta}
  \end{eqnarray*}
  Step~3:
  Since $\gamma$ represents the augmented region of uncertainty,
  the normal sized region is induced by $t\gamma$. \\
\noindent
  Step~4:
  Now we fix the bound $\varphi_f$ on the absolute value and set
  \begin{eqnarray*}
    \varphi(p)
      &=& \varphi_f(t\gamma(p)) \\
      &=& |a_\beta| \cdot \left(t\gamma(p)\right)^\beta \\
      &=& |a_\beta| \cdot \prod_{i=1}^k \left(t\gamma_i(p)\right)^{\beta_i} \\
      &=& |a_\beta| \cdot \left(t\gamma_1(p)\right)^{\beta^*} \\
      &=& |a_\beta| \cdot
        \left(\frac{t\delta_1 \left(1-\sqrt[k]{p}\right)}{\hat\beta}\right)^{\beta^*}
  \end{eqnarray*}
  where $\beta^*:=\sum_{i=1}^k \beta_i$. \\
\noindent
  Step~5:
  To derive the bound on the precision,
  we consider the inverse of Formula~(\ref{for-sfl-multivariate})
  which is
  \begin{eqnarray*}
      \Sinf^{-1}(\varphi(p))
    &=&
      \log_2\frac
      {(d+1+\lceil\log_2|\Ind|\rceil)
        \cdot |\Ind| \cdot \max|a_\iota| \cdot 2^{\E d+1}}
      {\varphi(p)}\\
    &=&
      \log_2\frac
      {(d+1+\lceil\log_2|\Ind|\rceil)
        \cdot |\Ind| \cdot \max|a_\iota| \cdot 2^{\E d+1}
        \cdot (2\hat\beta)^{\beta^*} }
      {|a_\beta| \cdot ({t\delta_1\left(1-\sqrt[k]{p}\right)})^{\beta^*}}\\
    &=&
      -{\beta^*} \log_2
      {\left(1-\sqrt[k]{p}\right)}
      \\
    && + \log_2\frac
      {(d+1+\lceil\log_2|\Ind|\rceil)
        \cdot |\Ind| \cdot \max|a_\iota| \cdot 2^{\E d+1}
        \cdot (2\hat\beta)^{\beta^*} }
      {|a_\beta| \cdot {(t\delta_1)}^{\beta^{*}}}.
  \end{eqnarray*}
  Finally the claim follows from
  $\LS(p):=\left\lceil \Sinf^{-1}(\varphi(p))\right\rceil$.
  \qed
\end{proof}
  The formula for $\LS(p)$
  in the lemma above looks rather complicated.
  Therefore we study the asymptotic behavior
  $\LS(p) = O\left(-d\log(1-\sqrt[k]{p})\right)$
  for $p\to1$
  in the following corollary:
  We show that ``slightly'' more than
  $d$ additional bits of the precision are sufficient
  to halve the failure probability.
\begin{corollary}
  Let $f$ be a $k$-variate polynomial ($k\ge 2$) of total degree $d$
  and let $\LS:(0,1)\to\N$
  be the precision function in Formula~(\ref{for-theo-multi-lsafe}).
  Then
  \begin{eqnarray*}
    \LS\left(\frac{1+p}{2}\right)
      &\le&
        \LS(p) + \left\lceil \lambda\beta^* \right\rceil
  \end{eqnarray*}
  where $\beta^*=\sum_{i=1}^{k}\beta_i\le d$ and
  \begin{eqnarray*}
    \lambda &:=&
      \LOG \left(
        \frac{1 - \sqrt[k]{p}}
	  {1 - \sqrt[k]{\frac{1+p}{2}}}
      \right).
  \end{eqnarray*}
\end{corollary}
\begin{proof}
  All quantities are as defined in Theorem~\ref{theo-multipoly-ana}.
  We obtain:
  \begin{eqnarray*}
    \LS\left(\frac{1+p}{2}\right)
    &=&
      \left\lceil
      -\beta^*
        \LOG \left(
          1-\sqrt[k]{\frac{1+p}{2}}
	\right)
	\; + \; \CM(\beta)
      \right\rceil \\
    &=&
      \left\lceil
      -\beta^*
        \LOG \left( \left(1-\sqrt[k]{p}\right)\cdot
          \frac{1-\sqrt[k]{\frac{1+p}{2}}}{1-\sqrt[k]{p}}
	\right)
	\; + \; \CM(\beta)
      \right\rceil \\
    &=&
      \left\lceil
      -\beta^*
        \LOG \left(1-\sqrt[k]{p}\right)
      -\beta^*
	\LOG
          \left(\frac{1-\sqrt[k]{\frac{1+p}{2}}}{1-\sqrt[k]{p}}
	\right)
	\; + \; \CM(\beta)
      \right\rceil \\
    &\le&
      \LS(p)
      +\left\lceil
        \beta^*
	\LOG
          \left(\frac{1-\sqrt[k]{p}}{1-\sqrt[k]{\frac{1+p}{2}}}
	\right) \right\rceil.
  \end{eqnarray*}
  This proves the claim.
  \qed
\end{proof}

\section[The Top-down Approach Using Replacements (1st Stage, rv-suit)]{The Top-down Approach Using Replacements}
\label{sec-top-down}

  This approach derives the bounding functions
  which are associated with region- and value-suitability
  in the first stage of the analysis
  (see Figure~\ref{fig-ana-td}).
  In the bottom-up approach
  we consider a sequence of functions
  which is incrementally built-up from simple functions and
  \emph{ends up} at the function $f$ under consideration.
  In contrast to that
  we now construct a sequence of functions top-down that \emph{begins} with $f$
  and leads to a (different) sequence
  by dealing with the arguments of $f$ coordinatewise.
  However,
  the top-down approach works in two phases:
  In the first phase we just derive the auxiliary functions and
  in the second phase
  we determine the bounds for the region- and value-suitability
  bottom-up.
  That is why we call this approach also \emph{pseudo-top-down}.
  \begin{figure}[h]\centering
    \includegraphics[width=.45\columnwidth]{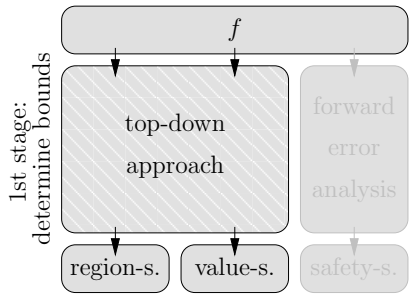}
    \caption{The top-down approach and its interface.}
    \label{fig-ana-td}
  \end{figure}

  We remark that the idea of developing a top-down approach is not new:
  The idea was first introduced by Mehlhorn et al.~\cite{MOS06} and
  their journal article appeared in~\cite{MOS11}.
  \emph{As opposed to previous publications,
  our top-down approach is different for several reasons:
  It is designed to fit to the method of quantified relations
  and
  it is based on our general conditions to analyze functions
  (we do not need auxiliary constructions like exceptional points,
  continuity or a finite zero set).}

  New definitions are introduced in Section~\ref{sec-td-preparation}.
  We define the basic idea
  of a \emph{replacement} in Section~\ref{sec-td-single-rep}.
  Afterwards
  we show how we can apply a \emph{sequence of replacements}
  to the function under consideration
  in Section~\ref{sec-td-multi-rep}.
  We present the top-down approach to derive the bounding functions
  in Section~\ref{sec-td-first-analysis}.
  Next we consider an example in Section~\ref{sec-td-example}.
  For clarity,
  we finally answer selected questions
  in Section~\ref{sec-td-deeper-insight}.

\subsection{Definitions}
\label{sec-td-preparation}

  We prepare the presentation with various definitions
  and begin with a projection.
  Let $\ell,k\in\N$ with $\ell\le k$, let $I:=\{1,\ldots,k\}$
  and let
  \begin{eqnarray*}
    s:\{1,\ldots,\ell\}\to I
  \end{eqnarray*}
  be an injective mapping.
  Then we define the projection
  \begin{eqnarray*}
    \pi_s(x)\label{def-projection-pi}
      &:=&
        \left(x_{s(1)},\ldots,x_{s(\ell)}\right).
  \end{eqnarray*}
  In a natural way we extend the projection to sets $X\subset\R^k$ by
  \begin{eqnarray*}
    \pi_s(X)
      &:=&  \left\{
              \pi_s(x) \ST x\in X
            \right\}.
  \end{eqnarray*}
  Since we often make use of the projection $\pi$
  in the context of an index $i\in I$,
  we define the following abbreviations in their obvious meaning:
  \begin{eqnarray*}
    \PII (x) &:=&
      (x_i), \\
    \PIL (x) &:=&
      (x_1,\ldots,x_{i-1}), \\
    \PIG (x) &:=&
      (x_{i+1},\ldots,x_k), \\
    \PIN (x) &:=&
      (x_1,\ldots,x_{i-1},x_{i+1},\ldots,x_k).
  \end{eqnarray*}
  We remark on this contextual definitions
  that the greatest index $k$
  is always given implicitly by the set $I$ of indices.
  The usage of such orthogonal projections leads to
  the following condition on the set $A$ of projected inputs:
  \emph{It is a necessary condition in the top-down analysis
  that $A$ as well as the perturbation area $\U_\delta(A)$
  are closed axis-parallel boxes without holes.}

  We briefly motivate the next notation:
  Assume that the function $f$ has a $k$-ary argument.
  During the analysis of $f$,
  we often bind $k-1$ of these variables to values
  given in a $(k-1)$-tuple, say $\xi$.
  We do this to study the local behavior of $f$
  in dependence on a single free argument, say $x_i$.
\begin{definition}[free-variable star]
  \label{def-free-var-star}
  Let $\PREDB$ be a predicate description
  where $A$ is an axis-parallel box without holes.
  In addition
  let $I:=\{1,\ldots,k\}$ and
  let $i\in I$.
  For each $(k-1)$-tuple
  $\xi := (\xi_1,\ldots,\xi_{i-1},\xi_{i+1},\ldots,\xi_k) \in \PIN(A)$
  we define the function $\FSIX(x_i)$ as
  \begin{eqnarray*}
    & \FSIX : \PII(A) \to \R,\\
    & x_i \, \mapsto \,
      \FSIX(x_i) =
      f(x_1,\ldots,x_k)|_{x_j=\xi_j \; \forall j\in I\!, \; j\neq i}
      \,=\,
      f(\xi_1,\ldots,\xi_{i-1},x_i,\xi_{i+1},\ldots,\xi_k).
  \end{eqnarray*}
\end{definition}
  In other words,
  we consider $\FSIX$
  as the function $f$
  where $x_i$ is a free variable
  and all remaining variables are bound to the tuple $\xi$.
  We illustrate the definition with an example and
  consider the function $f(x_1,x_2,x_3):=3x_1^2+2x_2^3-4x_3$.
  Then $f^{*2}_{(4,7)}$ is a function in $x_2$ and we have
  \begin{eqnarray*}
    f^{*2}_{(4,7)}(x_2)
    &=& f(x_1,x_2,x_3)|_{x_1=4 \,\wedge\, x_3=7}\\
    &=& 3 \cdot 4^2 + 2 x_2^3 - 4 \cdot 7 = 2 x_2^3-20.
  \end{eqnarray*}
  We sometimes do not attach the tuple $\xi$ to $\FSI$
  to relieve the reading
  if $\xi$ is uniquely defined by the context.

  Once we focus on the function $\FSIX$ in one variable, say $x_i$,
  we are interested in its induced critical set.
  Surely this critical set depends on the choice of $\xi$.
  We have seen that the region-suitability is a necessary condition
  for the analyzability of the function.
  Therefore the next definition is used to mark those $\xi$
  for which $\FSIX$ is or is not region-suitable.
\begin{definition}[region-regularity]
  \label{def-reg-reg}
  Let $\PREDB$ be a predicate description
  where $A$ is an axis-parallel box without holes.
  We call $\xi\in \PIN(A)$ \emph{region-regular}
  if $\FSIX$ is region-suitable on $\PII(A)$.
  Otherwise we call $\xi$ \emph{non-region-regular}.
\end{definition}
  Finally we remark that the region-suitability of $\FSIX$ implies that
  the functions $\nu_{\FSIX}$ and $\chi_{\FSIX}$ exist.
  If $i$ is fixed, there are families of functions $\FSIX$
  (and hence families of functions $\nu_{\FSIX}$ and $\chi_{\FSIX}$)
  that depend on the region-regular $\xi$.
  We examine these families in the next paragraph.

\subsection{Single Replacement}
\label{sec-td-single-rep}

  From now on we consider the following setting:
  \emph{Let $\PREDB$ be a predicate description
  where $A$ is an axis-parallel box without holes
  and let $I:=\{1,\ldots,k\}$.}
  In addition we denote the domain of $f$ by $\DOM(f)$.

  We develop the top-down approach step-by-step.
  For a given index $i\in I$,
  our first aim is to lower-bound the absolute value of $f$
  by a function $g$
  whose argument lists differ solely in the $i$-th position:
  While $f$ depends on $x_i\in \PII(U_\delta(A))$,
  the function $g$ depends on
  a new variable $\gamma_i\in\PII(\GAB)$.
  Hence we say that the construction of $g$
  is motivated by the \emph{replacement of $x_i$ with $\gamma_i$}
  in the argument list of $f$.

  Now we present the construction of the function $g$
  for a fixed index $i\in I$.
  We focus on the functions $\FSIX$
  to study the local behavior of $f$ in its \mbox{$i$-th} argument.
  We are interested in tuples $\xi\in\PIN(\DOM(f))$
  for which $\FSIX$ is region-suitable.
  We collect these points in the set
  \begin{eqnarray*}
    \X_{f,i} &:=& \left\{ \xi\in\PIN(\DOM(f)) : 
    \text{$\xi$ is region-regular}
    \right\}.
  \end{eqnarray*}
  To understand our interest in the set $\X_{f,i}$,
  we remind ourselves about the following fact:
  For region-regular $\xi$,
  open neighborhoods of the critical set $C_{\FSIX}$
  are guaranteed to exist for any given (arbitrarily small) volume.
  This is not true for non-region-regular points
  which therefore must belong to the critical set
  of the objective function.
  Next we define the objective function $g$. Let
  \begin{eqnarray*}
    g :
      \PIL (\DOM(f))
      \times \PII (\GAB)
      \times \PIG (\DOM(f))
      \to \R_{\ge 0},
  \end{eqnarray*}
  be the function with the pointwise definition
  \begin{eqnarray}
    g(\xi_1,\ldots,\xi_{i-1},\gamma_i,\xi_{i+1}\ldots,\xi_k)
      &:=& \left\{
    \begin{array}{l@{\quad:\quad}l}
      0 & \text{$\xi\not\in X_{f,i}$}
      \\[0.5ex]
      \inf\limits_{\text{(C1)}}
      \;
      \inf\limits_{\text{(C2)}}
      \;
      \left|\FSIX(x_i)\right|
      & \text{$\xi\in X_{f,i}$}
    \end{array} \right. \label{for-g-inf-inf} \\[1ex]
    \text{(C1)} &:& {\x_i\in \PII(A)} \nonumber \\
    \text{(C2)} &:& {x_i\in\U_{\FSI\!,\delta_i}(\x_i)
                   \setminus R_{\FSI\!,\gamma_i}(\x_i)} \nonumber
  \end{eqnarray}
  for all $\xi\in\PIN(\DOM(f))$ and all $\gamma_i\in\PII(\GAB)$.
  The domains $\DOM(f)$ and $\DOM(g)$ only differ in the $i$-th coordinate.
  Whenever $\xi$ is non-region-regular,
  we set $g$ to zero.
  (We remark that this is essential for the sequence of replacements
  in Section~\ref{sec-td-multi-rep} since
  this handling triggers the exclusion of
  an open neighborhood of $\xi$---and
  not just the exclusion of the point $\xi$ itself.)
  In case $\xi$ is region-regular,
  we set $g$ to the infimum of the absolute value of $f$
  outside of the region of uncertainty for the various $\x_i$.
  Note that we must consider the infimum
  in the definition of $g$ in Formula~(\ref{for-g-inf-inf})
  because $|\FSIX|$ does not need to have a minimum.
  We do not assume that $f$ is continuous or semi-continuous.
\begin{definition}
  We call the presented construction of the function $g$ the
  \emph{function resulting from the replacement of $f$'s argument
  $x_i$ with $\gamma_i$}.
  We denote the replacement by
  $\REP(f,{x_i} \to \gamma_{i}).$
\end{definition}
  We summarize the steps during the replacement of an argument of $f$
  and emphasize the relation between the quantities:
  Let $f$ be given.
  Then we begin with the consideration of the auxiliary function $f^{*i}_{\xi}$.
  We use it to determine the auxiliary set
  of region-regular points $\X_{f,i}$.
  To determine the function $g$ afterwards,
  we examine $f^{*i}_{\xi}$ again,
  but now only for the points in $\X_{f,i}$.

  In the proof of the analysis in Section~\ref{sec-td-first-analysis},
  we use the statement that the replacement $\REP(f,{x_i} \to \gamma_{i})$
  results in a positive function that lower bounds the absolute value of $f$
  in a certain sense.
  We formalize and prove this statement in the next lemma.
\begin{lemma}\label{lem-sequence-of-lower-bounds}
  Let $\PREDB$ be a predicate description
  where $A$ is an axis-parallel box without holes,
  let $I:=\{1,\ldots,k\}$ and
  let $i\in I$.
  Moreover, let $g:=\REP(f,x_i \to \gamma_i)$.
  Then we have
  \begin{eqnarray}\label{for-f-larger-rep}
    |f(\xi_1,\ldots,\xi_{i-1},x_i,\xi_{i+1}\ldots,\xi_k)|
    \ge g(\xi_1,\ldots,\xi_{i-1},\gamma_i,\xi_{i+1}\ldots,\xi_k)
    > 0
  \end{eqnarray}
  for all region-regular points $\xi\in \X_{f,i}$,
  for all $\gamma_i\in\PII(\GAB)$,
  for all ${\x_i\in \PII(A)}$
  and for all
  ${x_i\in\U_{\FSI\!,\delta_i}(\x_i) \setminus R_{\FSI\!,\gamma_i}(\x_i)}$.
\end{lemma}
\begin{proof}
  The left unequation in Formula~(\ref{for-f-larger-rep})
  follows immediately from the construction of
  the function $g=\REP(f,x_i \to \gamma_i)$
  because we only consider points
  lying outside of the region of uncertainty $R_{\FSI,\gamma_i}(\x_i)$.

  To prove the right unequation
  in Formula~(\ref{for-f-larger-rep}),
  we assume that there is a region-regular $\xi\in\X_{f,i}$ 
  and $\gamma_i\in\PII(\GAB)$ such that
  the objective function
  \mbox{$g(\xi_1,\ldots,\xi_{i-1},\gamma_i,\xi_{i+1}\ldots,\xi_k)=0$}.
  This implies, for $\x_i\in\PII(A)$,
  the existence of a sequence $(a_j)_{j\in \N}$
  in the area
  ${\U_{\FSI\!,\delta_i}(\x_i) \setminus R_{\FSI\!,\gamma_i}(\x_i)}$
  for which $\lim_{j\to\infty} \FSIX(a_j) = 0$.
  Consequently $a:=\lim_{j\to\infty} a_j$ must belong to the critical set.
  Since the region of uncertainty $R_{\FSI\!,\gamma_i}$
  guarantees the exclusion of the open $\gamma_i$-neighborhood
  of the critical set---which
  includes the open $\gamma_i$-neighborhood of $a$---almost
  all points of the sequence $(a_j)_{j\in \N}$
  must also lie in $R_{\FSI\!,\gamma_i}$.
  This leads to a contradiction to the assumption
  and proves the claim.
  \qed
\end{proof}
  We add the remark that
  the right unequation in Formula~(\ref{for-f-larger-rep})
  presumes that $\xi$ is region-regular
  as is stated in the lemma.
  We obtain $g\equiv 0$
  if $\X_{f,i}$ is the empty set.
  We continue with a simple example
  that illustrates the method to determine
  $\REP(f,x_i \to \gamma_i)$.
\begin{example}
  Let $f (x_1,x_2) = x_1^2 + x_2^2$.
  Then $I=\{1,2\}$.
  In addition let $i=2$ and
  let $A$ be an axis-parallel rectangle that contains the origin $(0,0)$.
  We consider $f^{*2}_{\xi_1} (x_2) = \xi_1^2 + x_2^2 $.
  Since $f^{*2}$ is region-suitable,
  this leads to $\X_{f,2}=\pi_{\neq 2}(A) = \pi_{1}(A)$.
  We obtain
  \begin{eqnarray*}
    g (\xi_1, \gamma_2)
      &:=& \left\{
      \begin{array}{l@{\quad:\quad}l}
        \gamma_2^2 & \xi_1=0 \\[0.5ex]
	\xi_1^2 & \text{otherwise.}
      \end{array}
      \right.
  \end{eqnarray*}
  The critical set of $g$ contains a single point in the case $\xi_1=0$
  and is empty in the other case.
  \hfill$\bigcirc$
\end{example}
  We end this subsection with two observations.
  Firstly, although $g(\xi_1,\gamma_2)>0$ in the example above,
  the limit
  \begin{eqnarray*}
    \inf_{\xi_1\in\X_{f,2}\,\wedge\,\xi_1\neq0} \; g(\xi_1,\gamma_2) &=& 0.
  \end{eqnarray*}
  Secondly, if the lower-bounding function $g$
  is region-value-suitable,
  the function $f$ is also region-value-suitable
  because of Theorem~\ref{theo-lower-bound}.
  This observation is the driving force of the top-down approach.

\subsection{Sequence of Replacements}
\label{sec-td-multi-rep}
  So far we know how a variable $x_i$
  of the argument list of the function $f$ under consideration
  can be replaced with a new variable $\gamma_i$.
  The advantage of the new variable $\gamma_i$ is
  that it reflects the distance to the critical set, somehow.
  We announce that, opposed to $x_i$,
  the variable $\gamma_i$ is appropriate for the analysis.
  A benefit of $\gamma_i$ is that
  it is not necessary
  to study the precise location of the critical set;
  the knowledge about the ``width'' of the critical set is sufficient.

  The idea behind the top-down approach is to apply
  the replacement procedure
  $k$ times in a row to replace all original arguments $(x_1,\ldots,x_k)$ of $f$
  by the new substitutes $(\gamma_1,\ldots,\gamma_k)\in\GAB$.
  To get the presentation as general as possible,
  we keep the order of the $k$ replacements variable.
  Let $\sigma: I \to I$ be a bijective function
  that defines the order
  in which we replace the arguments of $f$.
  We interpret $\sigma(i)=j$ as
  the replacement of $x_j$ with $\gamma_j$ in the $i$-th step.

  Now we look for a recursive definition
  to derive the sequence $g_1, \ldots, g_k$ of functions
  that result from these replacements.
  We define the basis of the recursion as
  $g_0:=f$ with $g_0:\U_\delta(A)\to\R$
  and $\DOM(g_0)=\U_\delta(A)$.
  We set
  $g_i := \REP(g_{i-1}, x_{\sigma(i)} \to \gamma_{\sigma(i)})$
  for $i\in I$.
  In other words:
  We focus on the replacement of $x_{\sigma(i)}$ in step $i\in I$,
  that means,
  we assume that we have just derived the functions $g_1,\ldots,g_{i-1}$.
  We then
  determine the set
  of region-regular points
  \begin{eqnarray*}
    \X_{g_{i-1},\sigma(i)} &:=& \left\{ \xi\in\pi_{\neq\sigma(i)}(\DOM(g_{i-1})) : 
    \text{$\xi$ is region-regular}
    \right\}\!,
  \end{eqnarray*}
  that means, we check if the function
  \begin{eqnarray*}
   & g_{i-1,\xi}^{*\sigma(i)} :
       \pi_{\sigma(i)} (\DOM(g_{i-1})) \to \R_{\ge 0}, \\
   &  g_{i-1,\xi}^{*\sigma(i)}\!\left(x_{\sigma(i)}\right) \mapsto
      g_{i-1}\!\left(\xi_1,\ldots,\xi_{\sigma(i)-1},x_{\sigma(i)},\xi_{\sigma(i)+1},\ldots,\xi_k\right)%
  \end{eqnarray*}
  is region-suitable for a given $\xi$.
  Thereafter, we
  define the domain of the succeeding function
  $g_i$ as
  \begin{eqnarray*}
    & g_i :
       \pi_{<\sigma(i)} (\DOM(g_{i-1}))
       \times \pi_{\sigma(i)} (\GAB)
       \times \pi_{>\sigma(i)} (\DOM(g_{i-1})) \to \R_{\ge 0}
  \end{eqnarray*}
  and
  use $\X_{g_{i-1},\sigma(i)}$ to
  define
  $g_i(\xi_1,\ldots,\xi_{\sigma(i)-1},\gamma_{\sigma(i)},\xi_{\sigma(i)+1}\ldots,\xi_k)$
  \begin{eqnarray}
      &:=& \left\{
    \begin{array}{l@{\quad:\quad}l}
      0 & \text{$\xi\not\in X_{g_{i-1},\sigma(i)}$}
      \\[0.5ex]
      \inf\limits_{\text{(C1)}}
      \;
      \inf\limits_{\text{(C2)}}%
      \;
      \left|g_{i-1}^{*\sigma(i)}(x_{\sigma(i)})\right|
      & \text{$\xi\in X_{g_{i-1},\sigma(i)}$}
    \end{array} \right. \label{for-gi-inf-inf} \\[1ex]
    \text{(C1)} &:& {\x_{\sigma(i)}\in \pi_{\sigma(i)}(\DOM(g_{i-1}))} \nonumber \\
    \text{(C2)} &:& {x_{\sigma(i)}\in\U_{g_{i-1}^{*\sigma(i)},\delta_{\sigma(i)}}(\x_{\sigma(i)})
                   \setminus R_{g_{i-1}^{*\sigma(i)},\gamma_{\sigma(i)}}(\x_{\sigma(i)})} \nonumber
  \end{eqnarray}
  for all $\xi\in{\pi_{\neq \sigma(i)}}(\DOM(g_{i-1}))$ and all $\gamma_{\sigma(i)}\in\pi_{\sigma(i)}(\GAB)$.
  We summarize the relation between the quantities
  during the $i$-th replacement
  in Figure~\ref{fig-ana-depend-multi-rep}.
  (The striped quantities are introduced later.)
  \begin{figure}[t]\centering
    \includegraphics[width=.95\columnwidth]{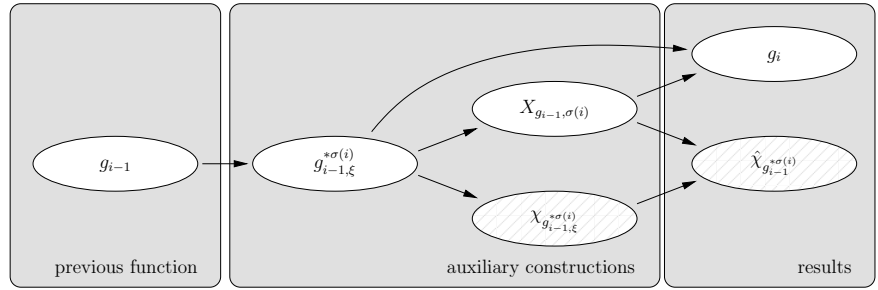}
    \caption{Illustration of the dependencies during the $i$-th replacement.
    The white-colored quantities are defined in Section~\ref{sec-td-multi-rep}
    and the striped quantities in Section~\ref{sec-td-first-analysis}.
    Here ``$\!A\to B$'' means that $B$ is derived from $A$.}
    \label{fig-ana-depend-multi-rep}
  \end{figure}

  The definitions above are chosen such that
  the function $g_i$ exists.
  After the \mbox{$k$-th} step,
  the recursion ends with
  $g_k:\GAB \to \R_{\ge 0}$.
  We remark that,
  if we apply this mechanism to functions
  which are not admissible for controlled perturbation,
  the sequence of replacements will end-up with a function $g_k$
  that fails the analysis from the next section.
\begin{example}
  \label{ex-inbox-gi-dom-crit}
  We get back to the 2-dimensional $\inbox$-predicate.
  For this example it is sufficient to assume
  that the box is fixed somehow and
  that the only argument of the predicate is the query point $q=(x_1,x_2)$.
  This time we consider the various domains and critical sets
  of the functions $g_i$ that result from the sequence of replacements.
  (The order of the replacements is not important for this example.)
  The situation is illustrated in Figure~\ref{fig-crit-set-inbox-1}.
  Picture (a) shows the domain (shaded region)
  of the function $f=g_0$ itself.
  We know that the critical set is the boundary of the query box.

  After the replacement $\REP(g_0,x_1\to\gamma_1)$,
  the first argument belongs to the set $\pi_1(\GAB)$
  resulting in an altered domain (see Picture (b)).
  We make two observations.
  Firstly, the critical set of $g_1$ is formed by two horizontal lines
  that are caused by the top and bottom part of the box $C_{g_0}$.
  What is the reason for that?
  If we consider the absolute value of $g_0$
  while moving its argument
  along a horizontal line that
  passes through the top or bottom line segment of the box
  ($x_2$ is fixed then),
  it leads to a mapping that is zero on an open interval;
  in this case the mapping cannot be region-suitable.
  Secondly, there are no further contributions to the critical set of $g_1$.
  What is the reason?
  If we consider the absolute value of $g_0$ along a horizontal line that
  passes through the interior of the box,
  it leads to a mapping which is region-suitable.%
  \begin{figure}[t]\centering%
    \includegraphics[width=.95\columnwidth]{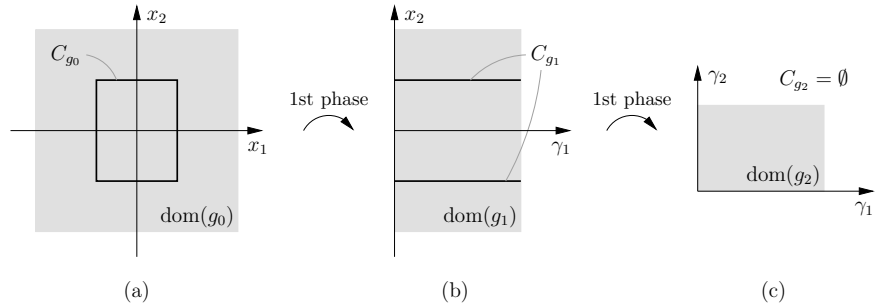}%
    \caption{Illustration of the various domains and critical sets
      that result from the sequence of replacements
      for the 2-dimensional $\inbox$-predicate.}%
    \label{fig-crit-set-inbox-1}%
  \end{figure}%

  Picture (c) shows the situation after the second replacement
  $\REP(g_1,x_2\to\gamma_2)$.
  The function $g_2$ is positive on its entire domain $\GAB$.
  The reason for this is that,
  if we consider the absolute value of $g_1$ along a vertical line
  ($\gamma_1$ is fixed then),
  it leads to a mapping which is region-suitable.
  \hfill$\bigcirc$
\end{example}

\subsection{Derivation and Correctness of the Bounds}
\label{sec-td-first-analysis}

  Although we have replaced each $x_i$ with $\gamma_i$
  in the argument list of $f$ in a top-down manner,
  we are not able to determine the bounds $\nu_f$ and $\varphi_f$ 
  in the same way.
  To achieve this goal,
  we need to go through the collected information bottom-up again.
  The reason is that, at the time we arrive at a function, say $g_{i-1}$,
  we cannot check directly if $g_{i-1}$ is region- and value-suitable.
  Instead of this,
  we want that these properties are inherited from the successor $g_i$
  to the predecessor.
  We will see that,
  once we arrive at $g_k$,
  we can easily check if $g_k$ has the desired properties.
  This way we can possibly show
  that $g_0$, i.e.\ $f\!,$ is also region-value-suitable.

  Therefore we divide the analysis in two phases.
  The first phase consists of the deduction of $g_k$
  via the sequence of replacements
  and is already presented in the last section.
  The second phase consists of the deduction of the bounding functions
  $\varphi_f$ and $\chi_f$ and
  is the subject of this section.
  We begin with an auxiliary statement
  which claims that $g_k$ is non-decreasing in each argument
  under certain circumstances.
\begin{lemma}\label{lem-gk-non-decreasing}
  Let $\PREDB$ be a predicate description
  where $A$ is an axis-parallel box without holes
  and let $I:=\{1,\ldots,k\}$.
  Let $\sigma:I\to I$ be bijective,
  i.e., an order on the elements of $I$.
  Finally,
  let $g_0:=f$ and $g_j:=\REP(g_{j-1},x_{\sigma(j)}\to\gamma_{\sigma(j)})$
  for all $1\le j \le k$,
  i.e.,
  $g_k$ is the resulting function after the $k$ replacements.
  If the function $g_k$ is positive\footnote{%
    That is why we have defined $\GAB$ as an \emph{open} set.},
  it is non-decreasing
  in $\gamma_i$
  on $\PII(\GAB)$
  for all $i\in I$.
\end{lemma}
\begin{proof}
  We refer to the explicit definition of $g_i$ in Formula~(\ref{for-gi-inf-inf})
  that reflects the replacement of the $i$-th argument:
  For growing $\gamma_{\sigma(i)}$
  we shrink the domain for $x_{\sigma(i)}$
  due to condition~(C2).
  Formally,
  for $\gamma',\gamma''\in\pi_{\sigma(i)}(\GAB)$
  with $\gamma'<\gamma''$
  the corresponding regions of uncertainty are related in the way
  \begin{eqnarray*}
    & R_{g_{i-1}^{*\sigma(i)},\gamma'}(\x_{\sigma(i)})
      \;\subset\;
      R_{g_{i-1}^{*\sigma(i)},\gamma''}(\x_{\sigma(i)}).
  \end{eqnarray*}
  Because the function value of $g_i$
  is defined by the infimum absolute value,
  the function $g_i$
  must be non-decreasing in its $i$-th argument $\gamma_{\sigma(i)}$
  for region-regular $\xi$
  by construction.

  The same argumentation is true for each of the $k$ replacements and
  is independent of the actual sequence of replacements.
  This finishes the proof.
  \qed
\end{proof}
  The domain of the function $g_k$ is naturally $\GAB$.
  Even if $\GAB$ has the same cardinality than $\R$
  for $k\ge2$,
  it is non-obvious how to define an invertible function $\chi_{g_k}$
  on $\GAB$.
  But such a bound is required to use the method of quantified relations.
  For that purpose we restrict the domain in the analysis
  to $\GAL$:
  It is true that the elements of $\gamma\in\GAL$
  are now interlinked,
  but the important fact is that
  we can still choose them arbitrarily close to zero.

  To further prepare the analysis,
  we have to focus on a peculiarity of the auxiliary function
  $g_{i-1,\xi}^{*\sigma(i)}$ for a given $i\in I$.
  Remember that
  $\nu_{g_{i-1,\xi}^{*\sigma(i)}}$ and $\chi_{g_{i-1,\xi}^{*\sigma(i)}}$
  are families of functions with parameter
  $\xi\in \X_{g_{i-1},\sigma(i)}$.
  Therefore we are facing the following issue:
  For a given $i\in I$,
  how can we deal with these two families of functions?
  The first solution that comes into mind
  is to replace each family with just one bounding function---so
  this is what we do.
  That means, we define the pointwise limits of these families as
  \begin{eqnarray}
    \hat\nu_{g_{i-1}^{*\sigma(i)}}\left(\gamma_{\sigma(i)}\right) &:=&
      \sup_{\xi\in\X_{f,i}} \; \nu_{g_{i-1,\xi}^{*\sigma(i)}}\left(\gamma_{\sigma(i)}\right) \nonumber
  \end{eqnarray}
  and
  \begin{eqnarray}
    \hat\chi_{g_{i-1}^{*\sigma(i)}}\left(\gamma_{\sigma(i)}\right) &:=&
      \inf_{\xi\in\X_{f,i}} \; \chi_{g_{i-1,\xi}^{*\sigma(i)}}\left(\gamma_{\sigma(i)}\right)
      \label{for-def-chi-star}
  \end{eqnarray}
  for $\gamma\in\GAB$
  and make use of these new bounds in the analysis.
  To illustrate this extra work in the analysis,
  we have added
  the two striped quantities in Figure~\ref{fig-ana-depend-multi-rep}.

  Now we are ready to present the top-down approach to analyze
  real-valued functions.
  We claim and prove the results in the following theorem.
\begin{theorem}[top-down approach]\label{theo-top-down-first-version}
  Let $\PREDB$ be a predicate description
  where $A$ is an axis-parallel box without holes
  and let $I:=\{1,\ldots,k\}$.
  Let $\sigma:I\to I$ be bijective,
  i.e., an order on the elements of $I$.
  Finally,
  let $g_0:=f$ and $g_j:=\REP(g_{j-1},x_{\sigma(j)}\to\gamma_{\sigma(j)})$
  for all $1\le j \le k$.
  We define $\varphi_f$ and $\chi_f$ as
  \begin{eqnarray*}
    \varphi_{f}(\gamma)
      &:=& g_k(\gamma) \\
    \chi_{f}(\gamma)
      &:=& \prod_{j=1}^k \,
        \hat\chi_{g_{j-1}^{*\sigma(j)}}\!\left(\gamma_{\sigma(j)}\right).
  \end{eqnarray*}
  If $g_k$ is positive on $\,\GAB$ and
  $\chi_f$ is invertible on\footnote{%
    Remember that $\GAL\subset\GAB$.}
  $\,\GAL$,
  then $f$ is region-value-suitable
  with the bounding functions\footnote{%
    Remember that we can use $\nu_f$
    instead of $\chi_f$
    because of Formula~(\ref{def-chi-region-suit}).}
  $\varphi_f$ and $\chi_f$.
\end{theorem}
\begin{proof}
  We prove the claim in three parts.
  First we show that
  there are certain bounding functions
  $\varphi_{g_k}$ and $\chi_{g_k}$
  for which $g_k$ is region-value-suitable.
  Afterwards we prove that,
  if the function $g_i$ has such bounding functions,
  then $g_{i-1}$ has also appropriate bounding functions.
  And in the end we deduce the claim of the theorem.

  Part~1 (basis).
  We assume that
  $g_k$ is positive on the open set $\,\GAB$, that means,
  we consider the function $g_k:\GAB\to\R_{>0}$.
  At first we decompose the domain in two parts
  (see Figure~\ref{fig-gamma-box-decomposition}).
  Let $\gamma\in\GAL$.
  We define the unique open axis-parallel box
  with opposite vertices $\gamma$ and\footnote{%
    Remember that we have introduced $\hat\gamma$
    to define $\GABG$ and $\GALG$.
    More information and the formal bound is given
    in Remark~\ref{rem-def-region-suit}.2
    on Page~\pageref{rem-def-region-suit}.}
  $\hat\gamma$ as
  \begin{eqnarray*}
    \GASG &:=&
      \left\{
        \gamma'\in\GAB :
	\text{$\gamma_i \le \gamma'_i$ for all $i\in I$}
      \right\}\!.
  \end{eqnarray*}
  We denote its complement within the $\GAB$ by
  \begin{eqnarray*}
    \GARG &:=& \GAB \setminus \GASG.
  \end{eqnarray*}
  We think of $\GARG$ as the region of uncertainty and
  $\GASG$ as the region
  whose floating-point numbers are guaranteed to evaluate fp-safe.
  \begin{figure}[t]\centering
    \includegraphics[width=.65\columnwidth]{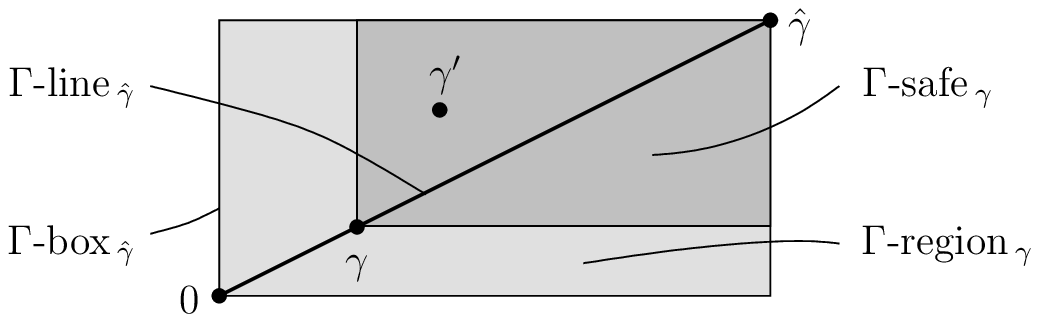}
    \caption{This is an exemplified 2-dimensional illustration
      of the decomposition of the $\GAB$
      into the sets $\GASG$ and $\GARG$
      for $\gamma\in\GAL$.}
    \label{fig-gamma-box-decomposition}
  \end{figure}
  We claim that $g_k$ is region-value-suitable on $\GAB$
  in the following sense:
  We set the bounding functions to
  \begin{eqnarray*}
    \varphi_{g_k}(\gamma) &:=& g_k(\gamma) \\
    \chi_{g_k}(\gamma)
      &:=&
        \prod_{j=1}^k \,
        \left( \hat\gamma_j - \gamma_j \right)
  \end{eqnarray*}
  and claim that two statements are fulfilled for every
  $\gamma\in\GAL$:
  \begin{enumerate}
    \item
      The absolute value of $g_k(\gamma')$ is at least $\varphi_{g_k}(\gamma)$
      for all points $\gamma'\in\GASG$.
    \item
      The volume of $\GASG$ is $\chi_{g_k}(\gamma)$.
  \end{enumerate}
  To prove the first statement,
  we consider the function value of $g_k$
  along a path of $k$ axis-parallel line segments
  from $\gamma$ to $\gamma'$.
  The path starts at
  $\gamma=(\gamma_1,\ldots,\gamma_k)$,
  connects
  the $(k-1)$ points
  $(\gamma'_1,\ldots,\gamma'_j,\gamma_{j+1}\ldots,\gamma_k)$
  with $1\le j<k$
  in ascending order of $j$
  and ends at
  $\gamma'=(\gamma'_1,\ldots,\gamma'_k)$.
  Along this path,
  the function value of $g_k$ is non-decreasing
  because of Lemma~\ref{lem-gk-non-decreasing}:
  For all $i\in I\!,$
  the function $g_k$
  is non-decreasing in its $i$-th argument
  $\gamma_i\in\PII(\GAB)$
  for fixed $\xi\in\PIN(\GAB)$.

  The proof of the second statement is straight forward:
  Because the box is axis-parallel,
  its volume is the product of its edge-lengths.
  We make the observation that the function $\chi_{g_k}(\gamma)$
  is strictly monotonically increasing on $\GAL$
  and hence must be invertible on this domain.

  We conclude the first part of the proof:
  \emph{For a given $\gamma\in\GAL$,
  we have shown
  that the function value of $g_k$
  is at least $\varphi_{g_k}(\gamma)$
  on an area of volume $\chi_{g_k}(\gamma)$.}
  This way we have found evidence
  that $g_k$ is region-value-suitable
  in the meaning above.

  Part~2 (induction).
  We claim:
  \emph{For $i\in I$ and $\gamma\in\GAL$,
  the function value of $g_{i-1}$
  is at least $\varphi_{g_{i-1}}(\gamma)$
  on an area of volume $\chi_{g_{i-1}}(\gamma)$ with}
  \begin{eqnarray}
    \varphi_{g_{i-1}}(\gamma)
      &:=& \varphi_{g_i} (\gamma)
      \label{for-proof-phi-star}
      \\
    \chi_{g_{i-1}}(\gamma)
      &:=& \chi_{g_{i}}(\gamma)
      \, \cdot \,
      \frac { \hat\chi_{g_{i-1}^{*\sigma(i)}} \left(\gamma_{\sigma(i)}\right) }
            {\hat\gamma_{\sigma(i)} - \gamma_{\sigma(i)}}.
      \label{for-proof-chi-star}
  \end{eqnarray}
  We prove the claim by mathematical induction for descending $i\in I$.
  Basis ($i=k$).
  Due to the first part,
  we can base the proof on the bounding functions
  $\varphi_{g_k}$ and $\chi_{g_k}$.
  Induction step ($i\in I$).
  We assume that the bounding functions
  are true for all $j\in I$ with $i\le j\le k$
  and prove the claim for $i-1$.
  This is what we do next.

  Remember the definition
  $g_i := \REP\!\left(g_{i-1}, x_{\sigma(i)} \to \gamma_{\sigma(i)}\right)$.
  In the step backwards
  from $g_i$ to $g_{i-1}$,
  we observe the following difference in their two axis-parallel domains
  due to condition (C2) of Formula~(\ref{for-gi-inf-inf}):
  The counterpart to the situation in which
  the \mbox{$\sigma(i)$-th} argument of $g_i$ lies in
  $\pi_{\sigma(i)}\left(\GASG\right)$
  is the situation in which
  the \mbox{$\sigma(i)$-th} argument of $g_{i-1}$ lies in
  \begin{eqnarray}\label{for-region-regular-gi-1}
    \U_{g^{*\sigma(i)}_{i-1,\delta_{\sigma(i)}}} \!\! \left(\x_{\sigma(i)}\right)
    \; \setminus \;
    R_{g^{*\sigma(i)}_{i-1,\gamma_{\sigma(i)}}} \!\! \left(\x_{\sigma(i)}\right)
  \end{eqnarray}
  and belongs to the region-regular case.
  Furthermore,
  the volume of this area is guaranteed to be at least
  $\hat\chi_{g^{*\sigma(i)}_{i-1}}\left(\gamma_{\sigma(i)}\right)$
  due to Formula~(\ref{for-def-chi-star}).
  Because the axis-parallel domains of $g_i$ and $g_{i-1}$
  do not differ in directions different to the $\sigma(i)$-th main axis,
  their volume (which is the product of edge lengths)
  solely differ in a factor.
  Therefore we can estimate the volume $\chi_{g_{i-1}}(\gamma)$
  at the product $\chi_{g_{i}}(\gamma)$
  where we replace the factor
  $({\hat\gamma_{\sigma(i)} - \gamma_{\sigma(i)}})$
  by $\hat\chi_{g_{i-1}^{*\sigma(i)}}\!(\gamma_{\sigma(i)})$;
  this validates Formula~(\ref{for-proof-chi-star}).

  Because of Lemma~\ref{lem-sequence-of-lower-bounds},
  the lower-bounding function $\varphi_{g_i}$
  is also a lower-bounding function
  on the volume of the area which is defined
  in Formula~(\ref{for-region-regular-gi-1}).
  This validates Formula~(\ref{for-proof-phi-star}).

  Part~3 (conclusion).
  So far we have shown that
  \emph{for a given $\gamma\in\GAL$,
  the function value of $f=g_0$
  is at least $\varphi_{f}(\gamma)$
  on an area of volume $\chi_{f}(\gamma)$ because}
  \begin{eqnarray*}
    \varphi_{f}(\gamma)
      &=&
    \varphi_{g_0}(\gamma)
      \;\; = \;\;
    \varphi_{g_1}(\gamma)
      \;\; = \;\;
    \cdots
      \;\; = \;\;
    \varphi_{g_k}(\gamma)
      \;\; = \;\;
    g_k(\gamma)
  \end{eqnarray*}
  and because
  \begin{eqnarray*}
    \chi_{f}(\gamma)
      &=&
        \chi_{g_{0}}(\gamma) \\
      &=&
        \chi_{g_{1}}(\gamma)
        \, \cdot \,
        \frac { \hat\chi_{g_{0}^{*\sigma(1)}} \left(\gamma_{\sigma(1)}\right) }
            {\hat\gamma_{\sigma(1)} - \gamma_{\sigma(1)}} \\
      &=&
        \chi_{g_{2}}(\gamma)
        \, \cdot \,
        \frac { \hat\chi_{g_{1}^{*\sigma(2)}} \left(\gamma_{\sigma(2)}\right) }
            {\hat\gamma_{\sigma(2)} - \gamma_{\sigma(2)}}
        \, \cdot \,
        \frac { \hat\chi_{g_{0}^{*\sigma(1)}} \left(\gamma_{\sigma(1)}\right) }
            {\hat\gamma_{\sigma(1)} - \gamma_{\sigma(1)}} \\[1.5ex]
      && \vdots \\[0.5ex]
      &=&
        \chi_{g_{k}}(\gamma)
	\, \cdot \,
	\prod_{i=1}^k \,
        \frac { \hat\chi_{g_{i-1}^{*\sigma(i)}} \left(\gamma_{\sigma(i)}\right) }
            {\hat\gamma_{\sigma(i)} - \gamma_{\sigma(i)}} \\
      &=&
        \prod_{j=1}^k \,
        \left( \hat\gamma_j - \gamma_j \right)
	\, \cdot \,
	\prod_{i=1}^k \,
        \frac { \hat\chi_{g_{i-1}^{*\sigma(i)}} \left(\gamma_{\sigma(i)}\right) }
            {\hat\gamma_{\sigma(i)} - \gamma_{\sigma(i)}} \\
      &=&
        \prod_{i=1}^k \,
        \left( \hat\gamma_{\sigma(i)} - \gamma_{\sigma(i)} \right)
	\, \cdot \,
	\prod_{i=1}^k \,
        \frac { \hat\chi_{g_{i-1}^{*\sigma(i)}} \left(\gamma_{\sigma(i)}\right) }
            {\hat\gamma_{\sigma(i)} - \gamma_{\sigma(i)}} \\
      &=&
	\prod_{i=1}^k \,
          { \hat\chi_{g_{i-1}^{*\sigma(i)}} \left(\gamma_{\sigma(i)}\right) }.
  \end{eqnarray*}
  If $\chi_f$ is in addition
  invertible on the domain $\,\GAL$,
  $f$ is region-value-suitable.
  This finishes the proof.
  \qed
\end{proof}
  One prerequisite in the last theorem is
  that $g_k$ is positive on the open $\GAB$.
  We make the observation that we cannot validate this property
  unless we have determined
  the entire sequence of replacements from $f=g_0$ down to $g_k$.
  That means,
  it is possible that the analysis fails at the end of the first phase.

  Furthermore, we make the observation
  that the bounding functions $\varphi_f$ and $\chi_f$
  are actually derived {bottom-up} in the the second phase
  of their derivation.
  That means,
  although we technically determine the sequence of functions $g_i$
  in a top-down manner on the surface,
  the validity of the formulas is derived bottom-up afterwards.
  We summarize the steps of the top-down approach
  in Figure~\ref{fig-analysis-3}.
  \begin{figure}[t]\centering
    \includegraphics[width=.95\columnwidth]{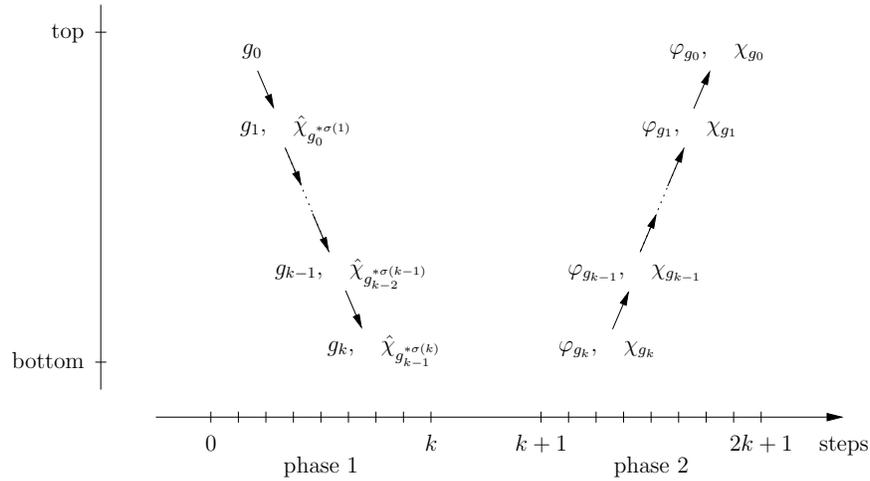}
    \caption{Instructions for performing the top-down approach.
      The illustration reflects the steps
      in which the quantities are determined
      according to Theorem~\ref{theo-top-down-first-version}.}
    \label{fig-analysis-3}
  \end{figure}

\subsection{Examples}
\label{sec-td-example}

\begin{example}\label{ex-inbox-topdown}
  We use the top-down approach to determine the bounding functions
  $\varphi_{\inbox}$ and $\chi_{\inbox}$
  for the predicate $\inbox$.
  Again, we assume that the box is fixed somehow and
  that the only argument of the predicate is the query point.
  (There is no much influence on the analysis by the remaining parameters.)
  The predicate can be realized, for example, by the function
  \begin{eqnarray*}
    f(x) &:=& \min_{1\le i\le k} \, \left\{ \ell_i^2-(x_i-c_i)^2 \right\}
  \end{eqnarray*}
  where $c\in\R^k$ is the center of the axis-parallel box
  and its edge lengths are given by $2\ell$.
  We eliminate the variables in ascending order from
  $x_1$ to $x_k$, that means, we set $\sigma(i):=i$ for all $1\le i\le k$.

  Part~1~($\varphi_{\inbox}$).
  To determine $\varphi_{\inbox}$
  we need $g_k$,
  to determine $g_k$
  we need the entire sequence of replacements,
  and to determine $g_i$
  we need to determine the value of the ``$\inf\inf$'' expression
  in dependence on $\gamma_i$
  in Formula~(\ref{for-gi-inf-inf}).
  This is what we do next.
  Because of the symmetry of $f$,
  the following discussion is valid for all coordinates $x_i$.

  To prepare the replacement of variables,
  we examine the function $f^{*i}$
  for the region-regular case (see Figure~\ref{fig-inbox-1}).
  \begin{figure}[t]\centering
    \includegraphics[width=.95\columnwidth]{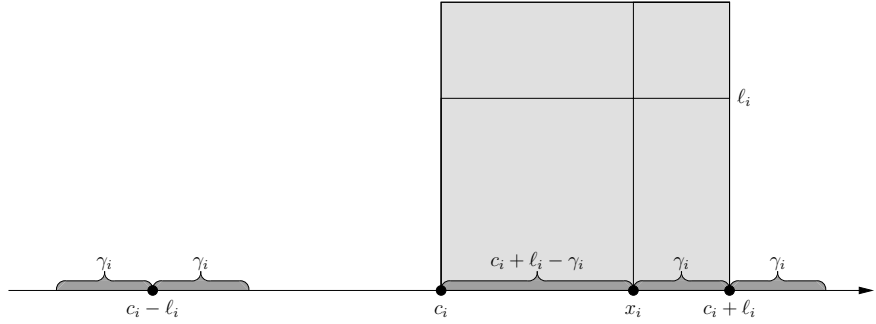}
    \caption{An illustration that supports
             the relation between the quantities of $f^{*i}$
             for the region-regular case of the predicate $\inbox$.}
    \label{fig-inbox-1}
  \end{figure}
  The critical set $C_{f^{*i}}$ contains two points,
  namely $c_i-\ell_i$ and $c_i+\ell_i$.
  By $\gamma_i$ we denote the minimal distance of $x_i$
  to a point in $C_{f^{*i}}$.
  (Again
  we assume that $\hat\gamma_i$ must be less than $\ell_i$;
  otherwise the interior of the box
  would be covered entirely by the region of uncertainty
  and the predicate would lose its meaning.)
  The absolute value of $f$ grows in the distance to $C_{f^{*i}}$.
  To determine a guaranteed lower bound on the absolute value of $f$,
  we assume that the distance of $x_i$ to $C_{f^{*i}}$ is exactly $\gamma_i$.
  In addition we make the observation
  that $|f|$ grows slower towards the interior of the box
  than away from the box;
  therefore we must also assume that $x_i$
  lies between $c_i-\ell_i$ and $c_i+\ell_i$
  to get a convincing bound.
  This leads to the worst-case consideration $|x_i-c_i|=|c_i+\ell_i-\gamma_i|$.
  We make use of the binomial theorem
  to derive the unequation
  \begin{eqnarray*}
    \left| \ell_i^2 - \left( x_i - c_i \right)^2 \right|
    &\ge&
      \left| \ell_i^2 - \left( c_i+\ell_i-\gamma_i \right)^2 \right| \\
    &=&
      \left| 2 \ell_i \gamma_i - \gamma_i^2\right| \\
    &=&
      \left| \left( 2 \ell_i - \gamma_i \right) \gamma_i \right|.
  \end{eqnarray*}
  Next we define the functions $g_i$ as
  \begin{eqnarray*}
    g_i (\gamma_1,\ldots,\gamma_i,x_{i+1},\ldots,x_{k})
    &:=&
      \min \;
      \Bigl(\!\!
      \begin{array}[t]{l}
        \bigl\{
	  (2\ell_j-\gamma_j)\gamma_j : 1\le j\le i
	\bigr\}
      \\[1ex]
	\cup
	\;
        \bigl\{
	  \ell_j^2 - \left( x_j-c_j \right)^2 : i < j \le k
	\bigr\}
      \Bigr)
      \end{array}
  \end{eqnarray*}
  and in the end, the sequence of replacements leads to
  \begin{eqnarray*}
    \varphi_\inbox (\gamma)
    &:=&
      g_k(\gamma) \\
    &=&
      \min_{1\le j\le k}
      \left( 2 \ell_j-\gamma_j \right) \gamma_j.
  \end{eqnarray*}
  Part~2~($\chi_{\inbox}$).
  Now we determine a bound on the volume of
  the complement of the region of uncertainty.
  For every $i\in I$,
  a valid bounding function is given by
  \begin{eqnarray*}
    \hat\chi_{g^{*i}_{i-1}}\left(\gamma_{i}\right)
      &=& 2 \delta_i - 4 \gamma_{i}.
  \end{eqnarray*}
  This results in the following bound on the total volume:
  \begin{eqnarray*}
    \chi_{\inbox}\left(\gamma\right)
      &=& \prod_{i=1}^k \,
          { \hat\chi_{g_{i-1}^{*i}} \left(\gamma_{i}\right) } \\
      &=& \prod_{i=1}^{k} \left( 2 \delta_i - 4 \gamma_{i} \right).
  \end{eqnarray*}
  Now that we have determined the bounding functions
  $\varphi_{\inbox}$ and $\chi_{\inbox}$,
  it would be possible to finish the analysis
  with the method of quantified relations---but
  this is not our interest in this section.
  \hfill$\bigcirc$
\end{example}
\begin{example}
  This is the continuation of Examples~\ref{ex-inbox-gi-dom-crit}
  and~\ref{ex-inbox-topdown}.
  Here we want to investigate the regions of uncertainty
  for the various functions $g_i$.
  More precisely,
  we are interested in the correlation between the regions
  which are defined bottom-up in the second phase of the approach.

  Figure~\ref{fig-crit-set-inbox-2}
  visualizes the regions of uncertainty for the functions $g_i$.
  The regions of uncertainty are light shaded
  whereas their complements are dark shaded.
  The decomposition is initiated by the choice of $\gamma\in\GAB$.
  Since each component $\gamma_i$ is positive,
  neighborhoods of the critical set are added to the region of uncertainty
  on the way back up to $g_0$.
  \begin{figure}[t]\centering%
    \includegraphics[width=.95\columnwidth]{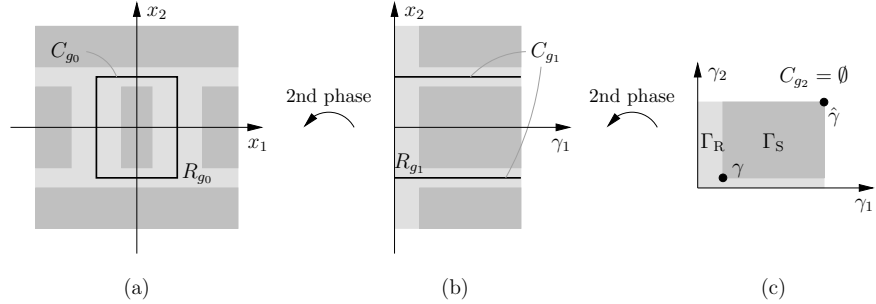}%
    \caption{Illustration of the regions of uncertainty
      for the various domains in the analysis of
      the 2-dimensional $\inbox$-predicate.}%
    \label{fig-crit-set-inbox-2}%
  \end{figure}%

  As we have seen in Example~\ref{ex-inbox-gi-dom-crit},
  the upper line segment of $C_{g_0}$
  causes the upper line of $C_{g_1}$.
  Conversely, we can now see that the upper line of $C_{g_1}$
  causes a region of uncertainty around the
  \emph{line which passes through} the upper line segment of $C_{g_0}$.
  Be aware
  that our top-down approach is designed such that this behavior
  is forced for all non-region-regular situations.
  This implies that our method does not need any kind of exceptional sets.
  In the contrary,
  there are no restrictions on the measure of the critical sets at all:
  The only thing that matters is the criterion
  if $f$ is region-suitable or not.
  \hfill$\bigcirc$
\end{example}%

\subsection{Further Remarks}
\label{sec-td-deeper-insight}

\newcommand{\CEXT}{{C^\text{\rm ext}}}
  A different concept of
  the top-down approach is published in \cite{MOS11}.
  Although both presentations rest upon the same motivation,
  there are some technical differences in the realization.
  To avoid misunderstandings in the presentation
  and gain a deeper insight into our approach,
  we end this section with selected questions.

  \emph{Does $f$ have to be (upper- or lower-) continuous
  to be top-down analyzable?}
  No, we do not assume any kind of continuity in our approach.
  Points of discontinuity may be critical,
  but they do not have to be critical.

  \emph{May we assume that $f$ is continuous?}
  No, the top-down approach is defined recursively and
  the auxiliary functions $g_i$ are not continuous in general.
  Consider for example the continuous polynomial
  $f(x_1,x_2) := x_1^2 + x_2^2 - 1$
  which is the planar ``in unit circle'' predicate.
  Then $g_1(x_1,\gamma_2)$ is not continuous in four points
  for fixed $\gamma_2$.
  The function is illustrated in Figure~\ref{fig-exa-in-circle}.
  That is the reason
  why the top-down approach \emph{must} work for discontinuous functions.
  \begin{figure}[t]\centering
    \includegraphics[width=.33\columnwidth]{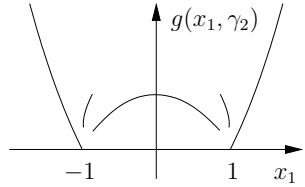}
    \caption{Exemplified drawing of the ``in unit circle'' predicate
      after the first replacement.
      The function values on the interval $[-1,1]$
      vary with $\gamma_2$.}
    \label{fig-exa-in-circle}
  \end{figure}

  \emph{Does a critical set of measure zero imply that $f$ is region-suitable?}
  No, not in general.
  A notorious example is the density of $\Q$ in $\R$.
  Let $A\subset\R$ be an interval.
  Although $A\cap\Q$ is a set of measure zero,
  there is no $\varepsilon>0$ such that the neighborhood
  $U_\varepsilon(A\cap\Q)$ has a volume smaller than $\mu(A)$.
  But the latter is a necessary criterion for region-suitability
  and the applicability of controlled perturbation.

  \emph{Does region-suitability imply a finite critical set?}
  No.
  A counter-example is the function $x\cdot\sin\left(\frac{1}{x}\right)$
  which is region-suitable
  although it has infinitely many zeros in any finite neighborhood of zero.
  (By the way, this function is also value-suitable.)
  We summarize:
  \emph{Critical sets of region-suitable functions are countable,
  but not every countable critical set implies region-suitability.}

  \emph{Is it possible to neglect isolated points in the analysis?}
  We may never exclude critical points from the analysis;
  they are always used to define the region of uncertainty.
  We may exclude less-critical points
  provided that we adjust the success-probability ``by hand''.
  We may neglect non-critical points
  provided that we still determine
  the correct inf-value-suitable bound $\varphi_{\inf f}$.
  (See also Section~\ref{sec-succ-prob}.)

  \emph{May we add additional points to the critical set?}
  Yes, we may add points to the critical set
  provided that $f$ is still guaranteed to be region-suitable.
  (See also Section~\ref{sec-succ-prob}.)

  \emph{Can we decide if $f$ is top-down analyzable
  without developing the sequence of replacements?}
  It is a necessary condition for the top-down analyzability of $f$
  that $g_k$ is positive everywhere.
  It is not clear
  how we can guarantee this property in general without deriving $g_k$.

\section[Determining the Lower Fp-safety Bound (1st Stage, s-suit)]{Determining the Lower Fp-safety Bound}
\label{sec-guards-safetybounds}

  Here we introduce the design of
  guards and fp-safety bounds.
  Guards are necessary to implement guarded evaluations in $\AG$.
  In Section~\ref{sec-fea-guards}
  we explain how guards can be implemented for a wide class of functions
  including polynomials.
  To analyze the behavior of guards,
  we introduce fp-safety bounds
  in Section~\ref{sec-fea-fpsafety}.
  We explain how we determine the fp-safety bound
  in the analysis
  (see Figure~\ref{fig-ana-fea}).
  Furthermore,
  we prove the fp-safety bounds
  which we have used in previous sections.
  \begin{figure}[h]\centering
    \includegraphics[width=.45\columnwidth]{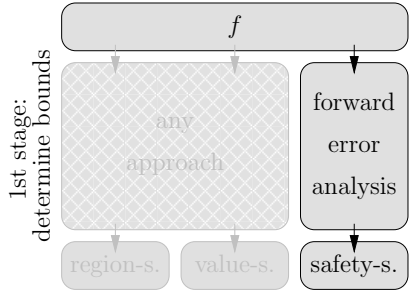}
    \caption{An error analysis is used
      to derive the bounding function for the safety-suitability
      in the first stage of the analysis.}
    \label{fig-ana-fea}
  \end{figure}

\subsection{Implementing Guarded Evaluations}
\label{sec-fea-guards}

  Our presentation of guarded evaluations
  is based on rounding error analyses
  following the approach in~\cite{F97,BFS01,MOS11}.
  A refinement is presented in the appendix of~\cite{MOS11}.

\subsection*{Rounding Error Analysis}

  The implementation of guards is based on maximum error bounds.
  To determine the error bounds we use rounding error analyses.
  Note that the error bound of a function $f$
  depends on the formula $E$ that realizes $f$
  and, especially, on the chosen \emph{sequence of evaluation}.
  In Table~\ref{def-fperrorbound-inline}
  we cite some rules to determine error bounds.
  Expressions $E$ that are composed of addition,
  subtraction, multiplication and absolute value can be bounded
  by the value $B_E$ in the last row of the table.
  This includes the evaluation of polynomials;
  for further operators see~\cite{F97,BFS01,MOS11}.
  The quantities $\IND_E$ and $\SUP_E$
  are derived according to the sequence of evaluation of $E$.
  The value $\IND_x$ is 0 if $x\in\FL$, and it is 1 if $x$ is rounded.
\begin{example}
  We determine the error bound for the expression
  \begin{eqnarray*}
    E(x_1,\ldots x_k)\RF
      &=& (((a\cdot x_1)\cdot x_2)\cdots x_k)\RF
  \end{eqnarray*}
  where $k\in\N$, $a\in\R$ is a coefficient and
  $x\in U_\delta(\x)\RG\subseteq[-2^\E,2^\E]^k$.
  A worst-case consideration leads to
  $\IND_{a}=1$ and $\SUP_{a}=\CARDI{a\RF}$ for the coefficient and
  $\IND_{x_i}=0$ and $\SUP_{x_i}=\CARDI{x_i\RF}$ for $1\le i\le k$.
  Then we obtain
  $\IND_{ax_1}=2$ and $\SUP_{ax_1}=\CARDI{ax_1\RF}$
  after the first multiplication.
  Taking all multiplications into account,
  we get
  $\IND_{E}=k+1$ and $\SUP_{E}=\CARDI{ax_1\cdots x_k\RF}$.
  According to Table~\ref{def-fperrorbound-inline}
  we obtain the \emph{dynamic error bound}
  \begin{eqnarray*}
    B_E(L,x) &=&
      (k+1) \cdot \CARDI{ax_1\cdots x_k\RF} \cdot 2^{-L}
  \end{eqnarray*}
  and the \emph{static error bound}
  \begin{eqnarray*}
    B_E(L) &=&
      (k+1) \cdot \CARDI{a\RF} \cdot 2^{k\E-L}
  \end{eqnarray*}
  where $2^\E$ is an upper bound on the absolute value of a perturbed input.
  \hfill$\bigcirc$
\end{example}
\begin{remark}\label{rem-errorbound-zero}
  We make the important observation that
  the bound $B_E(L)$ approaches zero when $L$ approaches infinity,
  that means,
  \begin{eqnarray*}
   \lim_{L\to\infty} B_E(L) &=& 0.
  \end{eqnarray*}
  Furthermore we observe
  that \emph{all error bounds which are derived from
  Table~\ref{def-fperrorbound-inline}
  have this property.}
  \hfill$\bigcirc$
\end{remark}
\begin{table}[t]
  \begin{eqnarray*}
    \begin{array}{|@{\qquad}c@{\qquad}|@{\qquad}c@{\qquad}|@{\qquad}c@{\qquad}|}\hline
    \tabrule E & \SUP_E & \IND_E
      \\\hline\hline
    \tabrule x & \CARDI{x\,\RF} & \text{0 or 1}
      \\\hline
    \tabrule E_1\pm E_2
      & (\SUP_{E_1}+\SUP_{E_2})\RF
      & 1+\max\left\{\IND_{E_1},\IND_{E_2}\right\}
      \\\hline
    \tabrule E_1 \cdot E_2
      & (\SUP_{E_1} \cdot \SUP_{E_2})\RF
      & 1 + \IND_{E_1} + \IND_{E_2}
      \\\hline
    \tabrule |E|
      & \SUP_{E}
      & \IND_E
      \\\hline\hline
    \multicolumn{3}{|c|}{\tabrule B_E := \IND_E \cdot \SUP_E \cdot 2^{-L}}
      \\\hline
    \end{array}
  \end{eqnarray*}
    \caption{This table reprints parts of Table~2.1
      in Funke~\cite[p.~11]{F97}.
      The row for $|E|$ is added by us.}
    \label{def-fperrorbound-inline}
\end{table}

\subsection*{Guarded Evaluation}

  In guarded algorithms $\AG$
  every predicate evaluation $f(x)\RF$
  must be protected by a guard $\GG_f(x)$ that verifies
  the sign of the result.
  Guards can be implemented
  using the dynamic
  (or the weaker static) error bounds.
  Let $B_f(L,x)$ be an upper bound on the rounding error of $f(x)\RF$
  for floating point arithmetic $\FL$,
  that means,
  \begin{eqnarray}\label{for-error-bound}
    B_f(L,x) &\ge& \CARDI{f(x)\RFL - f(x)}.
  \end{eqnarray}
  Then we can immediately derive the implication
  \begin{eqnarray}\label{for-guard-dynamic}
    \CARDI{f(x)\RF} > B_f(L,x)
    \FORMSEP &\Rightarrow& \FORMSEP
    \sign(f(x)\RFL) = \sign(f(x)).
  \end{eqnarray}
  We use the unequation on the left hand side to construct a
  \emph{guard $\GG_f$ for $f$} where
  \begin{eqnarray*}
    \GG_f(x)
      &:=&
        \big( \; \CARDI{f(x)\RF} > B_f(L,x) \, \big).
  \end{eqnarray*}
  If $\GG_f(x)$ is true, $f(x)$ has the correct sign.
  Note that this definition is in accordance with Definition~\ref{def-guard}
  on Page~\pageref{def-guard}.

\subsection{Analyzing Guards With Fp-safety Bounds}
\label{sec-fea-fpsafety}

  Now we explain how to analyze the behavior of guards
  according to~\cite{F97,BFS01,MOS11}.
  Remember that we perform the analysis in real space.
  The implication
  \begin{eqnarray}\label{for-fpsafety-static}
    \CARDI{f(x)} > 2B_f(L, x)
    \FORMSEP &\Rightarrow& \FORMSEP
    \CARDI{f(x)\RF} > B_f(L, x).
  \end{eqnarray}
  is true because of Formula~(\ref{for-error-bound}).
  The inequality on the left hand side
  is a relation that we can safely verify in real space.
  We can always use the static error bound
  to construct a \emph{fp-safety bound $\Sinf$ for $f$}
  \begin{eqnarray*}
    \Sinf(L) &:=& 2B_f(L)
  \end{eqnarray*}
  where $B_f(L)$ is the static error bound.
  Note that this definition is in accordance with
  Definition~\ref{def-fpsafetybound}
  on Page~\pageref{def-fpsafetybound}
  because
  the implications
  in Formulas~(\ref{for-guard-dynamic})
  and~(\ref{for-fpsafety-static})
  guarantee the desired implication in Formula~(\ref{for-fpsafety-final}).
  \emph{Because of Remark~\ref{rem-errorbound-zero},
  the fp-safety bound $\Sinf(L)$
  fulfills the safety-condition on page~\pageref{def-safety-cond}
  by construction.}
  Next we derive a \mbox{fp-safety} bound for univariate polynomials.
\begin{corollary}\label{col-unipolysafety}
  Let $f$ be a univariate polynomial
  \begin{eqnarray}\label{for-2-unipolyrep}
    f(x) &=& a_d \cdot x^d + a_{d-1} \cdot x^{d-1} +
               \ldots + a_1 \cdot x + a_0
  \end{eqnarray}
  of degree $d$.
  Then
  \begin{eqnarray}\label{for-sfl-univariate}
    \Sinf(L) &:=& (d+2) \cdot \max_{1\le i\le d} |a_i| \cdot 2^{\E(d+1)+1-L}
  \end{eqnarray}
  is a fp-safety bound for $f$ on $[-2^\E,2^\E]$ where $\E\in\N$.
\end{corollary}
\begin{proof}
  We apply the error analysis of this section.
  We evaluate Formula~(\ref{for-2-unipolyrep}) from the right to the left.
  For a static error bound we get
  \begin{eqnarray*}
    B_f(L)
      &:=& \IND_f \cdot \SUP_f \cdot 2^{-L}
      \;\;=\;\;
      (d+2) \cdot \left(\max_{1\le i\le d} |a_i| \cdot 2^{\E(d+1)}\right) \cdot 2^{-L}.
  \end{eqnarray*}
  Finally we set the fp-safety bound to $\Sinf(L) := 2B_f(L)$.
  \qed
\end{proof}
  Multiplications usually cause larger rounding errors
  than additions.
  Surprisingly,
  the evaluation of univariate polynomials
  with the Horner scheme\footnote{%
    For Horner scheme see Hotz~\cite{H90}.}
  (which minimize the number of multiplications)
  does not lead to
  a smaller error bound than the one we have derived in the proof.
  Next we derive an error bound for $k$-variate polynomials.
  We define $x^\iota:=x_1^{\iota_1} \cdot \ldots \cdot x_k^{\iota_k}$
  for $\iota\in\N_0^k$ and $x\in\R^k$.
\begin{corollary}\label{cor-multipolysafety}
  Let $f$ be the $k$-variate polynomial ($k\ge 2$)
  \begin{eqnarray*}
    f(x) &:=& \sum_{\iota\in\Ind} a_\iota x^\iota
  \end{eqnarray*}
  where $\Ind\subset\N_0^k$ is finite
  and $a_\iota\in\R_{\neq 0}$ for all $\iota\in\Ind$.
  Let $d$ be the total degree of $f$ and
  let $\NT$ be the number of terms in $f$.
  Then
  \begin{eqnarray}\label{for-sfl-multivariate}
    \Sinf(L) &:=& (d+1+\lceil\log \NT\rceil) \cdot
              \NT \cdot \max_{\iota\in\Ind} |a_\iota| \cdot
	      2^{\E d+1-L}
  \end{eqnarray}
  is a fp-safety bound for $f$ on $[-2^\E,2^\E]^k$ where $\E\in\N$.
\end{corollary}
\begin{proof}
  We begin with the determination of the error bound $B_f$.
  The maximum absolute value of the term $a_\iota x^\iota$
  is obviously upper-bounded by
  the product of a bound on $a_\iota$ and a bound on $x^\iota$.
  Because $|x_i|\le 2^\E$ for all $1\le i \le k$
  we have
  \begin{eqnarray*}
    \SUP_{a_\iota x^\iota}
      &\le& \max_{\iota\in\Ind} |a_\iota| \cdot 2^{\E d}.
  \end{eqnarray*}
  Since we know the number $\NT$ of terms in $f$\!,
  we can then upper-bound $\SUP_f$ by
  \begin{eqnarray*}
    \SUP_f
      &\le& \NT \cdot \max_{\iota\in\Ind} |a_\iota| \cdot 2^{\E d}.
  \end{eqnarray*}
  In addition, we have $\IND_{a_\iota x^\iota}=d+1$
  since we evaluate $d$ multiplications and
  only $a_\iota$ may not be in the set $\F$.
  (Remember that, because of the perturbation,
  the values $x_i$ belong to the grid $\G$
  which is a subset of $\F$.)

  To keep $\IND_f$ as small as possible,
  we sum up the $\NT$ terms pairwise such that the tree of evaluation
  has depth $\lceil\log \NT\rceil$.
  This leads to $\IND_f=d+1+\lceil\log \NT\rceil$.
  Therefore we conclude that
  \begin{eqnarray*}
    B_f(L) &=& \left(d+1+\lceil\log \NT\rceil\right) \cdot
              \left(\NT \cdot \max_{\iota\in\Ind} |a_\iota| \cdot 2^{\E d}\right)
	      \cdot 2^{-L}.
  \end{eqnarray*}
  As usual we set $\Sinf(L):=2 B_f(L)$.
  \qed
\end{proof}

\section{The Treatment of Range Errors (All Components)}
\label{sec-rational-function}

  In this section we address a floating-point issue that
  is caused by poles of rational functions.
  So far
  the implementation and analysis of functions
  is based on the fact that
  signs of floating-point evaluations
  are only non-reliable on certain environments of zero.
  Now we argue that signs of evaluations may
  also be non-reliable on environments of poles.
  We do this for the purpose to embed rational functions
  into our theory.
  In Section~\ref{sec-range-error-implement}
  we extend the previous implementation considerations such that
  they can deal with range errors.
  In Section~\ref{sec-range-error-analysis}
  we expand the analysis
  to range errors of the floating-point arithmetic $\F$.
  \emph{This is the first presentation that gains generality by
  the practical and theoretical treatment of range errors
  which, for example, are caused by poles of rational functions.}

\subsection{Extending the Implementation}
\label{sec-range-error-implement}

  We examine the simple rational function $f(x)=\frac{1}{x}$.
  It is well-known that the function value of $f$
  at the pole $x=0$ does not exist in $\R$
  (unless we introduce the unsigned symbolic value $\pm\infty$,
  see Forster~\cite{F06}).
  We make the important observation that
  we cannot determine the function value of $f$
  {in a neighborhood of a pole}
  with floating-point arithmetic $\FLK$
  because the absolute value of $f$ may be \emph{too large.}
  Moreover,
  we observe that the \emph{sign of $f$ may change}
  on a neighborhood of a pole.
  Both observations suggest that
  \emph{poles play a similar role like zeros
  in the context of controlled perturbation.}
  Now we extend the implementation
  such that it gets able to deal with range errors.

  We extend the implementation of guarded evaluations
  in the following way:
  If the absolute value of $f$ cannot be represented
  with the floating-point arithmetic $\FLK$ because it is too large,
  we abort $\AG$ with the notification of a \emph{range error}.
  We do not care about the source of the range error:
  It may be ``division by zero'' or ``overflow.''
  The implementation of the second guard per evaluation is straight forward.
  Some programming languages provide an exception handling that
  can be used for this objective.

  In addition we must change the implementation
  of the controlled perturbation algorithm $\ACP$.
  If $\AG$ fails because of a \emph{range error,}
  we increase the bit length $K$ of the exponent
  (instead of the precision $L$).
  Be aware that we talk about the exponent,
  that means,
  an additive augmentation of the bit length
  implies a multiplicative augmentation of the range.
  These simple changes guarantee that the floating-point arithmetic $\FLK$
  gets adjusted to the necessary dimensions in neighborhoods of poles
  or in regions where the function value is extremely large.

\subsection{Extending the Analysis of Functions}
\label{sec-range-error-analysis}

  For the purpose of dealing with range errors in the analysis,
  we need to adapt several parts of the analysis tool box.
  Below we present the necessary changes and extensions in the same order
  in which we have developed the theory.

\subsection*{Criticality and the region-suitability}

  The changes to deal with range errors
  affect the interface
  between the two stages of the analysis of functions.
  At first we extent the definition of criticality.
  We demand that certain points
  (e.g.~poles of rational functions)
  are {critical}, too, and
  refine Definition~\ref{def-critical-set} in the following way.
\begin{definition}[critical]\label{def-critical-set-second}
  Let $\PRED$ be a predicate description.
  We call a point $c\in\U_\delta(\x)$
  \emph{critical} if
  \begin{eqnarray*}
    \inf_{x\in U_\varepsilon(c)\setminus\{c\}} \; \left|f(x)\right|
    = 0
    \FORMSEP & \text{\rm or} & \FORMSEP
    \sup_{x\in U_\varepsilon(c)\setminus\{c\}} \; \left|f(x)\right|
    = \infty
  \end{eqnarray*}
  on a neighborhood $U_\varepsilon(c)$
  for infinitesimal
  small $\varepsilon>0$.
  Furthermore,
  we call $c$ \emph{less-critical}
  if $c$ is not critical,
  but $f(c)=0$ or $c$ is a pole.
  Points that are neither critical nor less-critical
  are called \emph{non-critical}.
\end{definition}
  For simplicity and as before,
  we define the
  \emph{critical set $C_{f,\delta}\label{def-crit-set-final-inline}$}
  to be the union of critical and less-critical points within $\U_\delta(\x)$.
  Be aware that the new definition of criticality
  may expand the region of uncertainty.
  As a consequence it affects the \emph{region-suitability}
  and the bound $\nu_f$, respectively $\chi_f$.
  Note that
  Definition~\ref{def-critical-set-second}
  guarantees that we exclude neighborhoods of poles from now on.
  Because we have integrated poles into the definition of criticality,
  we have implicitly adapted the region-suitability.

\subsection*{The sup-value-suitability}

  So far we have only considered $\inf|f|$
  outside of the region of uncertainty.
  But to get a quantified description of range issues in the analysis,
  we need to consider $\sup|f|$ as well.
  What we have called value-suitability so far
  is now called, more precisely, \emph{inf-value-suitability}.
  Its bounding function, that we have called $\varphi_f(\gamma)$ so far,
  is now called
  $\varphi_{\inf f}(\gamma)\label{def-phi-inf-final-inline}$.

  In addition to Definition~\ref{def-value-suit}
  we introduce \emph{sup-value-suitability},
  that means,
  there is an upper-bounding function
  $\varphisup(\gamma)\label{def-phi-sup-final-inline}$
  on the absolute value of $f$
  outside of the region of uncertainty $R_f$.
  We show how the new bound is determined with the bottom-up approach
  later on.
  Based on the new terminology,
  we call $f$ \emph{(totally) value-suitable}
  if $f$ is both: inf-value-suitable and sup-value-suitable.

\subsection*{The sup-safety-suitability and analyzability}

  We also extend Definition~\ref{def-safety-suit}.
  What we have called safety-suitability so far
  is now called, more precisely, \emph{inf-safety-suitability}.
  Its bounding function
  $\Sinf(L)\label{for-Sinf-final-inline}$
  is now called the \emph{lower fp-safety bound}.

  In addition
  we introduce \emph{sup-safety-suitability},
  that means,
  there is an invertible upper-bounding function $\Ssup(K)$
  on the absolute value of $f$
  with the following meaning:
  If we know that
  \begin{eqnarray*}
    |f(x)| &\le& \Ssup(K),
  \end{eqnarray*}
  then $f(x)\RF$ is definitely a finite number in $\FLK$.
  We call $\Ssup(K)$ the \emph{upper fp-safety bound}.
  Such a bound is trivially given by\footnote{%
    Firstly,
    the largest floating-point number that is representable with $\FLK$ is
    $(2-2^{-L}) 2^{2^{K-1}}$.
    Secondly,
    we must take the maximal floating-point rounding error into account.}
  \begin{eqnarray*}
    \Ssup(K)\label{for-Ssup-final-inline}
      &:=&
        2^{2^{K-1}} - \, \Sinf(L).
  \end{eqnarray*}
  Based on the new terminology,
  we call $f$ \emph{(totally) safety-suitable}
  if $f$ is both: inf-safety-suitable and sup-safety-suitable.
  As a consequence,
  we call $f$ \emph{analyzable}
  if $f$ is region-suitable, value-suitable (both subtypes)
  and safety-suitable (both subtypes).

\subsection*{The method of quantified relations}

  Next we extent the method of quantified relations
  such that the new bounds on the range of floating-point arithmetic
  are included into the analysis.
  In addition
  to the precision function $L_f(p)$,
  we determine the bounding function
  \begin{eqnarray*}\label{for-Kf-inline}
    K_f(p)
    &:=&
      \left\lceil
        \Ssup^{-1} \left(
	  \varphisup \left(
	    t\cdot\nu_f^{-1}\left({\varepsilon_{\nu}\left(p\right)}\right)
	\right)\right)
      \right\rceil.
  \end{eqnarray*}
  That means,
  we deduce the maximum absolute value of $f$
  outside of the region of uncertainty from the probability;
  afterwards
  we use the upper fp-safety bound
  to deduce the necessary bit length of the exponent.
  The derivation of $K_f(p)$ is absolutely analog to the derivation of $\LS(p)$
  in Steps~1--5 of the method of quantified relations.

  We summarize our results so far:
  If we have the bounding functions of the interface of the function analysis,
  we know that the floating-point arithmetic $\F_{L_f(p),K_f(p)}$
  is sufficient to safely evaluate $f$
  at a random grid point in the perturbation area
  with probability $p$.

  Furthermore,
  we can derive a probability function $p_f$
  if $f$ is analyzable and $\varphi_{\inf f}$ and $\varphisup$
  are both invertible.
  Analog to the definition of $\pinf(L)$
  in Remark~\ref{rem-meth-quan-rela}.4,
  we derive the additional bound on the probability
  \begin{eqnarray*}\label{for-supsafe-inline}
    \psup(K)
    &:=&
      \varepsilon_\nu^{-1}
        \left( \nu_f
          \left( \frac{1}{t} \cdot \varphisup^{-1}
            \left( \Ssup(K)
      \right) \right) \right)
  \end{eqnarray*}
  from $K_f(p)$.
  This leads to the final \emph{probability function\label{def-2-prob-func-inline}}
  $p_f:\N\times\N\to(0,1)$ where
  \begin{eqnarray*}\label{for-pfLK-inline}
    p_f(L,K) &:=&
      \min \left\{
        \pinf(L),
        \, \psup(K),
	\, \pgrid(L)
      \right\}
  \end{eqnarray*}
  for parameter $t\in(0,1)$.

\subsection*{The bottom-up approach}

  Now we extend the calculation rules of the bottom-up approach
  to also derive the bounding function $\varphisup(\gamma)$
  from simpler sup-value-suitable functions.
  At first we replace the lower-bounding rule in Theorem~\ref{theo-lower-bound}
  by the following sandwich-rule.
\begin{theorem}[sandwich]\label{theo-sandwich}
  Let $\PREDL$ be a predicate description.
  If there is a region-value-suitable function
  $g:\U_\delta(A)\to\R$ and $c_1,c_2\in\R_{>0}$ where
  \begin{eqnarray*}
    c_1 \, |g(x)|
      & \le &
    |f(x)|
      \;\; \le \;\;
    c_2 \, |g(x)|,
  \end{eqnarray*}
  then $f$ is also region-value-suitable
  with the following bounding functions:
  \begin{eqnarray*}
    \nu_f(\gamma) &:=& \nu_g(\gamma) \\
    \varphi_{\inf f}(\gamma) &:=& c_1 \varphi_{\inf g}(\gamma) \\
    \varphi_{\sup f}(\gamma) &:=& c_2 \varphi_{\sup g}(\gamma).
  \end{eqnarray*}
  If $f$ is in addition safety-suitable, $f$ is analyzable.
\end{theorem}
\begin{proof}
  The region-suitability and inf-value-suitability
  follows from the proof of Theorem~\ref{theo-lower-bound}.
  The sup-value-suitability is proven similar to Part~2
  of the mentioned proof.
  \qed
\end{proof}
  Next we extent the product rule in Theorem~\ref{theo-product}.
  We just add the assignment
  \begin{eqnarray*}
    \varphi_{\sup f}(\gamma) &:=&
      \varphi_{\sup g}(\gamma_1,\ldots,\gamma_\ell) \cdot
      \varphi_{\sup h}(\gamma_{j+1},\ldots,\gamma_k).
  \end{eqnarray*}
  after Formula~(\ref{for-phi-prod-rule}).
  Its proof follows Part~1 of the proof of Theorem~\ref{theo-product}.

  At last we extent the min-rule and the max-rule
  in Theorem~\ref{theo-ruleminmax}.
  We add the two assignments
  \begin{eqnarray*}
    \varphi_{\sup \fmin}(\gamma) &:=&
      \min \{ \varphi_{\sup g}(\gamma_1,\ldots,\gamma_\ell),
        \varphi_{\sup h}(\gamma_{j+1},\ldots,\gamma_k)\} \\
    \varphi_{\sup \fmax}(\gamma) &:=&
      \max \{ \varphi_{\sup g}(\gamma_1,\ldots,\gamma_\ell),
        \varphi_{\sup h}(\gamma_{j+1},\ldots,\gamma_k)\}.
  \end{eqnarray*}
  after Formula~(\ref{for-phi-max-rule}).
  Again, its proof follows Part~1 of the proof of Theorem~\ref{theo-product}.

\subsection*{The top-down approach}

  Similar to the functions $\varphi_{\inf g_i}$,
  which are simply called $\varphi_{g_i}$
  in the overview in Figure~\ref{fig-analysis-3},
  we determine
  the functions $\varphi_{\sup g_i}$ in the second phase
  of the pseudo-top-down approach
  in a bottom-up fashion.

  This completes the integration of the range considerations
  into the analysis tool box.
  Be aware that all changes presented in this section
  do not restrict the applicability of the analysis tool box
  in any way.
  On the contrary,
  {they are necessary for the correctness and generality
  of the tool box.}

\section{The Analysis of Rational Functions}
\label{sec-ana-rational-func}

  We have just solved the arithmetical issues
  that occur in the implementation and
  analysis of rational functions.
  Besides we must solve technical issues in the implementation of guards
  and, moreover, provide a general technique to derive a quantitative
  analysis for rational functions.
  \emph{This is the first presentation that contains
  the implementation and analysis of rational functions.}

  Let $f:=\frac{g}{h}$
  be a rational function,
  that means, let $g$ and $h$ be multivariate polynomials.
  Let $k$ be the number of arguments of $f$,
  i.e., we consider $f(x)$ where $x=(x_1,\ldots,x_k)$.
  The arguments of $g$ and $h$ may be any subsequence of $x$,
  but each $x_i$ is at least an argument of $g$ or an argument of $h$.
  We know that $g$ and $h$ are analyzable
  (see Section~\ref{sec-multivariate-poly}).

  At first we discuss the implementation of guards for rational functions.
  We make the important observation that---independent
  of the evaluation sequences of $g$ and $h$---the
  \emph{division} of the value of $g$
  by the value of $h$
  is the \emph{very last operation} in the evaluation of $f$.
  Because of the standardization of floating-point arithmetic
  (e.g., see \cite{IEEE08}),
  the sign of $f$ is computed correctly
  if the signs of $g$ and $h$ are computed correctly.
  Therefore it is sufficient for an implementation
  of a predicate that branches on the sign of a rational function $f$
  to use the guard
  $\GG_f := \left(\GG_g \wedge \GG_h\right)$.

  But how do we analyze this predicate,
  that means, how can we relate the known quantities?
  Let $x$ be given.
  In the case that
  the (dependent) arguments of $g$ and $h$
  lie outside of their region of uncertainty,
  we can deduce the relation
  \begin{eqnarray}\label{for-rat-func-limits}
    \frac{S_{\inf g}}{S_{\sup h}}
    \;\; \le \;\;
    f(x)
    \;\; \le \;\;
    \frac{S_{\sup g}}{S_{\inf h}}.
  \end{eqnarray}
  Unfortunately this is not what we need.
  This way,
  we can only deduce the value of $f$ from the values of $g$ and $h$,
  but not vice versa:
  If $f(x)$ fulfills Formula~(\ref{for-rat-func-limits}),
  we cannot deduce that the guards $\GG_g$ and $\GG_h$ are true.
  For example, assume that $f(x)=1$;
  then we know that the values of $g$ and $h$ are equal,
  but we do not know if their values are fp-safe or close to zero.

  Therefore we choose a different way
  to analyze the behavior of guard $\GG_f$.
  Since $g$ and $h$ are multivariate polynomials,
  we can analyze the behavior of $\GG_g$ and $\GG_h$
  and derive the precision functions $L_g(p)$ and $L_h(p)$
  as we have seen in earlier sections.
  If we demand that $g$ and $h$ evaluate successfully with probability
  $\frac{1+p}{2}$ each,
  $f$ evaluates successfully with probability $p$
  since the sum of the failure probability of $g$ and $h$
  is at most $(1-p)$.
  This leads to the precision function
  \begin{eqnarray*}
    L_f(p)
      &:=&
        \max \left\{
	  L_g\left( \frac{1+p}{2}\right),
	  L_h\left( \frac{1+p}{2}\right)
	\right\}
  \end{eqnarray*}
  which reflects the behavior of $\GG_f$
  and therefore analyzes the behavior of an implementation of 
  the rational function evaluation of $f$.

\section{General Analysis of Algorithms (Composition)}
\label{sec-ana-algo}

  So far we have only presented components of the tool box
  which are used
  to analyze functions.
  Now we introduce the components
  which are used
  to analyze controlled-perturbation algorithms $\ACP$.
  Figure~\ref{fig-illu-ana-algo} illustrates the analysis of algorithms.
  Similar to the analysis of functions,
  the algorithm analysis has two stages.
  The \emph{interface} between the stages is introduced
  in Section~\ref{sec-nec-con-algo}.
  It consists of necessary algorithm properties
  (to the left of the dashed line)
  and the analyzability of the used predicates
  (to the right of the dashed line).
  There we also show
  how to determine the bounds associated with the algorithm properties.
  In Section~\ref{sec-algo-prop}
  we give an overview of algorithm properties.
  The \emph{method of distributed probability}
  represents the actual analysis of algorithms
  and is presented in Section~\ref{sec-meth-distri-prob}.
  \begin{figure}[h]\centering
    \includegraphics[width=.95\columnwidth]{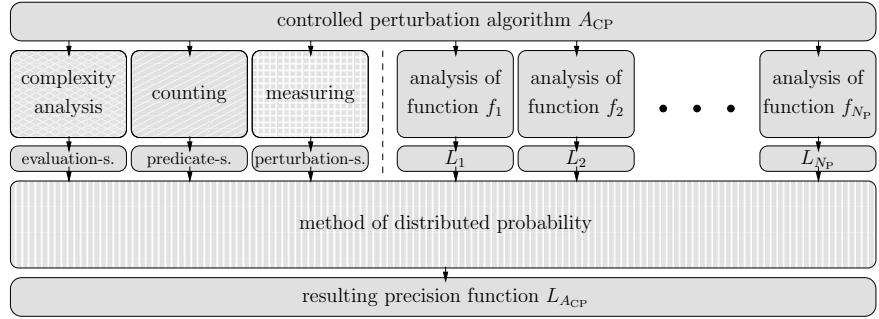}
    \caption{Illustration of the analysis
      of controlled-perturbation algorithms.}
    \label{fig-illu-ana-algo}
  \end{figure}

\subsection{Necessary Conditions for the Analysis of Algorithms}
\label{sec-nec-con-algo}

  Next we introduce several properties of controlled-perturbation algorithms.
  Sometimes we use the same names for algorithm and function properties
  to emphasize the analog.
  We describe to which algorithms we can \emph{apply} controlled perturbation,
  for which we can \emph{verify} that they terminate,
  and which we can \emph{analyze} in a quantitative way
  because they are \emph{suitable} for the analysis.
  In particular,
  three properties are necessary for the analyzability of algorithms:
  evaluation-, predicate- and perturbation-suitability.
  However,
  the three conditions are not sufficient for the analysis of algorithms
  since there are also prerequisites on the used predicates.
  In this section
  we define the various properties of controlled-perturbation algorithms,
  explain how we obtain the bounding functions
  that are associated with the necessary conditions,
  and show how the algorithm properties are related with each other.
\begin{definition}%
  \label{def-algo-prop}
  Let $\ACP$ be a controlled perturbation algorithm.
  \begin{itemize}
  \item\label{def-lacp}
    (applicable).
    We call $\ACP$ \emph{applicable}
    if there is a precision function
    $L_{\ACP}:(0,1)\times\N\to\N$
    and $\eta\in\N$
    with the property:
    At least one from $\eta$ runs
    of the embedded guarded algorithm $\AG$
    is expected to terminate successfully
    for a randomly perturbed input of size $n\in \N$
    with probability at least $p\in(0,1)$
    for every precision $L\in\N$ with $L\ge L_{\ACP}(p,n)$.
  \item
    (verifiable).
    We call $\ACP$ \emph{verifiable}
    if the following conditions are fulfilled:\\
      1. All used predicates are verifiable. \\
      2. The perturbation area ${\cal U}_{\ACP,\delta}(\y)$
        contains an open neighborhood of $\y$.\\
      3. The total number of predicate evaluations is bounded. \\
      4. The number of predicate types is bounded.
  \item\label{def-eval-suit}
    (evaluation-suitable).
    We call $\ACP$ \emph{evaluation-suitable}
    if the total number of predicate evaluations
    is upper-bounded by a function $\NE:\N\to\N$
    in dependence on the input size $n$.
  \item\label{def-pred-suit}
    (predicate-suitable).
    We call $\ACP$ \emph{predicate-suitable}
    if the number of different predicates
    is upper-bounded by a function $\NP:\N\to\N$
    in dependence on the input size $n$.
  \item\label{def-algo-pert-suit}
    (perturbation-suitable).
    Let ${\cal U}_{\ACP,\delta}(\y)$
    be the perturbation area of $\ACP$
    around $\y$;
    we assume that ${\cal U}_{\ACP,\delta}(\y)$
    is scalable with parameter $\delta$
    and that it has a fixed shape,
    e.g., cube, box, sphere, ellipsoid, etc.
    We call $\ACP$ \emph{perturbation-suitable}
    if there is a bounding function $V\!:\R_{>0}^k\to\R_{>0}$
    with the property
    that there is an open axis-parallel box $U_{\ACP,\delta}(\y)$
    with volume at least $V(\delta)$
    and $U_{\ACP,\delta}(\y) \subset {\cal U}_{\ACP,\delta}(\y)$.
  \item
    (analyzable).
    We call $\ACP$ \emph{analyzable}
    if the following conditions are fulfilled:\\
    1. All used predicates are analyzable. \\
    2. $\ACP$ is evaluation-suitable,
       predicate-suitable and perturbation-suitable.
  \end{itemize}
\end{definition}
\begin{remark}\label{rem-algo-prop}
  We add some remarks on the definitions above.

  1. The applicability of an algorithm has a strong meaning:
    For every arbitrarily large success probability $p\in(0,1)$ and
    for every arbitrarily large input size $n\in\N$
    there is still a \emph{finite} precision
    that fulfills the requirements.
    As a matter of fact,
    a controlled perturbation algorithm reaches this precision
    after \emph{finite} many steps.
    Because in addition the success probability
    is monotonically growing during the execution of $\ACP$,
    we conclude:
    If the algorithm $\ACP$ is applicable,
    its execution is guaranteed to terminate.

  2. In the definition of applicability,
    we define the precision function $L_{\ACP}(p,n)$
    as a function in the desired success probability $p$
    and the input size $n$.
    Naturally,
    the bound also depends on other quantities like
    the perturbation parameter $\delta$,
    an upper bound on the absolute input values
    or the maximum rounding-error.
    However,
    the latter quantities have some influence in the determination
    of the bounding functions in the analysis of functions.
    Here they occur as parameters in formula $L_{\ACP}$
    and are not mentioned as arguments.

  3. We remark on the perturbation-suitability that
    we allow any shape of the perturbation area ${\cal U}_{\ACP,\delta}$
    in practice
    which fulfills the condition in the definition.
    As opposed to that
    we have assumed that the perturbation area $U_{f,\delta}$
    in the analysis of functions is an axis-parallel box.
    This looks contradictorily and needs further explanation.

    As a matter of fact,
    there is just one perturbation $y\in {\cal U}_{\ACP,\delta}(\y)\RGLK$
    of the input
    before we try to evaluate the whole sequence of predicates.
    We assume that the random perturbation
    is chosen from a discrete uniform distribution
    in a subset of the $n$-dimensional space.
    Opposed to that,
    function $f_i$ has just $k_i \ll n$ arguments $x=(x_1,\ldots,x_{k_i})$.
    Mathematically speaking,
    we determine the input $x$ of $f_i$ by an orthogonal projection of
    $y$ onto a $k_i$-dimensional plane.
    Now we make the following important observation:
    \emph{If we examine the orthogonal projection onto a
    $k_i$-dimensional plane,
    the projected points do not occur with the same probability in general.}
    We refer to Figure~\ref{fig-distri-square}
    and Figure~\ref{fig-distri-sphere}.
    Despite of this observation,
    we prove in Section~\ref{sec-meth-distri-prob}
    that there is an implementation of $\ACP$
    that we can analyze---presumed that we know the bounding function $V$
    that is mentioned in the definition above.
\begin{figure}[t]\centering
  \includegraphics[width=.95\columnwidth]{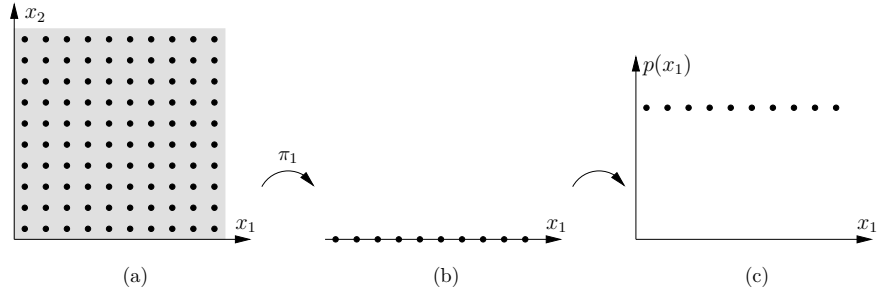}
  \caption{(a) The original perturbation area ${\cal U}_{\ACP,\delta}$
    is an axis-parallel box.
    (b) Its projection is uniformly distributed.
    (c) The points in the projection are chosen with the same probability.}
  \label{fig-distri-square}
\end{figure}
\begin{figure}[t]\centering
  \includegraphics[width=.95\columnwidth]{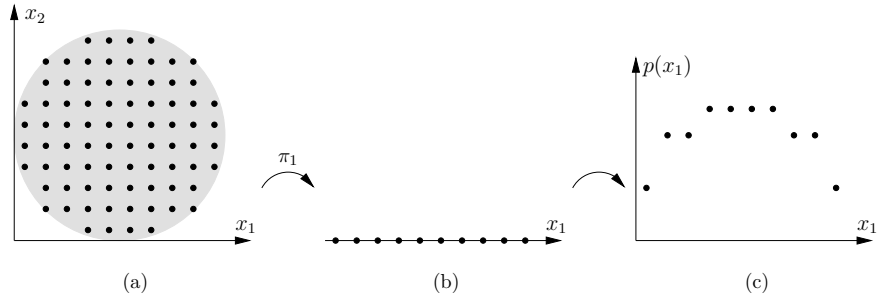}
  \caption{(a) The original perturbation area ${\cal U}_{\ACP,\delta}$
    is a sphere.
    (b) Its projection is uniformly distributed.
    (c) The points in the projection are not chosen with the same probability.}
  \label{fig-distri-sphere}
\end{figure}

  4. We remark on the predicate-suitability
    that the number $\NP\in\N$ of different predicates is usually fixed
    for a geometric algorithm.
    Anyway,
    since we will see that the analysis can also be performed
    for a function $\NP(n)$
    we keep the presentation as general as possible.
  \hfill$\bigcirc$
\end{remark}
  Next we explain
  how we determine the three bounding functions
  which are associated with the three necessary algorithm properties.
  We refer to Figure~\ref{fig-illu-ana-algo}.
  If the number $\NP$ of used predicates is fixed, we just count them;
  otherwise we perform a complexity analysis to determine the bounding function
  $\NP(n)$.
  We usually determine the bounding function $\NE(n)$
  on the number of predicate evaluations
  with a complexity analysis, too.
  The bound $\eta$ results from a geometric consideration:
  We only need to determine the real volume of the solid perturbation area.
  If the perturbation area has an ordinary shape,
  its computation is straight forward.
  We consider an example of this.
\begin{example}
  Let the input of $\ACP$ be $m$ points in the plain,
  that means, $n=2m$.
  In addition
  let the perturbation area for each point be a disc of radius $\delta$.
  Then the axis-parallel square of maximum volume inside of such a disc
  has edge length $\delta\sqrt{2}$.
  We obtain:
  \begin{eqnarray*}
    \eta &:=&
     \left\lceil
      \frac{\mu({\cal U}_{\delta})}{\mu(U_\delta)}
     \right\rceil \\
    &=&
     \left\lceil
      \frac{\mu(\text{$m$ discs of radius $\delta$})}
        {\mu(\text{$m$ cubes of edge length $\delta\sqrt{2}$})}
     \right\rceil \\
    &=&
     \left\lceil
      \frac{m\cdot \pi \delta^2}{m\cdot2\delta^2}
     \right\rceil \\
    &=&
     2
  \end{eqnarray*}
  We observe that the bound $\eta$ does not depend on $m$ (or $n$).
  \hfill$\bigcirc$
\end{example}
\noindent
  Now we state and prove the implications of algorithm properties.
\begin{lemma}\label{lem-algo-ana-is-veri}
Let algorithm $\ACP$ be analyzable. Then $\ACP$ is verifiable.
\end{lemma}
\begin{proof}
This is trivially true.
  \qed
\end{proof}
\begin{lemma}\label{lem-algo-veri-is-app}
Let algorithm $\ACP$ be verifiable. Then $\ACP$ is applicable.
\end{lemma}
\begin{proof}
\newcommand{\LPN}{{{\cal L}_{p,n}}}
  To show that $\ACP$ is applicable,
  we prove the following existence.
  \emph{There is $\eta\in\N$
  such that
  for every $p\in(0,1)$ and every $n\in\N$
  there is a precision $\LPN$
  with the property:
  For a randomly perturbed input of size $n$,
  at least one from $\eta$ runs of $\AG$
  is expected to terminate successfully
  with probability at least $p$
  for every precision $L\in\N$ with $L\ge \LPN$.}
  Then the function $L_{\ACP}(p,n) := \LPN$
  has the desired property
  which proves the claim.

  At first we show that there is an appropriate $\eta\in\N$.
  Because $\ACP$ is verifiable,
  the perturbation area
  ${\cal U}_{\ACP,\delta}(\y)$
  contains an \emph{open} set around $\y$.
  Therefore there is an open axis-parallel box around $\y$ with
  $U_\delta(\y)\subset{\cal U}_{\ACP,\delta}(\y)$.
  Then there is also a natural number
  \begin{eqnarray*}
    \eta &:=&
      \left\lceil
        \frac{\mu\left({\cal U}_{\ACP,\delta}(\y)\right)}
        {\mu\left(U_\delta(\y)\right)}
      \right\rceil.
  \end{eqnarray*}
  That means,
     if we randomly choose $\eta$
     points from a uniformly distributed grid in
     ${\cal U}_{\ACP,\delta}(\y)\RGLK$,
     we may expect that at least one point lies also inside of $U_\delta(\y)$.

     Let $p\in(0,1)$ and let $n\in\N$.
     In addition
     let $y\in U_\delta(\y)\RGLK$
     be randomly chosen.
     Since $\ACP$ is verifiable,
     there is an upper-bound $\NE\in\N$
     on the total number of predicate evaluations.
     Therefore we can distribute the total failure probability $(1-p)$
     among the $\NE$ predicate evaluations.
     Hence there is a probability
     \begin{eqnarray*}
     \varrho &:=&
     \frac{1-p}{\NE}.
     \end{eqnarray*}
     Obviously
     $\AG(y)$ is successful with probability $p$
     if every predicate evaluation fails with probability at most $\varrho$.

     Let $\NP\in\N$ be the number of different predicates in $\AG$
     which are decided by the functions $f_1,\ldots,f_\NP$.
     Because $\ACP$ is verifiable,
     all used predicates are verifiable and thus applicable.
     Then Definition~\ref{def-func-app}
     implies the existence of
     precision functions $L_{f_1},\ldots,L_{f_\NP}$.
     Therefore there is a precision
     \begin{eqnarray*}
{\cal L}_{p,n}
&:=&
\max_{1\le i \le \NP} \;
L_{f_i} \left(1 - \varrho\right)
     \end{eqnarray*}
     which has the desired property
     because of Definition~\ref{def-func-app}.
     This finishes the proof.
  \qed
\end{proof}
  As a consequence of Lemma~\ref{lem-algo-ana-is-veri}
  and Lemma~\ref{lem-algo-veri-is-app}
  the controlled perturbation implementation $\ACP$
  terminates with certainty and
  yields the correct result for the perturbed input
  if $\ACP$ is analyzable.

\subsection{Overview: Algorithm Properties}
\label{sec-algo-prop}

  An overview of the defined algorithm properties
  is shown in Figure~\ref{fig-algo-prop}.
  The meanings are:
  A controlled perturbation algorithm $\ACP$ is guaranteed to terminate
  if $\ACP$ is \emph{applicable}
  (see Remark~\ref{rem-algo-prop}.1).
  If $\ACP$ is \emph{verifiable},
  we can prove that $\ACP$ terminates---even
  if we are not able to analyze its performance.
  And finally,
  we can give a quantitative analysis of the performance of $\ACP$
  if $\ACP$ is analyzable.

  The implications are:
  An evaluation-, perturbation- and predicate suitable algorithm
  that uses solely analyzable predicates is analyzable
  (see Definition~\ref{lem-algo-ana-is-veri}).
  An analyzable algorithm is also verifiable
  (see Lemma~\ref{lem-algo-ana-is-veri}).
  And a verifiable algorithm is also applicable
  (see Lemma~\ref{lem-algo-veri-is-app}).
\begin{figure}[h]\centering
  \includegraphics[width=.95\columnwidth]{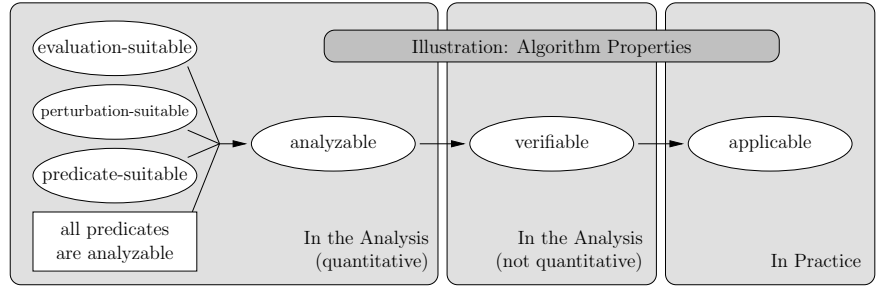}
  \caption{The illustration summarizes the implications of
    the various algorithm properties that we have defined in this section.}
  \label{fig-algo-prop}
\end{figure}

\subsection{The Method of Distributed Probability}
\label{sec-meth-distri-prob}

  Here we state the main theorem of this section.
  The proof contains the method of distributed probability
  which is used to analyze complete algorithms.
  Figure~\ref{fig-illu-ana-algo}
  shows the component and its interface.
\begin{theorem}[distributed probability]
  Let $\ACP$ be analyzable.
  Then there is a general method
  to determine a precision function
  $L_{\ACP}:(0,1)\times\N\to\N$ and
  $K_{\ACP}:(0,1)\times\N\to\N$
  and $\eta\in\N$ with the property:
  At least one from $\eta$ runs
  of the embedded guarded algorithm $\AG$
  is expected to terminate successfully for a randomly perturbed input of size $n$
  with probability at least $p\in(0,1)$
  for every arithmetic $\FLK$
  where $L\ge L_{\ACP}(p,n)$ and $K\ge K_{\ACP}(p,n)$.
\end{theorem}
\begin{proof}
  We prove the claim in three steps:
  At first we derive $\eta\in\N$
  from the shape of the region of uncertainty.
  Then we determine a bound on the
  failure probability of each predicate evaluation.
  And finally we analyze each predicate type
  to determine the worst-case precision.
  An overview of the steps is given in Table~\ref{tab-steps-dp}.
  \begin{table}[b]
  \centerline{\fbox{
  \begin{minipage}{0.9\columnwidth}\medskip
    Step 1: determine ``in axis-parallel box'' probability (define $\eta$)\\
    Step 2: determine ``per evaluation'' probability (define $\rho$)\\
    Step 3: compose precision function
      (define $L_{\ACP}$ and $K_{\ACP}$)\medskip
  \end{minipage}}}
    \caption{Instructions for performing the method of distributed probability.}
    \label{tab-steps-dp}
  \end{table}

Step~1 (define $\eta$).
  We define $\eta$ as the ratio
  \begin{eqnarray*}
    \eta
    &=&
      \left\lceil
        \frac{V(\delta)}{\mu(U_\delta(\y))}
      \right\rceil.
  \end{eqnarray*}
  That means, 
  if we randomly choose $\eta$
  points from a uniformly distributed grid in
  ${\cal U}_{\ACP,\delta}(\y)\RGLK$,
  we may expect that at least one point lies also inside of
  $U_{\ACP,\delta}(\y)$.

Step~2 (define $\rho$).
  Let $p\in(0,1)$
  be the desired success probability of the guarded algorithm $\AG$.
  Then $(1-p)$ is the failure probability of $\AG$.
  There are at most $\NE(n)$
  predicate evaluations for an input of size $n$.
  That means, the guarded algorithm succeeds
  if and only if we evaluate all predicates successfully in a row
  for \emph{the same} perturbed input.
  We observe that the evaluations do not have to be independent.
  Therefore we define the failure probability of each predicate evaluation
  as the function
  \begin{eqnarray*}
    \varrho(p,n) &:=& \frac{1-p}{\NE(n)}
  \end{eqnarray*}
  in dependence on $p$ and $n$.

Step~3 (define $L_{\ACP}$ and $K_{\ACP}$).
  There are at most $\NP(n)$ different predicates.
  Let $f_1,\ldots,f_{\NP(n)}$
  be the functions that realize these predicates.
  Since all functions are analyzable,
  we determine their precision function $L_{f_i}$
  with the presented methods of our analysis tool box.
  Then we define the precision function for the algorithm as
  \begin{eqnarray*}
    L_{\ACP} (p,n)
      &:=& \max_{1\le i\le \NP(n)}
        \; L_{f_i} (1-\varrho(p,n)) \\
      &=& \max_{1\le i\le \NP(n)}
        \; L_{f_i}
        \left(1-\frac{1-p}{\NE(n)}\right).
  \end{eqnarray*}
  Analogically we define
  \begin{eqnarray*}
    K_{\ACP} (p,n)
      &=& \max_{1\le i\le \NP(n)}
        \; K_{f_i}
        \left(1-\frac{1-p}{\NE(n)}\right).
  \end{eqnarray*}
  Then every arithmetic $\FLK$
  with $L\ge L_{\ACP}(p,n)$ and $K\ge K_{\ACP}(p,n)$
  has the desired property by construction.
  \qed
\end{proof}

\section{General Controlled Perturbation Implementations}
\label{sec-gen-cp-imple}

  We present a general way
  to implement controlled perturbation algorithms $\ACP$
  to which we can apply our analysis tool box.
  The algorithm template is illustrated as Algorithm~\ref{algo-cp}.
  It is important to see that
  all statements which are necessary for the controlled perturbation management
  are simply wrapped around the function call of $\AG$.
  \begin{algorithm}
    \caption{: $\ACP(\AG, \y, \UU_\delta, \psi, \eta)$}
    \label{algo-cp}
    \begin{algorithmic}
      \STATE \emph{/* initialization */}
      \STATE $L \leftarrow$ precision of built-in floating-point arithmetic
      \STATE $K \leftarrow$ exponent bit length of built-in floating-point arithmetic
      \STATE $\E \leftarrow$ determine upper bound $2^\E$ on $|\y_i|+\delta$
     \medskip
     \REPEAT
      \STATE \emph{/* run guarded algorithm */}
      \FOR{$i=1$ \TO $\eta$}
        \STATE $y \leftarrow$ random point in $\UUU_\delta(\y)\RGLKE$
        \STATE $\omega \leftarrow \AG(y,\FLK)$
        \IF{$\AG$ succeeded}
	  \STATE leave the for-loop
	\ENDIF
      \ENDFOR
     \medskip
      \STATE \emph{/* adjust parameters */}
      \IF{$\AG$ failed}
	\IF {floating point overflow error occurred}
	  \STATE \emph{/* guard failed because of range error */}
          \STATE $K \leftarrow K + \psi_K$
	\ELSE
	  \STATE \emph{/* guard failed because of insufficient precision */}
          \STATE $L \leftarrow \lceil\psi_L \cdot L\rceil$
	\ENDIF
      \ENDIF
     \UNTIL{$\AG$ succeeded}
     \medskip
      \STATE \emph{/* return perturbed input $y$ and result $\omega$ */}
      \RETURN $(y,\omega)$
    \end{algorithmic}
  \end{algorithm}

  Remember that the original perturbation area is
  $\UUU_\delta(\y)\RG$.
  The implementation of a uniform perturbation
  seems to be a non-obvious task for most shapes.
  Therefore we propose axis-parallel perturbation areas
  in applications.
  (For example, we can replace spherical perturbation areas with cubes
  that are contained in them.)
  For axis-parallel areas there is the special bonus
  that the perturbation is composed of random integral numbers
  as we have explained in Remark~\ref{perturb-implem-inline}.

  \label{intext-def-psi}
  An argument of the controlled perturbation implementation
  is the tuple $\psi=(\psi_L, \psi_K)\in\R\times\N$ of constants
  which are used for the augmentation of $L$ and $K$.
  The real constant $\psi_L>1$
  is used for a multiplicative augmentation of $L$,
  and the natural number $\psi_K$
  is used for an additive augmentation of $K$.

  \label{intext-algo-ana-vari-delta}
  We remark that there is a variant of Algorithm~\ref{algo-cp}
  that also allows the increase of perturbation parameter $\delta$.
  Beginning with $\delta=\DMIN\in\R_{>0}^k$,
  we augment the perturbation parameter $\delta$
  by a real factor $\psi_\delta>1$
  each time we repeat the for-loop.
  When we leave the for-loop, we reset $\delta$ to $\DMIN$.
  We observe that this strategy implies an upper-bound on
  the perturbation parameter by $\DMAX:=\DMIN \cdot \psi_{\delta}^{\eta-1}$.
  This is the bound that we use in the analysis.
  To keep the presentation clear,
  we do not express variable perturbation parameters explicitly in the code.

  A variable precision floating-point arithmetic
  is necessary for an implementation of $\ACP$.
  When we increase the precision
  in order to evaluate complex expressions successfully,
  the evaluation of simple expressions
  start suffering from the wasteful bits.
  Therefore we suggest floating-point filters
  as they are used in interval arithmetic.
  That means,
  we use a multi-precision arithmetic that refines the precision
  on demand up to the given $L$.
  If it is necessary to exceed $L$, $\AG$ fails.
  In the analysis we use
  this threshold on the precision.

\section{Perturbation Policy}
\label{sec-pertub-policy}

  The meaning of perturbation is introduced
  in Section~\ref{sec-basic-quantities} and
  its implementation is explained in Remark~\ref{perturb-implem-inline}
  on Page~\pageref{perturb-implem-inline}.
  So far we have considered the original input
  to be the point $\y\in\R^n$
  which is the concatenation of \emph{all} coordinates of
  \emph{all} input points
  for the geometric algorithm $\ACP$.
  In contrast to that,
  we now care for the geometric interpretation of the input
  and consider it as a sequence of geometric objects
  $\GO_1, \ldots, \GO_m$.
  Then a perturbation of the input is the sequence of perturbed objects.
  In this section we define two different perturbation policies:
  The \emph{pointwise} perturbation
  in Section~\ref{sec-pert-indi}
  and the \emph{object-preserving} perturbation
  in Section~\ref{sec-pert-robust}.
  The latter has the property
  that the topology of the input object is preserved.
  \emph{This is the first presentation that
  integrates object-preserving perturbations
  in the controlled-perturbation theory.}

\subsection{Pointwise Perturbation}
\label{sec-pert-indi}

  For pointwise perturbations
  we assume that the geometric object is
  given by a sequence of points.
  A circle in the plain, for example, is given by three points.
  Another example is the polygon
  in Figure~\ref{fig-pert-point}(a)
  which is represented by the sequence of four vertices
  $abcd$.
  \begin{figure}[h]\centering
    \includegraphics[width=.95\columnwidth]{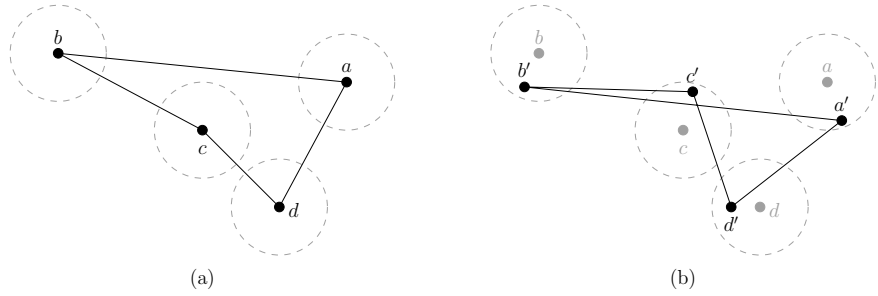}
    \caption{Example of a \emph{pointwise perturbation} in the plane:
      (a) original input and
      (b) perturbed input.}
    \label{fig-pert-point}
  \end{figure}
  The \emph{pointwise perturbation} of a geometric object
  is the sequence of individually perturbed points
  of its description,
  i.e., randomly chosen points of their neighborhoods.
  Figure~\ref{fig-pert-point}(b)
  shows a pointwise perturbed polygon $a'b'c'd'$ for our example.
  Because the perturbations are independent of each other,
  this policy is quite easy to implement.
  But we observe that pointwise perturbations do not preserve
  the structure of the input object in general:
  The original polygon $abcd$
  is simple whereas the perturbed polygon $a'b'c'd'$ in our example is not.
  And the orientation of a circle
  that is defined by three perturbed points
  may differ from the orientation of the circle
  that is defined by the original points.
  Be aware
  that our analysis is particularly designed for pointwise perturbations.
  We suggest to apply this perturbation policy to inputs
  that are disturbed by nature, e.g., scanned data.

\subsection{Object-preserving Perturbation}
\label{sec-pert-robust}

  For object-preserving perturbations
  we assume that the geometric object is
  given by an anchor point and a sequence of fixed measurements.\footnote{%
    The measurements may be given explicitly or implicitly.
    Both is fine.}
  A circle in the plain, for example,
  is given by a center (anchor point)
  and a radius (fixed measurement).
  Another example is the polygon $abcd$
  in Figure~\ref{fig-pert-object}(a)
  which is given by an anchor point, say $a$,
  and implicitly by the sequence of vectors (the measurements)
  pointing from $a$ to $b$, from $a$ to $c$, and from $a$ to $d$.
  \begin{figure}[h]\centering
    \includegraphics[width=.95\columnwidth]{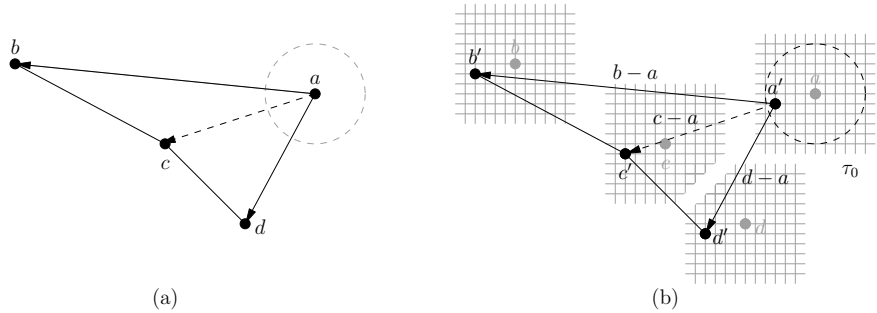}
    \caption{Example of an \emph{object-preserving perturbation}
      in the plane:
      (a) original input and
      (b) perturbed input.}
    \label{fig-pert-object}
  \end{figure}
  The \emph{object-preserving perturbation} of a geometric object
  is a pointwise perturbation of its anchor point
  while maintaining all given measurements.
  Figure~\ref{fig-pert-object}(b)
  shows polygon $a'b'c'd'$
  that results from an object-preserving perturbation.
  There we have $b':=a'+{b}-{a}$, etc.
  We observe that this perturbation is actually a translation
  of the object and hence preserves
  the structure of the input object in any respect:
  its orientation, measurements and angles.
  The object-preserving perturbation of a circle,
  for example,
  changes its location but not its radius.

  The input must provide further information
  to support object-preserving perturbations.
  For the explicit representation,
  this policy requires a \emph{labeling} of input values
  as \emph{anchor points} (perturbable)
  or \emph{measurements} (constant).
  For the implicit representation,
  the policy requires the subdivision of the input
  into single objects;
  then we make one of these points the anchor point
  and derive the measurements for the remaining points.
  To allow the object-preserving perturbation,
  the implementation must offer the labeling of values or
  the distinction of input objects.

  In this context
  it is pleasant to observe that our perturbation area $\UUU_\delta(A)\RG$
  supports object-preservation because it is composed of a regular grid.
  \emph{If the original object is represented without rounding error,
  the perturbed object is represented exactly as well.}
  Of course,
  {we can always apply object-preserving perturbations
  to finite-precision input objects.}\footnote{%
    This is true
    because we can derive a sufficient grid unit
    from the given fixed-precision input.}
  We suggest to apply this policy to inputs
  that result from computer-aided design (CAD):
  By design, the measurements are often multiples of a certain unit
  which can be used as an upper bound on the grid unit.

  How can we analyze object-preserving perturbations?
  We consider the analysis of function $f$ that realizes a predicate.
  For pointwise perturbations
  we demand in Section~\ref{sec-basic-quantities}
  that $f$ only depends on input values.
  For object-preserving perturbations
  we only allow dependencies on anchor points:
  Every other point in the description of the object
  must be replaced in the formula by an expression that
  depends on the anchor point of the affected object.
  Be aware that these expressions can be resolved error-free
  due to the fixed-point grid $\G$.
  Then the new formula,
  depends only on anchor points (variables) and measurements (constants).
  The dependency of the function on the variables is analyzed as before.
  Finally we remark that we do not recommend perturbation policies
  that are based on scaling, stretching, sheering or rotation
  since the perturbed input cannot be represented error-free in general.

\clearpage
\section{Appendix: List of Identifiers}
\label{sec-append-identifiers}

\newcommand{\Index}[3]{%
  \noindent%
  \begin{minipage}[t]{.15\columnwidth}\begin{flushleft}{#1}\end{flushleft}\end{minipage}\hspace*{\fill}%
  \begin{minipage}[t]{.71\columnwidth}{#2}\end{minipage}\hspace*{\fill}%
  \begin{minipage}[t]{.10\columnwidth}\begin{flushright}{#3}\end{flushright}\end{minipage}\vspace{1.4ex}}%
\newcommand{\Indexsec}[1]%
  {\subsection*{{#1}\hspace*{\fill}Page}}%

  Page numbers refer to definitions of the identifiers.
  References to preliminary definitions are parenthesized.\\

\Indexsec{Algorithms}

  \Index{$\Alg$}{%
    the given geometric algorithm $\Alg(\y)$.}{-}

  \Index{$\AG$}{%
    the guarded version $\AG(y,\FLK)$ of algorithm $\Alg$,
    i.e., all predicate evaluations are guarded.}{\pageref{def-guarded-algo}}

  \Index{$\ACP$}{%
    the controlled perturbation version $\ACP(\AG,\y,\delta,\psi)$
    of algorithm $\Alg$.
    The implementation of $\ACP$ makes usage of $\AG$.}{\pageref{algo-cp}}

\Indexsec{Sets and Number Systems}

  \Index{$\C$}{the set of complex numbers.}{-}

  \Index{$\FLK$}{1. the set of floating point numbers
  with radix 2
  whose precision has up to $L$ digits and
  whose exponent has up to $K$ digits.\\
  2. the floating point arithmetic that is induced this way.}{\pageref{def-fp-numbers}}

  \Index{$\GLKE$}{the set of grid points.
  They are a certain subset of the floating point numbers $\FLK$
  within the interval $[-2^\E,2^\E]$.}{\pageref{def-grid-points}}

  \Index{$\N$; $\N_0$}{the set of natural numbers;
  set of natural numbers including zero.}{-}

  \Index{$\Q$}{the set of rational numbers.}{-}

  \Index{$\R$; $\R_{>0}$; $\R_{\neq 0}$}{the set of real numbers;
  set of positive real numbers;
  set of real numbers excluding zero.}{-}

  \Index{$\Z$}{the set of integer numbers.}{-}

  \Index{$X\RFLK$}{the restriction of a set $X$ to points in $\FLK$.}{\pageref{def-fp-numbers}}

  \Index{$X\RGLKE$}{the restriction of a set $X$ to points in $\GLKE$.}{\pageref{def-grid-points}}

\Indexsec{Identifiers of the Analysis}

  \Index{$A$}{the set of valid projected arguments $\x$ for $f$.}{\pageref{def-A-inline}}

  \Index{$B_E(L)$}
    {a floating point error bound on the arithmetic expression $E$.}
    {\pageref{def-fperrorbound-inline}}

  \Index{$C_f(\cdot)$}{the critical set of $f$.}
    {(\pageref{def-critical-set}), \pageref{def-crit-set-final-inline}}

  \Index{$\GG_f$}{a guard for $f$ on the domain $X$.}{\pageref{def-guard}}

  \Index{$K$}{the bit length of the exponent (see $\FLK$).}{\pageref{def-K-inline}}

  \Index{$K_f(p)$}{a lower bound on the bit length of the exponent.}
  {\pageref{for-Kf-inline}}

  \Index{$L$}{the bit length of the precision (see $\FLK$).}{\pageref{def-L-inline}}

  \Index{$L_{\ACP}(p,n)$}{the precision function of $\ACP$.}{\pageref{def-lacp}}

  \Index{$L_f(p)$}{the precision function of $f$.}{\pageref{def-lf-final}}

  \Index{$\LG$}{a bound on the precision; caused by the grid unit condition.}{(\pageref{for-def-lgrid}), \pageref{for-lgrid-in-proof}}

  \Index{$\LS$}{a bound on the precision; caused by the region- and safety-condition.}{\pageref{for-def-lsafe}}

  \Index{$\NE(n)$}{an upper-bound on the number of predicate evaluations.}{\pageref{def-eval-suit}}

  \Index{$\NP(n)$}{an upper-bound on the number of different predicates.}{\pageref{def-pred-suit}}

  \Index{$R_{f,\gamma}(\cdot)$}{the region of uncertainty of $f$.}
    {\pageref{def-region-uncertainty}}

  \Index{$R_{f,\AUG(\gamma)}(\cdot)$}{the augmented region of uncertainty of $f$.}
    {\pageref{inline-def-aug-rou}}

  \Index{$\Sinf(L)$}{the lower fp-safety bound.}
    {\pageref{def-fpsafetybound}, (\pageref{for-Sinf-final-inline})}

  \Index{$\Ssup(K)$}{the upper fp-safety bound.}
    {\pageref{for-Ssup-final-inline}}

  \Index{$U_{f,\delta}(\cdot)$}{the perturbation area of $f$;
    its shape is an axis-parallel box.}{\pageref{def-U-inline}}

  \Index{$U_{\ACP,\delta}(\cdot)$}{the perturbation area of $\ACP$;
    its shape is an axis-parallel box.}{\pageref{def-algo-pert-suit}}

  \Index{${\cal U}_{\ACP,\delta}(\cdot)$}{the perturbation area of $\ACP$;
    it may have any shape.}{\pageref{def-algo-pert-suit}}

  \Index{$\E$}{the input value parameter
    (see Formula~(\ref{for-e-min})).}{\pageref{for-e-min}}

  \Index{$f$}{the real-valued function $f:\U_\delta(A)\to\R$
    under consideration.
    We assume that the sign of $f$ decides a geometric predicate.}
    {\pageref{def-f-inline}}

  \Index{$k$}{the arity of $f$.}
    {\pageref{def-f-inline}}

  \Index{$n$}{the size of input $\y$.}
    {\pageref{def-n-inline}}

  \Index{$p_f(L,K)$}{the probability function of $f$.}
    {(\pageref{def-prob-func-inline}), \pageref{for-pfLK-inline}}

  \Index{$\pgrid(L)$}
    {a bound on the probability; caused by the grid unit condition.}
    {\pageref{for-pgrid-inline}}

  \Index{$\pinf(L)$}
    {a bound on the probability; caused by the region- and inf-safety-condition.}
    {\pageref{for-pinf-inline}}

  \Index{$\psup(K)$}
    {a bound on the probability; caused by the sup-safety-condition.}
    {\pageref{for-supsafe-inline}}

  \Index{$\PR(f\RG)$}{the least probability that
    a guarded evaluation of $f$ is successful for inputs in $\G$
    under the arithmetic $\F$.}{\pageref{for-prob-rest-G}}

  \Index{$\frac{1}{t}$}{the augmentation factor for the region of uncertainty.}
    {\pageref{def-t-inline}}

  \Index{$\x$}{the arguments of $f$; projection of $\y$.}
    {\pageref{def-xbar-inline}}

  \Index{$x$}{the perturbed arguments of $f$; projection of $y$.}
    {\pageref{def-x-inline}}

  \Index{$\y$}{the original input to the algorithm.}
    {\pageref{def-ybar-inline}}

  \Index{$y$}{the perturbed input $y\in U_\delta(\y)$.}
    {\pageref{def-y-inline}}

  \Index{$\delta$}{the perturbation parameter
  which bounds the maximum amount of perturbation componentwise.}{\pageref{def-pert-para-inline}}

  \Index{$\gamma$}{the tuple of
  componentwise distances to the critical set.}{\pageref{def-gamma-inline}}

  \Index{$\Gamma$}{the set of valid augmented $\gamma$.}{\pageref{def-Gamma-inline}}

  \Index{$\GAB$}{like $\Gamma$; the set is an axis parallel box.}{\pageref{def-GAB-inline}}

  \Index{$\GAL$}{like $\Gamma$; the set is a line.}{\pageref{def-GAB-inline}}

  \Index{$\nu_f(\gamma)$}{an upper-bound
    on the volume of $R_{f,\gamma}$.}
    {\pageref{def-nu-inline}}

  \Index{$\tau$}{the grid unit.}{\pageref{def-tau-inline}}

  \Index{$\varphiinf(\gamma)$}
    {a lower-bound
    on the absolute value of $f$ outside of $R_{f,\gamma}$.}
    {\pageref{def-varphi-inline}, (\pageref{def-phi-inf-final-inline})}

  \Index{$\varphisup(\gamma)$}
    {an upper-bound
    on the absolute value of $f$ outside of $R_{f,\gamma}$.}
    {\pageref{def-phi-sup-final-inline}}

  \Index{$\chi_f(\gamma)$}{a lower-bound on the complement
    of $\nu_f$ within the perturbation area.}{\pageref{def-chi-region-suit}}

  \Index{$\psi$}{the tuple $\psi=(\psi_L, \psi_K)\in\R\times\N$ is used
    for the augmentation of $L$ and $K$.}{\pageref{intext-def-psi}}

\Indexsec{Miscellaneous}

  \Index{$\mu(\cdot)$}{the Lebesgue measure.}{\pageref{def-mu-inline}}

  \Index{$\pi(\cdot)$}
    {the projection of points and sets, e.g., $\PII$, $\PIL$, $\PIG$, $\PIN$.}
    {\pageref{def-projection-pi}}

  \Index{$\Lex$}{the reverse lexicographic order.}{\pageref{def-lex-inline}}

  \Index{$\Lex_\sigma$}{the reverse lexicographic order
    after the permutation of the operands.}{\pageref{def-lexafter-inline}}

\end{document}